\documentclass[fleqn,11pt,letterpaper]{article}

\usepackage{amsfonts}
\usepackage{amssymb}
\usepackage{amsmath}
\usepackage{amsthm}
\usepackage{subcaption}
\usepackage[utf8]{inputenc}
\usepackage[english]{babel}
\usepackage{algpseudocode}
\usepackage{algorithm}

\usepackage{rotating,tabularx} 
\usepackage{enumitem}
\usepackage{multirow}
\usepackage{float}
\usepackage{graphicx}
\usepackage{epstopdf}
\usepackage{color}
\usepackage[margin=2.57cm]{geometry}
\definecolor{darkred}{rgb}{0.5,0.2,0.2}
\usepackage[pdftitle={Transfer Learning (Il)liquidity - Conti, Morelli},pdfauthor={Giacomo Morelli},colorlinks=TRUE,allcolors=darkred]{hyperref}
\usepackage{booktabs}
\usepackage{pdflscape}
\usepackage{xcolor}

\usepackage[
  backend=biber,
  style=authoryear,
  maxcitenames=99,
  mincitenames=99,
  maxbibnames=99,
  minbibnames=99
]{biblatex}

\usepackage{algpseudocode}
\usepackage{algorithm}
\usepackage{comment}
\usepackage{url}
\usepackage{systeme}

\usepackage{nccmath}
\usepackage{animate}

\theoremstyle{mytheorem}

\theoremstyle{myremark}

\nonstopmode

\usepackage{tabularx}

\usepackage{float}
\floatstyle{plaintop}
\restylefloat{table}

\usepackage{adjustbox}
\usepackage{booktabs}

\usepackage{bbm}

\restylefloat{table}

\theoremstyle{definition}

\theoremstyle{remark}

\def\bi{\begin{itemize}}
\def\ei{\end{itemize}}

\allowdisplaybreaks
\graphicspath{{Figures/}}
\numberwithin{equation}{section}

\definecolor{applegreen}{rgb}{0.55, 0.71, 0.0}
\definecolor{ballblue}{rgb}{0.13, 0.67, 0.8}
\definecolor{mauvetaupe}{rgb}{0.57, 0.37, 0.43}
\definecolor{uclagold}{rgb}{1.0, 0.7, 0.0}
\definecolor{amethyst}{rgb}{0.6, 0.4, 0.8}
\definecolor{alizarin}{rgb}{0.82, 0.1, 0.26}
\definecolor{tiffanyblue}{rgb}{0.04, 0.73, 0.71}
\definecolor{royalpurple}{rgb}{0.47, 0.32, 0.66}
\definecolor{mistyrose}{rgb}{1.0, 0.89, 0.88}
\definecolor{nadeshikopink}{rgb}{0.96, 0.68, 0.78}
\definecolor{pink}{rgb}{1.0, 0.75, 0.8}
\definecolor{unmellowyellow}{rgb}{1.0, 1.0, 0.4}
\definecolor{uscgold}{rgb}{1.0, 0.8, 0.0}
\definecolor{tearose(rose)}{rgb}{0.96, 0.76, 0.76}
\definecolor{bittersweet}{rgb}{1.0, 0.44, 0.37}
\definecolor{apricot}{rgb}{0.98, 0.81, 0.69}
\definecolor{bisque}{rgb}{1.0, 0.89, 0.77}
\definecolor{cosmiclatte}{rgb}{1.0, 0.97, 0.91}
\definecolor{blue(ncs)}{rgb}{0.0, 0.53, 0.74}
\definecolor{darkred}{rgb}{0.55, 0.0, 0.0}
\definecolor{coralred}{rgb}{1.0, 0.25, 0.25}
\definecolor{darkgreen}{rgb}{0.0, 0.2, 0.13}
\usepackage{cleveref}
\usepackage{caption}
\definecolor{darkred}{RGB}{139,0,0}
\definecolor{darkgreen}{RGB}{0,128,85} 

\newtheoremstyle{colored}
  {3pt} 
  {3pt} 
  {\itshape} 
  {} 
  {\color{darkred}\bfseries} 
  {.} 
  {0.5em} 
  {} 

\theoremstyle{colored}

\makeatletter
\renewenvironment{proof}[1][\proofname]{%
  \par
  \pushQED{\qed}%
  \normalfont
  \topsep6\p@\@plus6\p@\relax
  \trivlist
  \item[\hskip\labelsep
        {\color{darkgreen}\bfseries #1\@addpunct{.}}]\ignorespaces
}{%
  \popQED\endtrivlist\@endpefalse
}
\makeatother

\newtheorem{assumption}{Assumption}[section]
\newtheorem{lemma}[assumption]{Lemma}
\newtheorem{theorem}[assumption]{Theorem}

\newtheorem{example}[assumption]{Example}

\newif\ifi

\usepackage{lineno}
\usepackage[title]{appendix}

\begin{document}
\title{Transfer Learning (Il)liquidity} 

\author{Andrea Conti\thanks{
Department of Social and Economic Sciences,
Sapienza University of Rome, 
00185 Rome, Italy 
E-mail: \url{andrea.conti@uniroma1.it}.}
\and
Giacomo Morelli\thanks{
Corresponding author:
Department of Statistical Sciences, 
Sapienza University of Rome, 
00185 Rome, Italy.
E-mail: \url{giacomo.morelli@uniroma1.it}}
}

\maketitle

\begin{center}
\textbf{Abstract}
\end{center}

\begin{flushleft}
\newcommand{\leftindent}[1]{\leftskip=#1}
\leftindent{1cm}
\newcommand{\rightindent}[1]{\rightskip=#1}
\rightindent{1cm}

\small



The estimation of the Risk Neutral Density (RND) implicit in option prices is challenging, especially in illiquid markets. We introduce the Deep Log-Sum-Exp Neural Network, an architecture that leverages Deep and Transfer learning to address RND estimation in the presence of irregular and illiquid strikes. We prove key statistical properties of the model and the consistency of the estimator. We illustrate the benefits of transfer learning to improve the estimation of the RND in severe illiquidity conditions through Monte Carlo simulations, and we test it empirically on SPX data, comparing it with popular estimation methods. Overall, our framework shows recovery of the RND in conditions of extreme illiquidity with as few as three option quotes.

\end{flushleft}

\vspace{0.15 cm}
\begin{flushleft}
\newcommand{\leftindent}[1]{\leftskip=#1}
\leftindent{0.5cm}
\newcommand{\rightindent}[1]{\rightskip=#1}
\rightindent{0.5cm}

\end{flushleft}


\newpage

\section{Introduction}

The risk-neutral density (RND)  describes the risk-neutral distribution of prices at a specified horizon and combines expectations with the market risk aversion. The correct estimation of the RND is central in asset pricing and financial risk management as it is used for pricing, hedging, construction of optimal portfolios, and market timing (\cite{kostakis2011market}). This importance is amplified when option markets are illiquid and option quotes are scarce, since the resulting market incompleteness leaves infinitely many risk-neutral densities admissible (\cite{almeida2023nonparametric}). In addition, as the expiration approaches the bid-ask spread increases (\cite{hsieh2019volatility}), reducing the liquidity and making the recovery of the RND harder when it matters most.

Classical methods to estimate the RND fall into two broad classes: (i) parametric methods assume a specific RND shape or distributional form and estimate the parameters from option quotes; (ii) nonparametric methods estimate the RND by minimizing a likelihood or loss function. Parametric approaches can be further distinguished into expansion methods (\cite{JARROW1982347}), generalized distribution and moment methods (\cite{garcia2011estimation}), and mixture methods (\cite{giacomini2008mixtures},  \cite{li2024parametric}). Nonparametric methods consist of maximum entropy approaches (\cite{bondarenko2003estimation},  \cite{rompolis2010retrieving}), kernel regression estimations (\cite{ait1998nonparametric}, \cite{ait2003nonparametric},  \cite{FENG20161}), and curve-fitting methods (\cite{lu2021sieve}, \cite{frasso2022direct}).

The dense literature on RND estimation reflects the difficulty of the problem. Nonetheless, most approaches still rely on the assumption of market completeness, and the estimation of the RND in the presence of liquidity frictions remains unexplored, despite being a major challenge 
(\cite{anthonisz2017asset},  \cite{glebkin2023illiquidity}).

In practice, 
quotes are available only at discrete and irregular strikes, even for the most liquid and traded stock options whereas much of the literature assumes frictionless and sufficiently broad option markets, restricting the estimation of the RND to a fine grid of strikes over a compact set. For example, a non-exhaustive list of examples covers approaches such as parametric expansion, generalized and mixture methods (\cite{rompolis2008recovering}, \cite{li2024parametric}); and nonparametric methods based on curve fitting approaches (\cite{monteiro2008recovering}, \cite{lai2014comparison}, \cite{frasso2022direct}). In other cases, the problem of data sparsity is confined to the tails (\cite{bollinger2023principled}) and managed by extrapolation or non-arbitrage constraints. In any case, there is no ad-hoc method to estimate the RND in a situation of severe and structural illiquidity. 

Furthermore, RND extraction is highly sensitive to grid design and numerical conditioning so that irregular strike spacing and the use of global or high-degree polynomial bases make the fit ill-conditioned (\cite{ait2003nonparametric}, \cite{Jackwerth2004}). Consequently, approaches that infer the RND by twice differentiating the pricing function with respect to the strike (\cite{Litzenberger1978}) amplify market noise and may lead to unstable results in settings with sparse and irregular quotes. It is essential to note that market noise, in addition to rendering highly sensitive fitting approaches such as entropy-based methods undesirable (\cite{FENG20161}), can also lead to violation of asset-pricing theory and no-arbitrage principles. For example, curve-fitting approaches first interpolate the implied volatility curve and then derive the RND from the fitted curve. However, when market quotes are sparse and noisy, these methods can yield concave shapes.

We develop a framework to estimate the RND from European option prices in a setup of irregular, discrete, and illiquid quotes, filling the aforementioned gaps. We propose a deep learning model, Deep Log-Sum-Exp Neural Network (Deep-LSE), and prove key statistical properties such as its capacity to approximate any convex function and the consistency of the estimator. Our approach relies on transfer learning to overcome the difficulties of estimating the RND in conditions of illiquidity. We first train the model on a liquid proxy with dense, reliable option quotes, then transfer the learned structure to the illiquid target market, where a light fine-tune on sparse observations leverages the proxy-induced prior to stabilize pricing and for estimating the RND. We test our framework in a simulation environment using the \textcite{bates1996jumps}, \textcite{kou2002jump}, \textcite{andersen2002empirical} and Three-Factor double exponential stochastic volatility jump diffusion models (\cite{andersen2015risk}), and empirically test it on SPX option data. We first estimate the ground truth RND from the full, liquid cross-section. We then emulate illiquidity by censoring most option quotes and fit our model on this reduced sample using transfer learning, hence enforcing the conditions of an illiquid market. Comparing the ground truth RND to the illiquid-fit RND reveals the model robustness under severe illiquidity and provides a clear visual check of model performance. Overall, our framework allows for the estimation of the RND in a situation of extreme market illiquidity, having as few as three option quotes available.

We contribute to the literature in two ways. 
First, from a modelling perspective, our work contributes to a novel nonparametric approach to estimate the RND of illiquid option prices. Specifically, the combination of deep learning and transfer learning theory to address illiquidity is an approach new to the literature. We fit the implied volatility curve through Deep-LSE. Then, we use transfer learning to address the challenge of irregular, discrete, and illiquid quotes. The model learns on a liquid proxy and transfers the knowledge to the illiquid options market. We, thus, enrich the recent literature on nonparametric risk-neutral density estimation (\cite{dalderop2020nonparametric}, \cite{almeida2023nonparametric}, \cite{qu2025estimating}) by introducing a novel deep-learning model that is specifically constructed for illiquid markets, rather than relying on methods developed for liquid conditions. Overall, the Deep Log-Sum-Exp Neural Network is a novel machine learning algorithm for regression tasks. It builds on the individual strengths of the Input Convex Neural Network (\cite{amos2017input}) and the Log-Sum-Exp class of functions (\cite{calafiore2019log}) to obtain a final model that ensures convexity and a deep, multilayer, Neural Network architecture. The latter is a crucial characteristic during the transfer learning phase where a deep architecture helps to embed knowledge and makes the transfer more effective.

Second, on theoretical grounds, we prove that the estimator guarantees convexity in inputs, which is a crucial feature to interpolate implied volatilities. This characteristic becomes crucial in illiquid market conditions, where the interpolation of the implied volatility function would otherwise be inadequately constrained due to the absence of dense market option quotes. We obtain analytical bounds that link the Deep-LSE model with the max-affine class function, which allows us to prove the Universal Approximation theorem for the Deep-LSE. This result establishes that the Deep-LSE architecture is expressive enough to approximate the implied volatility curve while preserving the convex structure required for stable fitting and interpolation even in conditions of extreme illiquidity. In addition, we prove the consistency of the estimator.

The rest of the paper is organized as follows. In Section \ref{section_dnn_tl}, we present the methodology of transfer and deep learning to estimate the RND in illiquid markets. We show in Section \ref{theoretical_framework} the theoretical framework of the Deep Log-Sum-Exp Neural Network, and its theoretical properties. In Section \ref{simulation_section}, we perform a Monte Carlo simulation analysis to evaluate our framework. Specifically, we simulate the Bates model and emulate the condition of an illiquid market. Then, we estimate the Deep-LSE model and perform transfer learning to fit it on illiquid strikes, and compare the result with quadratic splines. In Section \ref{empa_section}, we apply our framework to SPX option data to highlight the benefits of the method in illiquid market conditions.

\newpage

\section{Deep Transfer Learning for illiquid RND Estimation of European options}\label{section_dnn_tl}

\noindent In a complete and arbitrage-free market, \textcite{COX1976} show that the value $V_t$ of a European option with payoff function $\Phi(S_t;K)$ at time $t$, with expiry $T$, term $\tau = T-t$ and strike $K$, is 
\begin{equation*}
  V_t(K,T) = e^{-r\tau} \int_{0}^{\infty} \Phi(S_T;K)\, f_{t,T}(S_T)\, dS_T,
  \label{eq:callprice}
\end{equation*}
where $r$ is the risk-free rate, $S_T$ is the terminal price of the underlying and $f_{t,T}$ is the terminal risk-neutral distribution of the underlying equity. For European options, the payoff is $\Phi(S_T;K)=(S_T-K)^+$ for calls and $\Phi(S_T;K)=(K-S_T)^+$ for puts. Taking the second derivative with respect to strike (\cite{Litzenberger1978}) yields the risk-neutral density
\begin{equation*}
  e^{r\tau} \frac{\partial^{2} V_t(K,T)}{\partial K^{2}} = f_{t,T}(K).
  \label{eq:BL}
\end{equation*}

In theory, one needs a continuum of option prices across strike levels for a given term in order to estimate the RND. However, only a discrete set of strikes is listed and actively traded, and many of these options suffer from low trading activity and wide bid–ask spreads. As a result, the available quotes may be noisy and sparse, particularly for deep in–the–money or deep out–of–the–money strikes. The illiquidity that is common in option market makes it hard to construct a stable and continuous pricing function, and this complicates the estimation of the underlying risk-neutral density.

We leverage transfer learning to overcome the challenges posed by severe market illiquidity. In general, transfer learning involves first training a model, and then using the pre-trained model and fine-tuning it on a different dataset, resulting in modifications to the original model. This allows the model to utilize information learned during the initial training phase, which is especially useful when the new setting has limited data. In our setup, we estimate the implied RND for an illiquid options market using information transferred from a similar and liquid market. The deep and transfer learning–based approach proceeds in two phases:

\begin{itemize}

    \item[(i)] First step recovery. We introduce Deep-LSE to model the implied volatility function in a liquid option market, where the underlying (proxy) asset closely resembles that of the illiquid option market of interest. 
    
    \item[(ii)] Second step recovery. We fine-tune the learned implied volatility function with the data available from the illiquid option market. This method allows us to model accurately the implied volatility function of the illiquid option market. The Deep-LSE has learned in the first training phase the general feature of an implied volatility function and is able to interpolate the illiquid implied volatility function even with just a few observations.

\end{itemize}

Our method belongs to nonparametric volatility smoothing approach, as we first interpolate the volatility curve and then reconstruct the RND. We illustrate the steps in Algorithm \ref{alg_TL_RND}.

\begin{algorithm}[H]
  \caption{Illiquid RND estimation with Transfer Learning}
  \label{alg_TL_RND}
  \begin{algorithmic}[1]
    \small
    \Require Liquid option dataset $\mathcal{D}^{\mathrm{liq}}$, illiquid option dataset $\mathcal{D}^{\mathrm{ill}}$, Deep-LSE $f_{\theta}$, risk-free rate $r$, maturity $T$, strike grid $\{K_g\}_{g=1}^G$
    \Ensure Estimated risk-neutral density $\hat{f}_{t,T}$ for the illiquid market

    \State \textbf{First Step Recovery: Pre-training (liquid market)}
    \State Define the Deep-LSE estimator $f_{\theta}$
    \State Train $f_{\theta}$ on $\mathcal{D}^{\mathrm{liq}}$ to fit the implied volatility surface of the liquid market

    \State

    \State \textbf{Second Step Recovery: Transfer learning (illiquid market)}
    \State Consider the pre-trained model $f_{\theta}$ 
    \State Fine-tune $f_{\theta}$ on $\mathcal{D}^{\mathrm{ill}}$

    \State

    \State \textbf{Volatility smoothing and RND recovery}
    \For{each strike $K_g$ in $\{K_g\}_{g=1}^G$}
      \State Compute the moneyness $x_g$ 
      \State $\hat{\sigma}^{\mathrm{imp}}(K_g, T) \gets f_{\theta}(x_g)$
      \State $\hat{V}_t(K_g, T) \gets \text{OptionPrice}\big(\hat{\sigma}^{\mathrm{imp}}(K_g, T), K_g, T, r\big)$
    \EndFor

    \State Use finite differences to approximate $\partial^2 \hat{V}_t / \partial K^2$ on $\{K_g\}$ 
    \State For each $K_g$, set $\hat{f}_{t,T}(K_g) \gets e^{r \tau} \, \partial^2 \hat{V}_t / \partial K^2 (K_g, T)$
    \State \Return $\hat{f}_{t,T}$ on the strike grid
  \end{algorithmic}
\end{algorithm}

\newpage
\section{Deep Log-Sum-Exp Neural Network}\label{theoretical_framework}

In this section, we present the Deep Log-Sum-Exp Neural Network and the application of Transfer Learning for the estimation of the RND in an illiquid option market. In theory, transfer learning can be applied to various machine learning algorithms. We choose a Deep Neural Network architecture to leverage the flexibility of transfer learning and for the performance as an interpolation function. The deep architecture is a crucial characteristic  during the transfer learning phase as it helps to embed knowledge and makes the transfer more effective. We illustrate the Deep-LSE with $L$ layers, and discuss an example with two layers in Appendix \ref{theory_appendix}.

Let \(x\in\mathbb{R}^d\) denote the input, for each layer \(\ell=1,\dots,L\), let \(K_\ell\in\mathbb{N}\) be the number of affine linear functions, that play the role of the number of neurons, and let \(T_\ell>0\) be the temperature parameters, and \(c_{\mathrm{out}}\in\mathbb{R}\) be a global output bias. For a vector \(u=(u_1,\dots,u_m)^\top\in\mathbb{R}^m\) and \(T>0\), define
\[
\operatorname{LSE}_T(u)
\;=\;
T\,\log\!\Big(\sum_{i=1}^m e^{\,u_i/T}\Big),
\qquad
\text{as } T\downarrow 0,\ \operatorname{LSE}_T(u)\to \max_i u_i.
\]
For each layer \(\ell=1,\dots,L\) and each neuron \(k=1,\dots,K_\ell\), define an affine linear function
\[
\ell^{(\ell)}_k(x)
\;=\;
a^{(\ell)}_{k}{}^{\!\top}x + b^{(\ell)}_k,
\qquad
a^{(\ell)}_k\in\mathbb{R}^d,\; b^{(\ell)}_k\in\mathbb{R}.
\]
Then, we stack these $K_\ell$ affine linear function in a vector, which represents a generic $\ell$ layer with $K_\ell$ neurons
\[
L^{(\ell)}(x)
=\begin{bmatrix}
\ell^{(\ell)}_1(x)\\ \vdots\\ \ell^{(\ell)}_{K_\ell}(x)
\end{bmatrix}
\in\mathbb{R}^{K_\ell}.
\]
Let \(A^{(\ell)}\in\mathbb{R}^{K_\ell\times d}\) collect the row vectors
\((a^{(\ell)}_k)^\top\) and let \(b^{(\ell)}\in\mathbb{R}^{K_\ell}\) collect the biases, then
\[
L^{(\ell)}(x)=A^{(\ell)}x+b^{(\ell)},
\]
and the output of the first layer is the scalar
\[
z_1(x)
\;=\;
\operatorname{LSE}_{T_1}\!\big(L^{(1)}(x)\big)
\;=\;
\operatorname{LSE}_{T_1}\!\big(A^{(1)}x+b^{(1)}\big)
\in\mathbb{R}.
\]
For each succeeding layer \(\ell=2,\dots,L\) and each affine piece \(k=1,\dots,K_\ell\), we define a positive recursion weight from \(z_{\ell-1}(x)\) to the \(k\)-th affine term of layer \(\ell\) by
\[
\alpha^{(\ell)}_k
\;=\;
\operatorname{softplus}\!\big(\eta^{(\ell)}_k\big),
\qquad
\eta^{(\ell)}_k\in\mathbb{R},
\]
and collect them into
\[
\alpha^{(\ell)}
=
(\alpha^{(\ell)}_1,\ldots,\alpha^{(\ell)}_{K_\ell})^\top
\in\mathbb{R}^{K_\ell}_{> 0}.
\]
Now, for $\ell \ge2$, we define the vector $S^{(\ell)}(x)$ that collects the recursion term and the $K_\ell$ affine linear functions
\[
S^{(\ell)}(x)
\;=\;
\begin{bmatrix}
\alpha^{(\ell)}_1 z_{\ell-1}(x) + \ell^{(\ell)}_1(x)\\
\vdots\\
\alpha^{(\ell)}_{K_\ell} z_{\ell-1}(x) + \ell^{(\ell)}_{K_\ell}(x)
\end{bmatrix}
\;=\;
\alpha^{(\ell)}\, z_{\ell-1}(x) + L^{(\ell)}(x)
\;\in\;\mathbb{R}^{K_\ell},
\]
which yields
\[
S^{(\ell)}(x)
\;=\;
\alpha^{(\ell)}\, z_{\ell-1}(x) + A^{(\ell)}x + b^{(\ell)}.
\]
The scalar output of a generic layer \(\ell\) is then
\[
z_\ell(x)
\;=\;
\operatorname{LSE}_{T_\ell}\!\big(S^{(\ell)}(x)\big)
\;=\;
T_\ell \log\!\Big(\sum_{k=1}^{K_\ell}
e^{\,(\alpha^{(\ell)}_k z_{\ell-1}(x)+\ell^{(\ell)}_k(x))/T_\ell}\Big)
\in\mathbb{R}.
\]
Finally, the Deep-LSE output of the final layer $L$ is
\[
y(x)
\;=\;
z_L(x) + c_{\mathrm{out}}.
\]
In compact form, the \(L\)-layer Deep-LSE can be written as
\[
\begin{aligned}
z_1(x)
&=\operatorname{LSE}_{T_1}\!\big(A^{(1)}x+b^{(1)}\big),\\
z_\ell(x)
&=\operatorname{LSE}_{T_\ell}\!\big(\alpha^{(\ell)} z_{\ell-1}(x)
+ A^{(\ell)}x + b^{(\ell)}\big),
\qquad \ell=2,\dots,L,\\
y(x)
&= z_L(x) + c_{\mathrm{out}}.
\end{aligned}
\]

\subsection{Convexity}

In practice, the estimation of the RND from illiquid option quotes is a challenging task. A common approach involves fitting the implied volatility curve of the available option quotes, converting them back to option prices, and then differentiating twice with respect to the strike to obtain the RND (\cite{Litzenberger1978}). The level, shape, and curvature of the implied volatility curve are unknown and cannot be inferred accurately from illiquid quotes. For this reason, another crucial aspect of our model is that it is convex in its inputs and always returns a convex function. This property encourages the network to learn a meaningful and general representation of the implied volatility curve even under extreme illiquidity, without imposing convexity as an explicit modelling assumption and data requirement.

Lemma \ref{epigraph_lemma} describes the preservation of convexity under composition, and then we show that the Deep-LSE architecture is indeed convex with respect to its inputs. In fact, the Deep-LSE estimator is a composition of linear affine functions, positive recursion weights, and convex transformations, therefore it preserves convexity. The proofs of results related to convexity are reported in Appendix \ref{convexity_appendix}.

\begin{lemma}[Monotone convex composition]\label{epigraph_lemma}
\smallskip
Let $h:\mathbb{R}^m\to\mathbb{R}$ be convex and nondecreasing for each argument, and let $f_i:\mathbb{R}^d\to\mathbb{R}$ be convex for $i=1,\dots,m$. Then $x\mapsto h\big(f_1(x),\dots,f_m(x)\big)$ is convex.
\end{lemma}

\medskip
\begin{theorem}[Convexity of deep-LSE network]\label{convexity2}

Consider the Deep-LSE estimator $y(x)=z^{(L)}(x)+c_{\mathrm{out}}$ where, for $\ell\ge 2$, the recursion weights are $\alpha^{(\ell)}\in\mathbb{R}^{K_\ell}_{> 0}$.

\noindent Then $y:\mathbb{R}^d\to\mathbb{R}$ is convex in $x$, as each $z^{(\ell)}$ is convex for all layers $\ell=1,\dots,L$.
\end{theorem}

\noindent We illustrate in Example \ref{convexity} the case of the Deep-LSE estimator with 2 layers.

\medskip
\begin{example}[Convexity of the Deep-LSE network with 2 layers.]\label{convexity}

Define the Deep-LSE estimator with two layers as $y(x)=\operatorname{LSE}_{T_2}\!\big(\alpha\,z_1(x)+A^{(2)}x+b^{(2)}\big)+c_{\mathrm{out}}$ with finite temperature $T_1,T_2>0$ and neurons for the two layers $K_1,K_2\in\mathbb{N}$. 

\noindent Then, with the vector of recursive weights at the second layer $\alpha=(\alpha_1,\dots,\alpha_{K_2})^\top\in\mathbb{R}^{K_2}_{> 0}$, $y:\mathbb{R}^d\to\mathbb{R}$ is a convex function of $x$.

\end{example}

\subsection{Bounds}

Given any number of layers $\ell=1,...L$, our model $z^{(\ell)}(x)$ is a combination of affine transformations, recursive connections, and Log-Sum-Exp class of functions. This functional class is important because $\text{ as } T\downarrow 0, \operatorname{LSE}_T(u)\to \max_i u_i$ and therefore it collapses to a max function. In other words, it is possible to establish bounds between $z^{(\ell)}(x)$, which depends on the log-sum-exp function, and the max-affine class of functions. In turn, the max-affine class of functions is known to be a universal approximator of convex functions under certain conditions. By induction, we obtain the analytical bounds for $\ell$ layers in Theorem \ref{thm:deep_LSE_vs_max}. In Appendix \ref{bounds_appendix}, we report the  proof.

\medskip

\begin{theorem}[Deep-LSE bounds to deep max-affine surrogate]\label{thm:deep_LSE_vs_max}

Define the deep max-affine surrogate of the Deep-LSE model as 
\[
\bar z^{(1)}(x)\ =\ \max_{i\in K_1}\ \ell^{(1)}_i(x),\qquad
\bar z^{(\ell)}(x)\ =\ \max_{k\in K_\ell}\ \Big(\alpha^{(\ell)}_k\,\bar z^{(\ell-1)}(x)+\ell^{(\ell)}_k(x)\Big)\quad \text{for } \ \ell\ge2,
\]
with positive recursion weights. Then for every $\ell=1,\dots,L$ and every $x\in\mathbb{R}^d$,
\begin{equation}\label{eq:layerwise_sandwich}
\bar z^{(\ell)}(x)\ \le\ z^{(\ell)}(x)\ \le\ \bar z^{(\ell)}(x)\ +\ \Delta_\ell,
\end{equation}
where $\Delta_\ell$ satisfies the recursion
\begin{equation}\label{eq:Delta_recursion}
\Delta_1=T_1\log K_1,\qquad
\Delta_\ell\ =\ T_\ell\log K_\ell\ +\ \alpha^{(\ell)}_{\max}\,\Delta_{\ell-1}\quad \text{for } \ \ell\ge2,
\end{equation}
and in closed form with $\alpha^{(\ell)}_{\max}=\max_{k}\alpha^{(\ell)}_k$
\begin{equation}\label{eq:Delta_closed_form}
\Delta_\ell\ =\ \sum_{j=1}^{\ell}\Bigg(T_j\log K_j\ \prod_{r=j+1}^{\ell}\alpha^{(r)}_{\max}\Bigg).
\end{equation}
In addition, with $\bar y(x)=\bar z^{(L)}(x)+c_{\mathrm{out}}$, it then holds that
\[
\bar y(x)\ \ \le\ y(x) \ \le\ \bar y(x)\ + \ \Delta_L\qquad\text{for all }x\in\mathbb{R}^d.
\]
\end{theorem}

\noindent We illustrate the analytical bounds for a Deep-LSE model with two layers in Example \ref{bounds1}.

\begin{example}[Two layer Deep-LSE bounds to deep max-affine]\label{bounds1}

Consider the max-affine surrogates with two layers
\[
\bar y(x)=\max_{k\in K_2 }\Big(\alpha_k\,\bar z_1(x)+\big\langle A^{(2)}_{k,\cdot},x\big\rangle+b^{(2)}_k\Big)+c_{\mathrm{out}},
\]
with $\bar z_1(x)=\max_{i\in K_1}\big\langle A^{(1)}_{i,\cdot},x\big\rangle+b^{(1)}_i$.
Then, with $\alpha=(\alpha_1,\ldots,\alpha_{K_2})^\top\in\mathbb{R}^{K_2}_{>0}$
\[
\quad \bar y(x)\ \le\ y(x)\ \le\ \bar y(x)\;+\;T_2\log K_2\;+\;\alpha_{\max}\,T_1\log K_1\quad
\qquad\text{for all }x\in\mathbb{R}^d.
\]

\end{example}

\medskip
\noindent We enforce that $\alpha^{(\ell)}_k> 0$ to preserve monotonicity and as a consequence convexity across layers. In fact, $z^{(\ell-1)}\ge \bar z^{(\ell-1)}$ implies $\alpha^{(\ell)}_k z^{(\ell-1)}\ge \alpha^{(\ell)}_k \bar z^{(\ell-1)}$.
In case some $\alpha^{(\ell)}_k$ were negative, the lower bound could fail, and for this reason we enforce positivity on $\alpha^{(\ell)}_k$ with the softplus function. Notice also that $\Delta_\ell$ is independent of $x$, hence the approximation error $z^{(\ell)}-\bar z^{(\ell)}$ is uniformly bounded over $\mathbb{R}^d$.


It is important to notice that the multiplicative term 
\[
\Delta_L
= \sum_{j=1}^{L} \Bigl( T_j \log K_j \, \prod_{r=j+1}^{L} \alpha_{\max}^{(r)} \Bigr)
\]
is crucial for the upper bound defined in Theorem \ref{thm:deep_LSE_vs_max}. It makes the dependence on the positive recursion weights $\alpha_{\max}^{(\ell)}$ explicit, and highlights that the behavior of the estimator as depth grows depends on the recursion weights that control how much the upper bound grows (Eq. \ref{eq:Delta_recursion})
\[
\Delta_\ell\ =\ T_\ell\log K_\ell\ +\ \alpha^{(\ell)}_{\max}\,\Delta_{\ell-1}.
\]
So by controlling $\alpha_{\max}^{(\ell)} < q$ and $q\le 1$ for all $\ell$, and with bounded temperatures and widths $0 \le T_j \log K_j \le M$, we obtain 
\[
\sum_{j=1}^{L} q^{L-j} 
= q^{L-1} + q^{L-2} + \cdots + q^{1} + q^{0} 
= \sum_{k=0}^{L-1} q^k 
= \frac{1 - q^{L}}{1 - q} \quad \text{with } \ q \neq 1.
\]
\[
  \Delta_L \le M \sum_{j=1}^{L} q^{L-j}
  = M \frac{1-q^{L}}{1-q}
  \le \frac{M}{1-q},
\]
so $\Delta_L$ is {uniformly bounded in depth} and $\text{sup}_L\Delta_L \le \frac{M}{1-q}.$ We enforce $q<1$ with $\alpha_{\max}^{(\ell)} \le q$, and the bound for $\ell$ layers Deep-LSE network becomes
\[
\bar y(x)\ \le\ y(x) \ \ \le\ \bar y(x) + \frac{M}{1-q}
\qquad\text{for all }x\in\mathbb{R}^d.
\]

\bigskip

\subsection{Universal Approximation}

The result of Theorem \ref{thm:deep_LSE_vs_max} is crucial to link the class of functions of the Deep-LSE model to the max-affine function, a universal approximator of convex functions. We leverage these results of convex analysis to prove that our model is a universal approximator of convex functions and data points.

In our setup, we obtain the deep surrogate of max-affine functions because our framework accommodates $\ell \ge 1$. For this reason, in Theorem \ref{ybar_pwma} we show that the deep max–affine surrogate is still an affine base, and that the overall functional class of the deep max–affine surrogate is interpreted as a max function of finitely many affine functions. In particular, the network output remains a convex and continuous function of the input. We leverage this results as it allows us to link the deep max-affine surrogate to the classical max-affine function and its known bounds with respect to the log-sum-exp function class. The proofs are in Appendix \ref{uat_appendix}.

\begin{theorem}[Deep max–affine surrogate is a finite max of affines]\label{ybar_pwma}

Define the {combination sets} that neurons across layers as $P_\ell=\{1,\dots,K_1\}\times\cdots\times\{1,\dots,K_\ell\}$ with $|P_\ell|=\prod_{j=1}^\ell K_j<\infty$ for a finite number of layers and with positive recursion weights $\alpha_k^{(\ell)}>0$.
Then there exist affine coefficients $\{(A_p^{[\ell]},b_p^{[\ell]})\}_{p\in P_\ell}\subset\mathbb{R}^d\times \mathbb{R}$ such that, for every $\ell=1,\dots,L$, the deep max-affine surrogate
\[
\bar y(x)=\max_{p\in P_L}\bigl(\langle A^{[L]}_p,x\rangle+b^{[L]}_p\bigr)+c_{\mathrm{out}},
\]
that replicates the combination path set of the Deep-LSE estimator, is a finite maximum of affine functions, where convexity holds.
\end{theorem}

The next is a universal approximation result for the Deep-LSE estimator on the class of continuous convex functions. Given a compact convex set $K\subset\mathbb{R}^d$ and a continuous convex map $f:K\to\mathbb{R}$, it establishes that it is possible to choose $\varepsilon>0$ arbitrarily small and obtain sup-norm convergence of the estimator in approximating any convex and continuous function or data points. In other words, the functional class of the Deep-LSE estimator is dense in the space of continuous and convex functions. We leverage this result in the first step of our approach to interpolate the implied volatility function of the proxy liquid option market.

\medskip
\begin{theorem}[Uniform approximation on a compact convex set by an $L$-layer Deep-LSE network]\label{thm:uniform_approx_lse_streamlined}
Let $K\subset\mathbb{R}^d$ be compact and convex, and let $f:K\to\mathbb{R}$ be continuous and convex. Consider Deep-LSE $y(x)$ estimator with L layers and recursive parameters $\alpha^{(\ell)}\in\mathbb{R}_{>0}^{K_\ell}\ \text{for }\ \ell\ge2$. Then, for every $\varepsilon>0$, there exist specific parameters with finite value of the estimator such that
\[
\sup_{x\in K}\,|y(x)-f(x)|\ <\ \varepsilon .
\]
\end{theorem}

\bigskip

\subsection{Sieve M-Estimation}\label{sieve_sec}

Consider a nonparametric regression problem
\[
y_i = f_0(x_i) + \epsilon_i,
\]
with i.i.d.\ errors such that $\mathbb E[\epsilon_i]=0$, $\mathrm{Var}(\epsilon_i)=\sigma^2<\infty$, $x_i\in\mathcal X\subset\mathbb R^d$. The objective is to estimate a regression function $f_0$, which is unknown, and this can be achieved by minimizing the empirical squared error loss
\[
\hat f_n
  = \arg\min_{f\in\mathcal F} Q_n(f)
  = \arg\min_{f\in\mathcal F}\frac{1}{n}\sum_{i=1}^n\bigl(y_i-f(x_i)\bigr)^2. \quad (f_0\in\mathcal F)
\]
The objective of Sieve estimation is to minimize the squared error loss over a function space $\mathcal F_n$, which approximates the true function $\mathcal F$ as the error tends to $0$ and the sample size increases (\cite{shen1994convergence}, \cite{shen2023asymptotic}). 
Define the sequence of functions $\mathcal F_1 \subseteq \mathcal F_2 \subset \cdots \subset \mathcal F_n \subset \cdots \subset \mathcal F$
and assume $\bigcup_{n=1}^\infty \mathcal F_n$ is dense in $\mathcal F$ with respect to a pseudo-metric $d$, so that for every $f\in\mathcal F$ there exists $\pi_n f\in\mathcal F_n$ such that
\[
d(f,\pi_n f)\to 0 .
\]
A sieve estimator $\hat f_n$ is over $\mathcal F_n$ satisfies
\[
Q_n(\hat f_n)\ \le\ \inf_{f\in\mathcal F_n} Q_n(f) + O_p(\eta_n),
\qquad \eta_n\to 0.
\]
In our setting, we define the sieve class estimator as
\[
\mathcal F_{n}
=\left\{ f_\theta : \theta\in\Theta_n,\ \text{where for } x\in\mathbb{R}^d
\ \begin{aligned}
&\ell^{(\ell)}_{k}(x)=\langle a^{(\ell)}_{k},x\rangle+b^{(\ell)}_{k},\\[2pt]
&z^{(1)}_\theta(x)=\operatorname{LSE}_{T_1}\!\big((\ell^{(1)}_{k}(x))_{k=1}^{K_1}\big),\\[2pt]
&z^{(\ell)}_\theta(x)=\operatorname{LSE}_{T_\ell}\!\big((\alpha^{(\ell)}_{k}\,z^{(\ell-1)}_\theta(x)+\ell^{(\ell)}_{k}(x))_{k=1}^{K_\ell}\big),\ \ \ell\ge2,\\[2pt]
&f_\theta(x)=z^{(L)}_\theta(x)+c_{\mathrm{out}}
\end{aligned}
\,\right\},
\]
and
\[
\Theta_n=\Bigl\{\theta:\ 
\begin{aligned}
&\max_k \|a_k^{(\ell)}\|_* \le S_\ell^{(n)},\quad
 \max_k |b_k^{(\ell)}| \le B_\ell^{(n)},\quad
 \max_k \alpha_k^{(\ell)} \le q_\ell^{(n)}<1,\\
& T_\ell \le \Theta_\ell^{(n)},\quad
 K_\ell \le K_\ell^{(n)},\quad
 |c_{\mathrm{out}}| \le C^{(n)}\ \text{for all }\ell
\end{aligned}
\Bigr\}.
\]
where $B, V_n, S \to \infty$ as $n\to\infty$.

\subsubsection{Existence-Measurability}

We show in Theorem \ref{envelope} that the functional class of the Deep-LSE estimator is bounded, and this result of finite measurability is a requirement for the Sieve class to achieve consistency in this setup. In Appendix \ref{sieve_appendix}, we report the proof. 

\begin{theorem}[Finite Sieve Envelope]\label{envelope}

Let $\mathcal X=\{x:\|x\|\le R\}$ be a bounded input set, and for each fixed $n$, consider the sieve parameter space $\Theta_n$.


\noindent Given $\Theta_n$, let $\mathcal F_n=\{\,f_\theta:\theta\in\Theta_n\,\}$ be the corresponding functional class of Deep-LSE estimator, and define
\[
V_n = C^{(n)} + \sum_{\ell=1}^{L}
\Bigl(R\,S_\ell^{(n)} + B_\ell^{(n)} + \Theta_\ell^{(n)} \log K_\ell^{(n)}\Bigr)
\prod_{r=\ell+1}^{L} q_r^{(n)}.
\]
For each fixed $n$ we have that
\[
\;
\sup_{f\in\mathcal F_n}\|f\|_\infty
\;=\;
\sup_{\theta\in\Theta_n}\ \sup_{x\in\mathcal X} |f_\theta(x)|
\;\le\; V_n\;.
\]

\end{theorem}

\noindent The minimization of the empirical squared error over the sieve class $\mathcal F_r$ ensures that, for any fixed realization of the data, the criterion $Q_n(f)$ is continuous in $f$. Then, the existence and finite measurability of the sieve estimator is guaranteed once the set of functions $\mathcal F_n$ is compact in its support $\mathcal X$.

For a fixed sample $\{x_i\}_{i=1}^n \subset \mathcal X$ of size $n$, consider the evaluation map
\[
E:\Theta_n \longrightarrow \mathbb{R}^n,\qquad
E(\theta)=\bigl(f_\theta(x_1),\,\ldots,\,f_\theta(x_n)\bigr).
\]
Each point of the map $\theta\mapsto f_\theta(x_i)$ is continuous because the Deep-LSE estimator is a composition of affine linear functions, positive recursion weights, and convex transformation, and we define the parameters of the estimator in a compact set $\Theta_n$ with $T_\ell\ge \underline T_\ell>0$. So $E$ is continuous and $E(\Theta_n)$ is compact in $\mathbb{R}^n$.

Now, $\|f-g\|_n$ coincides with the Euclidean norm of $E(\theta_f)-E(\theta_g)$, up to $1/\sqrt{n}$. Therefore, the metric space $(\mathcal F_n,\|\cdot\|_n)$ is isometric with respect to the compact set $E(\Theta_n)\subset\mathbb{R}^n$, and hence, $\mathcal F_n$ is compact under $\|\cdot\|_n$.

\bigskip

\subsubsection{Consistency}

We obtain, in Theorem \ref{lse_map}, a basic Lipschitz estimate for the log-sum-exp function, which is a result we leverage to prove the consistency of the estimator. It shows that the map $\mathrm{LSE}_T:\mathbb{R}^m \to \mathbb{R}$ with $u \mapsto T \log\!\Bigl(\sum_{i=1}^m e^{u_i/T}\Bigr)$ is 1-Lipschitz with respect to standard norms, meaning that small changes in the input vector lead to proportionally small changes in the LSE output.

\begin{theorem}[The log-sum-exp function is $1$-Lipschitz]\label{lse_map}
For $m\in\mathbb{N}$ and $T>0$, define the log-sum-exp function $\mathrm{LSE}_T:\mathbb{R}^m\to\mathbb{R}$
\[
\mathrm{LSE}_T(u)=T\log\!\left(\sum_{i=1}^{m} e^{u_i/T}\right).
\]
Then $\mathrm{LSE}_T$ is $1$-Lipschitz with respect to the sup-norm $\|\cdot\|_\infty$
and also $1$-Lipschitz with respect to the Euclidean norm $\|\cdot\|_2$
\[
\bigl|\mathrm{LSE}_T(u)-\mathrm{LSE}_T(v)\bigr|
\le \|u-v\|_\infty
\quad\text{and}\quad
\le \|u-v\|_2
\qquad \text{with } \ \forall\,u,v\in\mathbb{R}^m.
\]
\end{theorem}

\medskip
\noindent Regarding the parametrization, the Deep-LSE estimator at each layer $\ell$ produces scalars $z^{(\ell)}(x)\in\mathbb{R}$ as the log-sum-exp is a scalar function. So for each neuron $k$ we have
\[
S^{(\ell)}(x)=\alpha^{(\ell)}\, z^{(\ell-1)}(x) + A^{(\ell)}x + b^{(\ell)} \in \mathbb{R}^{K_\ell}.
\]

The vector that recursively links adjacent layers $(\,z^{(\ell-1)}(x),\,x\,)\in\mathbb{R}^{1+d}$ is of dimension $d+1$, so the parameter set $(\alpha^{(\ell)}_{k},\, a^{(\ell)}_{k},\, b^{(\ell)}_{k})$ is of dimension $d+2$. In the first layer, there is no recursive term, so for $K_1$ neurons, the dimension of the first layer is $K_1(d+1)$. In succeeding layers $\ell\ge2$, we have $K_\ell(d+2)$ parameters each, plus the bias $c_{\mathrm{out}}$. Overall, we obtain
\noindent 
\[
m_n = K_1(d+1) + \sum_{\ell=2}^{L} K_\ell(d+2) + 1
\ \asymp\ \sum_{\ell=1}^{L} K_\ell(d+2) + 1 .
\]

\noindent Theorem 14.5 in \textcite{anthony2009neural} states that for the functional class $\mathcal{F}$, if $\epsilon \le 2b$, then
\[
\mathcal{N}_\infty(\epsilon,\mathcal{F},m)
\;\le\;
\left(\frac{4\,e\,m\,b\,W\,(L V)^{\ell}}{\epsilon\,(L V-1)}\right)^{W}.
\]

\noindent To apply this result, we correct the $m_n$ term with a reparametrization trick, so that only adjacent layers are connected. We create $d$ relay units $r^{(1)},\ldots,r^{(L-1)}$ that copy the input to the next layer without modifications
\[
r^{(0)}(x)=x, 
\qquad
r^{(\ell)}(x)=r^{(\ell-1)}(x)\quad \text{for } \ \ell=1,\ldots,L-1.
\]
With this reparametrization, the inner part of layer $\ell\ge2$ becomes 
\[
S_k^{(\ell)}(x)
= \alpha_k^{(\ell)}\, z^{(\ell-1)}(x)
+ \bigl(a_k^{(\ell)}\bigr)^{\!\top} r^{(\ell-1)}(x)
+ b_k^{(\ell)} .
\]
The estimator computes the same function, but all connections are between adjacent layers. Then, the covering number becomes
\begin{align*}
W &= m_n + d(L-1)\\
  &= \sum_{\ell=1}^{L} K_\ell(d+2) + 1 + d(L-1).
\end{align*}

\noindent We combine together the approximation and stability properties of Deep-LSE networks in an asymptotic framework.  
Let $\mathcal{F}_{n}$ denote the Deep-LSE sieve functional class, and let $Q_n(f)$ and $\overline Q_n(f)$ be the empirical and population risk functions associated with the function from the sieve functional class $f\in\mathcal{F}_{n}$. The first result identifies an asymptotic growth condition ($W V_n^{2}\,\log\!\big(V_n^L W\big)=o(n)$) that allow the estimator to be consistent.

\medskip
\begin{theorem}[Asymptotic behavior of Deep-LSE sieve class parameters]\label{consistency}
Under the asymptotic growth condition
\[
W V_n^{2}\,\log\!\big(V_n^L W\big)=o(n)\quad as \ n\to\infty,
\]
we have that
\[
\sup_{f\in \mathcal{F}_{n}}
\bigl|\,Q_n(f)-\overline{Q}_n(f)\,\bigr|
\;\xrightarrow{\,p^{*}\,}\;0 \qquad as \ n\to\infty.
\]
\end{theorem}

\medskip
\noindent Since $Q$ is continuous at $f_0\in\mathcal F$ and $Q(f_0)=\sigma^2<\infty$, the function $f_0$ is a minimizer of the population risk (\cite{shen2023asymptotic}). For any $\varepsilon>0$ we have
\begin{align*}
\inf_{\;f:\ \|f-f_0\|_n\ge \varepsilon}\Bigl(Q_n(f)-Q_n(f_0)\Bigr)
&= \inf_{\;f:\ \|f-f_0\|_n\ge \varepsilon}
    \frac1n\sum_{i=1}^{n}\!\bigl(f(x_i)-f_0(x_i)\bigr)^2 \\
&= \inf_{\;f:\ \|f-f_0\|_n\ge \varepsilon} \|f-f_0\|_n^2 \\
&\ge \varepsilon^{2}
\;>\; 0 .
\end{align*}
In other words, $f_0$ is a minimizer of the empirical risk $Q_n$ with respect to the norm $\|\cdot\|_n$. Hence, by Theorem \ref{consistency} applied to $Q_n$ and $Q$, we have that
\[
\|\hat f_n-f_0\|_n \xrightarrow{p} 0 .
\]

\medskip
\noindent The result of Theorem \ref{consistency} allows us to identify optimal structural conditions for the Deep-LSE model. For example, let $R_n=\sum_{\ell=1}^L K_{\ell}$ be the total number of affine terms. For a fixed $L, W \asymp c(d) R_n$, the consistency condition becomes, up to constants,
$$
R_n V_n^2\left(\log R_n+L \log V_n\right)=o(n).
$$
A possible scenario in which the condition is satisfied is $V_n=O(1)$ and fixed depth. The condition gives $R_n \log R_n=o(n)$. So it grows almost linearly in $n$ and $R_n=\frac{n}{(\log n)^{1+\delta}}$ is a possible solution. Another scenario consists in $V_n=O(1)$ and the depth grows. Now $R_n\left(\log R_n+L_n\right)=o(n)$, and a possible solutions is
$$
R_n=\frac{n}{(\log n)^{1+\delta}}, \quad L_n=o(\log n).
$$

\subsection{Optimal Stopping Criteria for Transfer Learning}

Another key design choice concerns when to stop the transfer-learning fine-tuning step, which forms the second phase of the inference procedure. Our approach is simple, the second phase of the estimation procedure stops when the benefit of continuing it is outweighed by the costs. In practice, we quantify the cost by evaluating the divergence between the weights of the Deep-LSE during the first and second estimation steps. The rationale is that when divergence grows, the model is estimating a function that is far from the initial source function, and hence we stop the fine-tuning.

Let $P$ be a prior distribution of the pre-trained estimator weights $w_0$. In order to obtain the analytical expression, we consider a Gaussian prior for both the initial estimator weights $P$ and the same weights during fine-tuning $Q$
\begin{equation*}
P=\mathcal{N}\!\big(w_0,\;\Sigma_P\big),\qquad
Q_t=\mathcal{N}\!\big(\mu_t,\;\Sigma_Q\big),
\label{eq:gaussians}
\end{equation*}
where $\mu_t$ the mean of $Q_t$, and $\Sigma_Q$ is the posterior covariance. The Bayes bounds allows us to control the population risk $R(Q_t)=\mathbb{E}_{w\sim Q_t}[R(w)]$ via the empirical target risk $\widehat R(Q_t)=\mathbb{E}_{w\sim Q_t}[\widehat R(w)]$ as
\begin{equation*}
R(Q_t)\;\le\;\widehat R(Q_t)\;+\;\sqrt{\frac{\mathrm{KL}(Q_t\|P)\;+\;\ln\!\frac{2\sqrt n}{\delta}}{2(n-1)}}
\qquad\text{with probability at least }1-\delta.
\label{eq:pac-bayes}
\end{equation*}
where $KL$ is the Kullback-Leibler divergence between the estimator weights distributions of the pre-training and fine-tuning phase
\begin{equation*}
\mathrm{KL}(Q\|P)
=\frac{1}{2}\Big(
\mathrm{tr}(\Sigma_P^{-1}\Sigma_Q)
+(\mu_P-\mu_Q)^\top\Sigma_P^{-1}(\mu_P-\mu_Q)
-p+\ln\frac{\det \Sigma_P}{\det \Sigma_Q}
\Big).
\label{eq:kl-gaussians}
\end{equation*}

Define the objective
\begin{equation*}
\mathcal{B}(t)\;=\;\underbrace{\widehat R(Q_t)}_{\text{empirical fit}}
\;+\;c\;\sqrt{\mathrm{KL}(Q_t\|P)}\!
\label{eq:free-energy}
\end{equation*}
with $c>0$ a constant. We stop at the first point in time $t^\star$ that satisfies
\begin{equation*}
\frac{d}{dt}\,\mathcal{B}(t)\Big|_{t=t^\star}=0
\quad\Longleftrightarrow\quad
\Big|\tfrac{d}{dt}\widehat R(Q_t)\Big|
\;=\;
c\;\tfrac{d}{dt}\sqrt{\mathrm{KL}(Q_t\|P)}\ \ \text{at }t=t^\star.
\label{eq:stationary}
\end{equation*}


\section{Simulation Studies}\label{simulation_section}

We perform a simulation study using the  \textcite{bates1996jumps} stochastic volatility model. We simulate a cross-section of option prices and market conditions (liquid and illiquid), and compare the estimated RND with the ground truth implied by the data-generating process. For completeness, we also test the simulation study on the \textcite{kou2002jump}, \textcite{andersen2002empirical}, and Three-Factor Double Exponential Stochastic volatility (\cite{andersen2015risk}) models in Appendix \ref{more_sim}. The Bates stochastic volatility jump--diffusion model is
\begin{align}
    dS_t &= S_t\Bigl[(r - \lambda k)\,dt + \sqrt{v_t}\,dW_t^{(1)} + (J_t - 1)\,dN_t\Bigr], \\
    dv_t &= \kappa(\theta - v_t)\,dt + \eta\sqrt{v_t}\,dW_t^{(2)}, \\
    \langle dW_t^{(1)}, dW_t^{(2)}\rangle &= \rho\,dt, \qquad
    \mathbb{P}(dN_t = 1) = \lambda\,dt.
\end{align}
with $k = \mathbb{E}[J_t - 1] = e^{\mu_j + \frac{1}{2}\sigma_j^2} - 1$ and $Y_t \sim \mathcal{N}(\mu_j,\sigma_j^2)$. Here, $S_t>0$ denotes the stock price, $v_t>0$ is the instantaneous variance, $r\in\mathbb{R}$ is the risk-free rate, $\sigma>0$ is a constant diffusion volatility, $\mu_j\in\mathbb{R}$ and $\sigma_j>0$ are the mean and standard deviation of the normal jump sizes $Y_t$, and $k$ is the mean relative jump size. The processes $(W_t^{(1)})_{t\ge0}$ and $(W_t^{(2)})_{t\ge0}$ form a two-dimensional Brownian motion with correlation $\rho\in[-1,1]$, $(N_t)_{t\ge0}$ is a Poisson process with jump intensity $\lambda>0$ independent of $(W_t^{(1)},W_t^{(2)})$, and $(J_t)_{t\ge0}$ are i.i.d.\ lognormal jump multipliers given by $J_t=e^{Y_t}$. The parameters $\kappa>0$, $\theta>0$, and $\eta>0$ are respectively the speed of mean reversion, the long-run mean, and the volatility of the variance process $(v_t)_{t\ge0}$. In Table \ref{tab:Bates_param} we report the parameters we use for the simulation.

\begin{table}[h]
    \centering
    \begin{tabular}{c c c c c c c c c c c}
    \hline
    $S_0$ & $r$ & $v_0$ & $\kappa$ & $\theta$ & $\eta$ & $\rho$ & $\lambda$ & $\mu_j$ & $\sigma_j$ \\
    \hline
    100 & 0.06 & 0.09 & 3.0 & 0.07 & 0.3 & -0.34 & 0.5 & -0.09 & 0.45 \\
    \end{tabular}
    \caption{Simulated parameters for Bates model}
    \label{tab:Bates_param}
\end{table}



In Fig. \ref{fig:setup_bates}, the blue line is the implied volatility curve that we retrieve with the synthetic option prices, while the orange line is the proxy implied volatility (IV) curve. To obtain the target implied volatility curve, we assume it is a translation of the source implied volatility curve. Specifically, we apply a $-10\%$ decrease on the implied volatilities (y-axis) and a $+20\$$ increase on the strikes (x-axis). In Appendix \ref{more_sim}, we test the framework generating the proxy and illiquid stocks from two different data-generating processes, thereby obtaining two distinct implied volatilities directly, rather than translating the curve. We use the entire proxy to train the model, then, leveraging transfer learning, we fine-tune the pretrained model on the illiquid market observations. The illiquid strikes are simulated by random sampling in the interval of option prices that are 10\%-25\% out of the money (green illiquid strikes) and by random sampling in the interval of option prices that are 10\%-25\% in-the-money (orange illiquid strikes). 
They are represented by orange and green dots, with the associated implied volatility and option prices. 
The first experiment we perform involves emulating the condition of a severely illiquid market. In fact, we select only the three in-the-money strikes to represent the illiquid market (Strikes $K=83, 94, 95$), and maturity one year. 

\begin{figure}[H]
    \centering
    \begin{subfigure}{0.6\textwidth}
        \includegraphics[width=\linewidth]{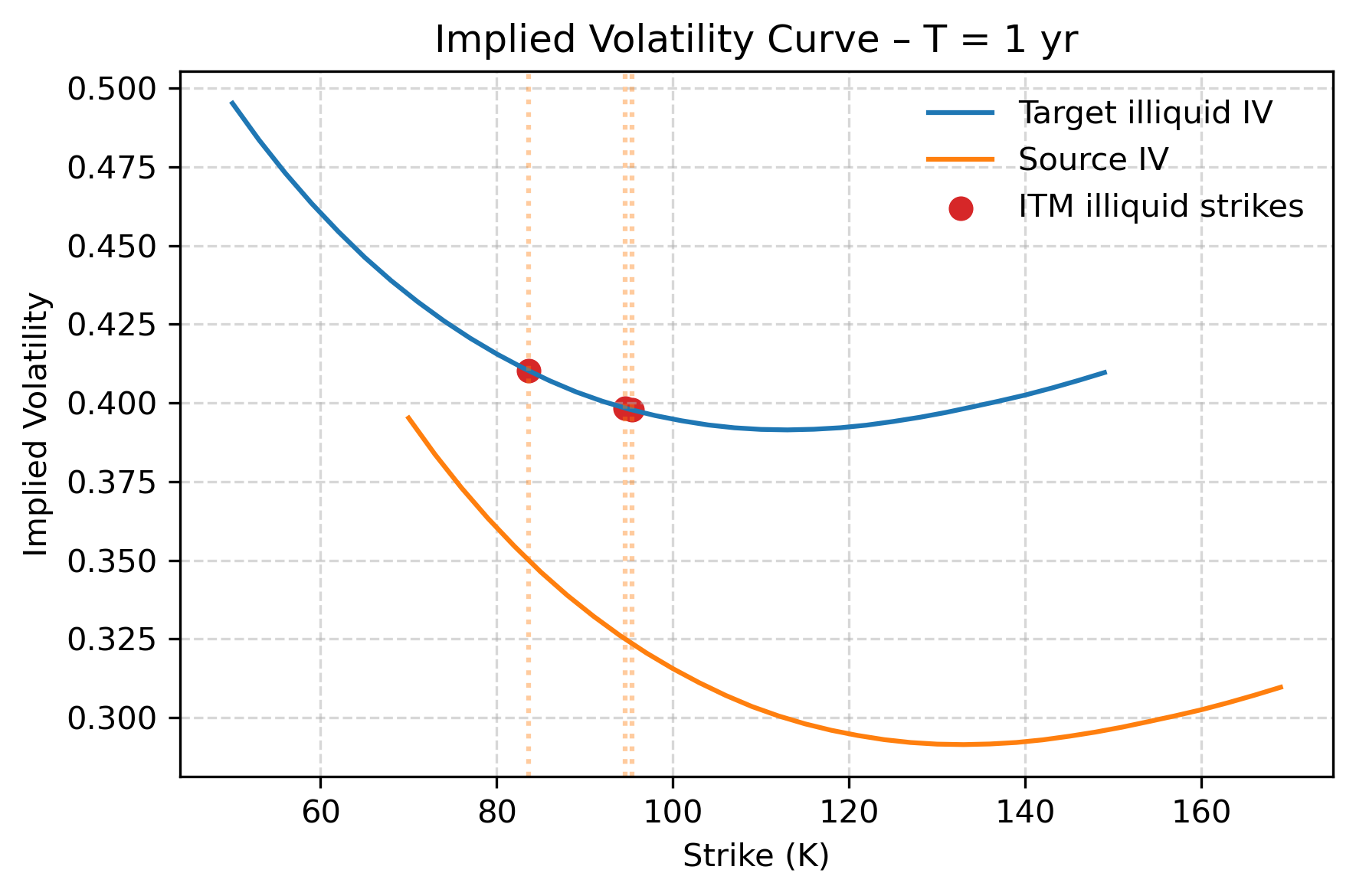} 
    \end{subfigure}
    \caption{Setup of implied volatility curves using Bates model. In blue, the illiquid target implied volatility. In orange, the liquid proxy we use to train the model.}
    \label{fig:setup_bates}
\end{figure}

For the estimation of the IV curve of the liquid proxy and illiquid target, we use the Deep-LSE model with 2 layers and 3 affine terms each. In Fig. \ref{fig:rnd_bates}, we compare the estimate of the RND of our model versus the interpolation of quadratic splines. The blue curve is the ground-truth RND, while the orange and green represent the estimate of the RND of our model and quadratic splines, respectively. It is possible to observe how quadratic splines are not able to recover correctly the ground-truth illiquid RND, especially on the right tail. In contrast, our Deep-LSE, after performing transfer learning, produces a good fit of the RND.

\begin{figure}[H]
    \centering
    \begin{subfigure}{0.76\textwidth}
        \includegraphics[width=\linewidth]{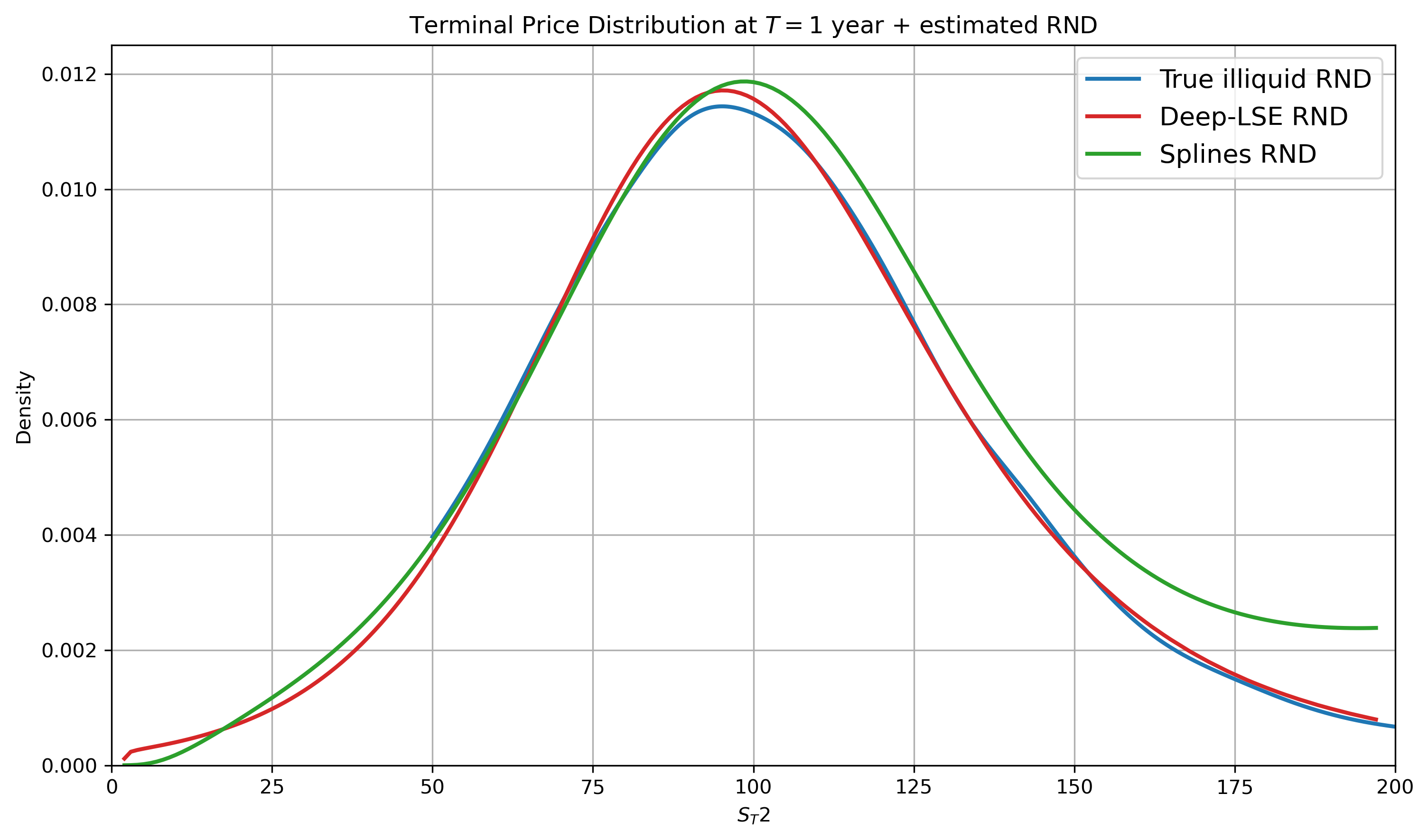} 
    \end{subfigure}
    \caption{Illiquid RND recovery of Deep-LSE (orange curve) and quadratic splines (green curve) in comparison with the target ground truth simulated RND (blue curve).}
    \label{fig:rnd_bates}
\end{figure}

In Fig. \ref{fig:train_source_bates}, we observe the capacity of our Deep-LSE model to approximate convex functions. In this case, the convex function represents the implied volatility curve of the liquid, proxy, asset. During the learning process, the model starts learning the shape and curvature of the proxy implied volatility, and between iteration $6-8$ it starts fitting a convex function. At the end of the learning process on the source data (liquid proxy), the model perfectly fits its implied volatility curve.

\begin{figure}[]
\centering
\begin{minipage}{\textwidth}
\begin{subfigure}{0.42\textwidth}
  \includegraphics[width=\linewidth]{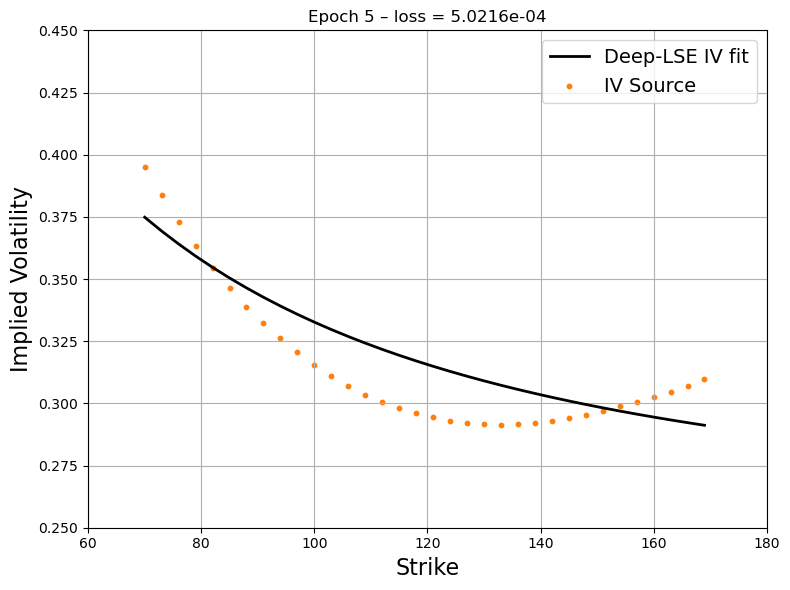}
\end{subfigure}\hspace{1cm}
\begin{subfigure}{0.42\textwidth}
  \includegraphics[width=\linewidth]{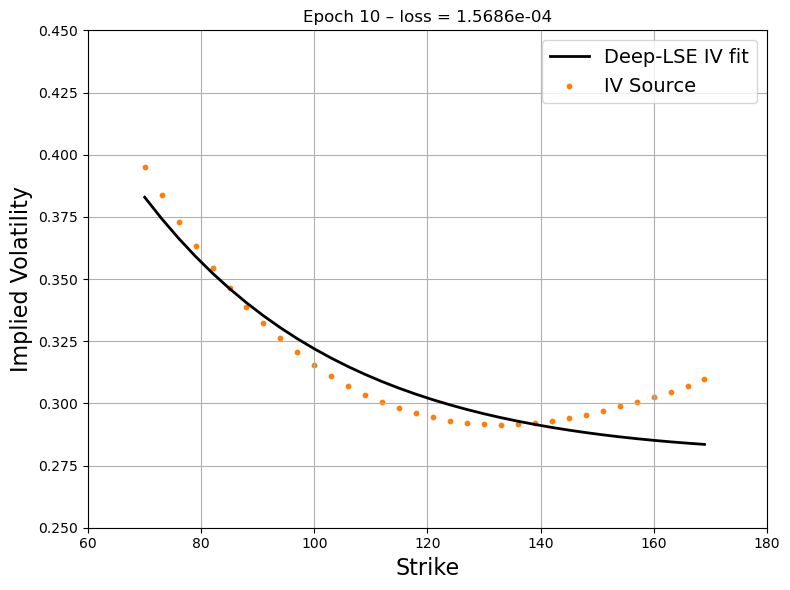}
\end{subfigure}

\medskip

\begin{subfigure}{0.42\textwidth}
  \includegraphics[width=\linewidth]{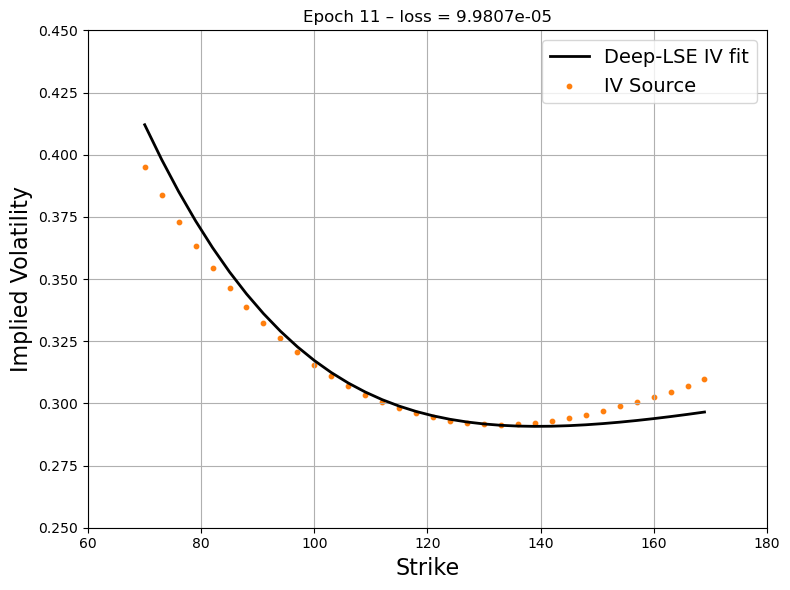}
\end{subfigure}\hspace{1cm}
\begin{subfigure}{0.42\textwidth}
  \includegraphics[width=\linewidth]{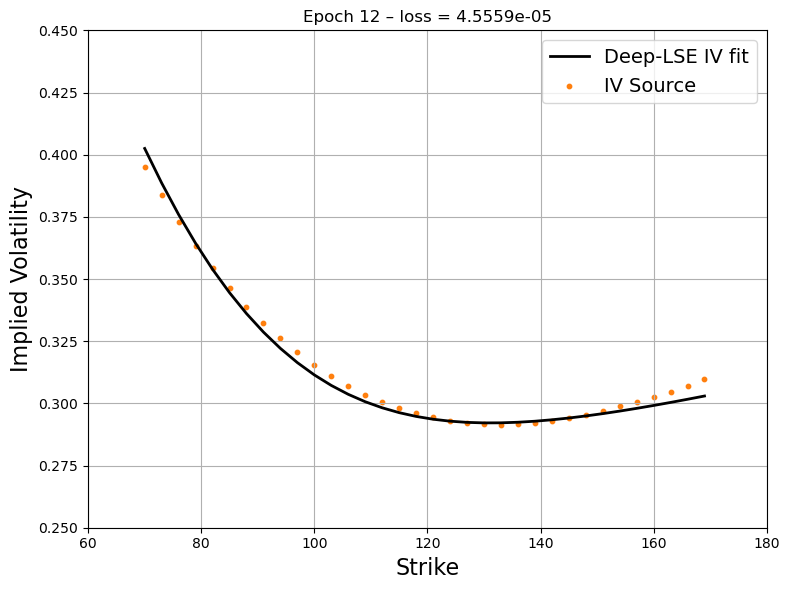}
\end{subfigure}

\medskip

\begin{subfigure}{0.42\textwidth}
  \includegraphics[width=\linewidth]{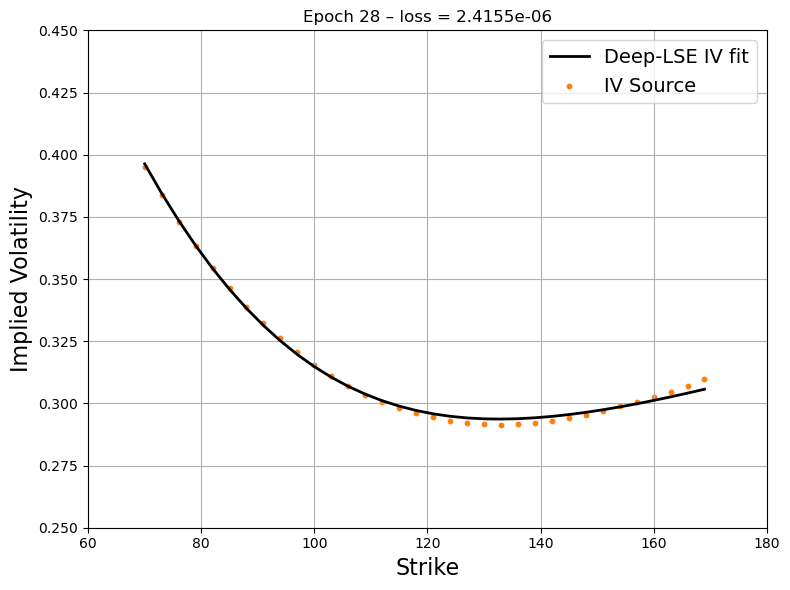}
\end{subfigure}\hspace{1cm}
\begin{subfigure}{0.42\textwidth}
  \includegraphics[width=\linewidth]{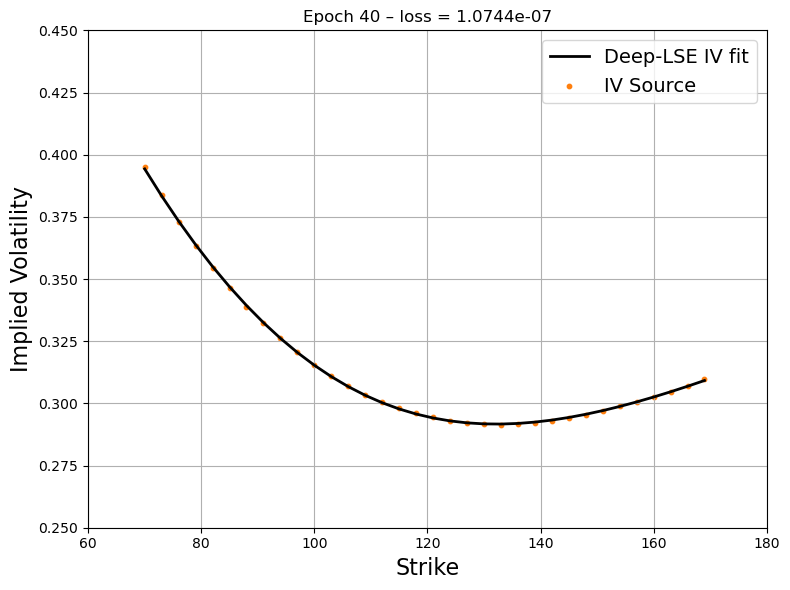}
\end{subfigure}

\caption{First step recovery - Source Deep-LSE Fit. The blue dots represent the implied volatility curve of option quotes of the liquid proxy asset while the blue solid line represents the fit of the interpolating function of the Deep-LSE model.}
\label{fig:train_source_bates}
\end{minipage}
\end{figure}

The second step of the estimation process of the illiquid RND involves performing transfer learning to fit the model on the (illiquid) target option data. To achieve this, we fine-tune with illiquid option data the model pre-trained on the liquid proxy.

We observe this process in Fig. \ref{fig:train_target_bates}. The starting point is the orange curve represents the implied volatility curve that the model has learned from the simulated proxy implied volatility (red curve), and they coincide. As the fine-tuning continues, the green curve moves from the starting point, indicating that the model is learning the new function of the volatility curve of the illiquid target strikes (orange points). We observe that the model adapts to the illiquid strikes by adjusting the level, then its shape and convexity. At the end of the fine-tuning process, the model has learned accurately (green curve) the target implied volatility of the illiquid market (blue curve).

Overall, we observe in Fig. \ref{fig:train_target_bates} that the Deep-LSE model is able to recover the IV surface and RND in conditions of extreme illiquidity, with as few as three option quotes. In addition, we gather recovery in areas where option quotes are missing.

\begin{figure}[H]
\centering
\begin{minipage}{\textwidth}
\begin{subfigure}{0.42\textwidth}
  \includegraphics[width=\linewidth]{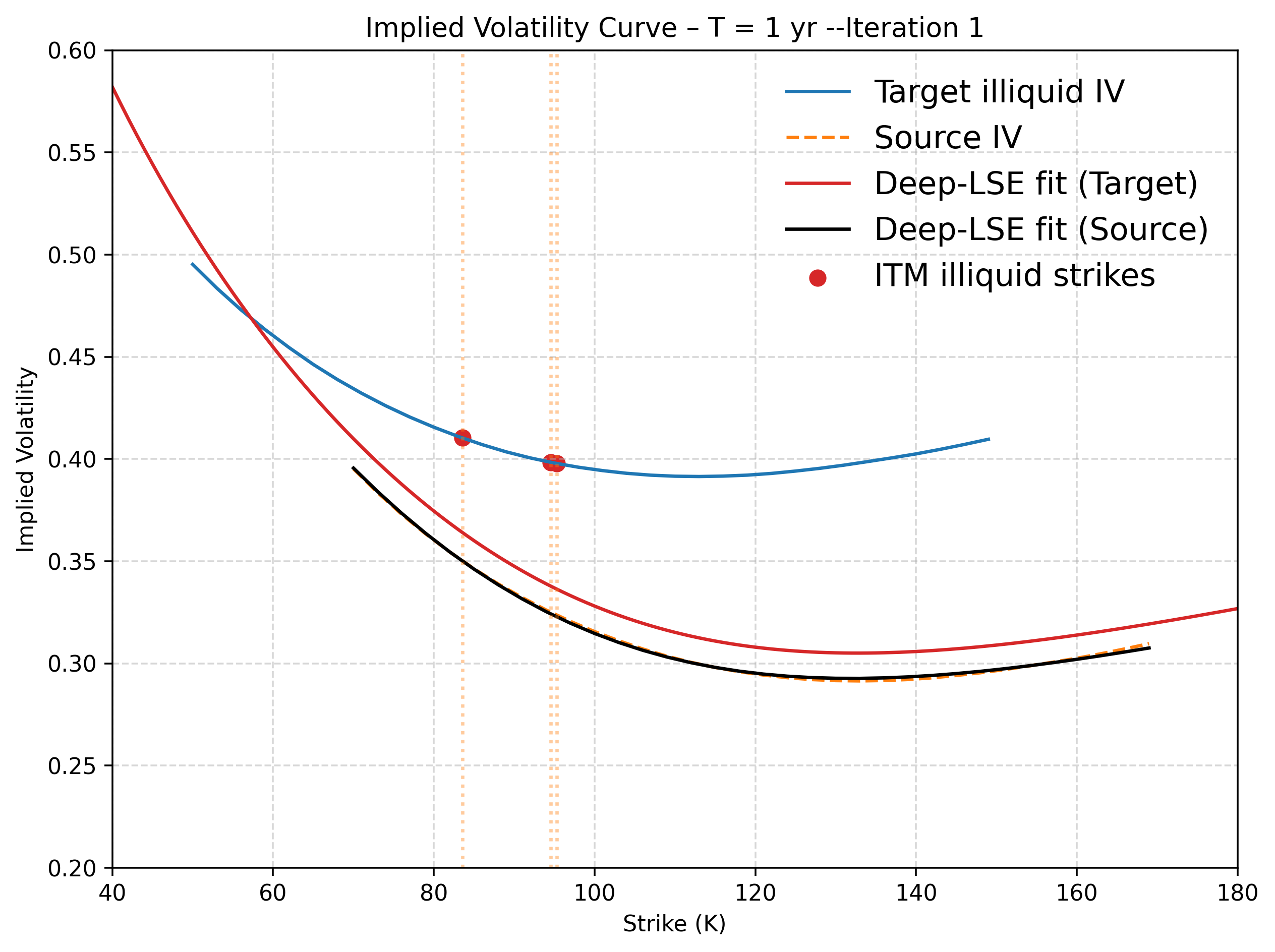}
\end{subfigure}\hspace{1cm}
\begin{subfigure}{0.42\textwidth}
  \includegraphics[width=\linewidth]{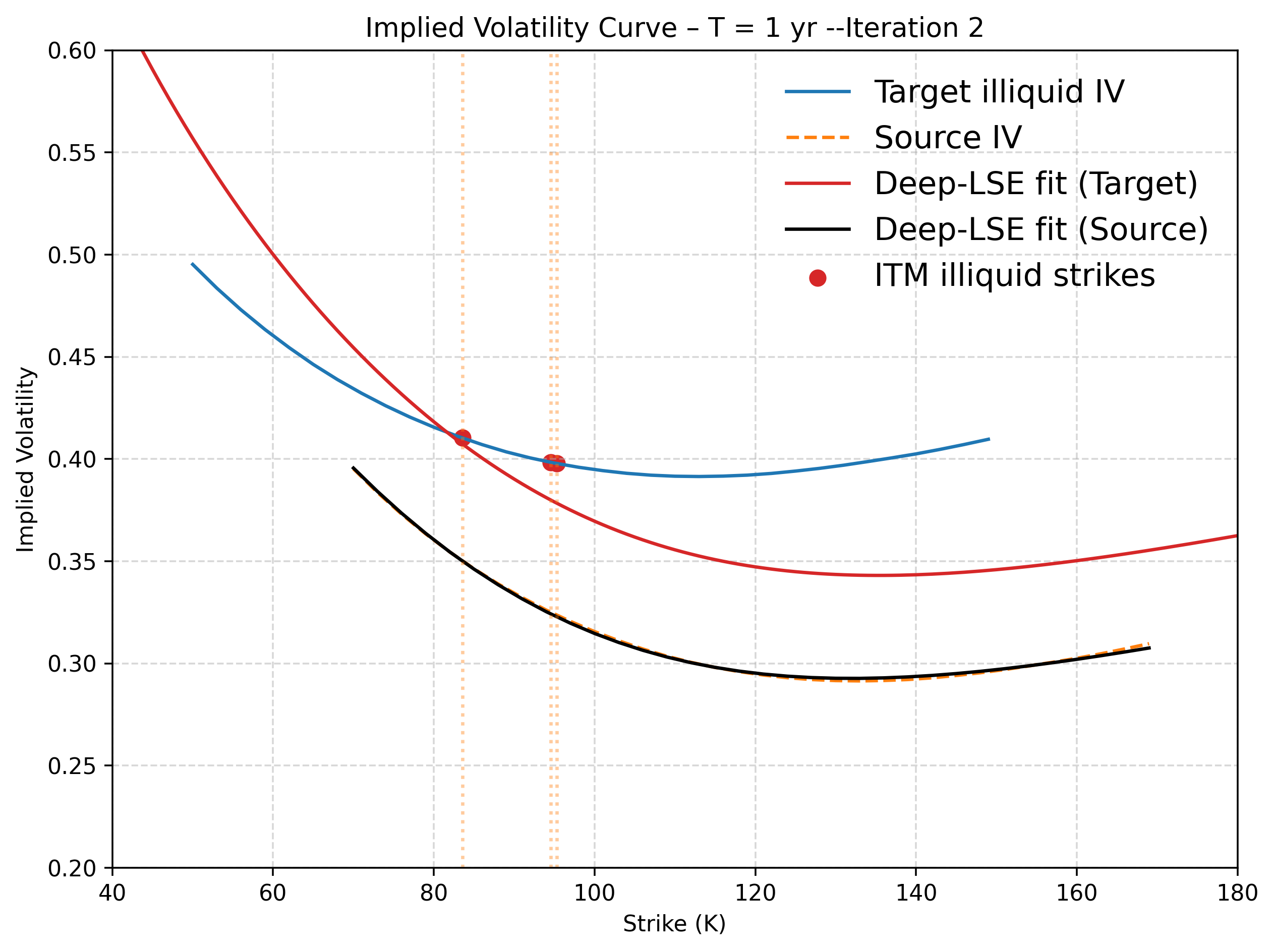}
\end{subfigure}

\medskip

\begin{subfigure}{0.42\textwidth}
  \includegraphics[width=\linewidth]{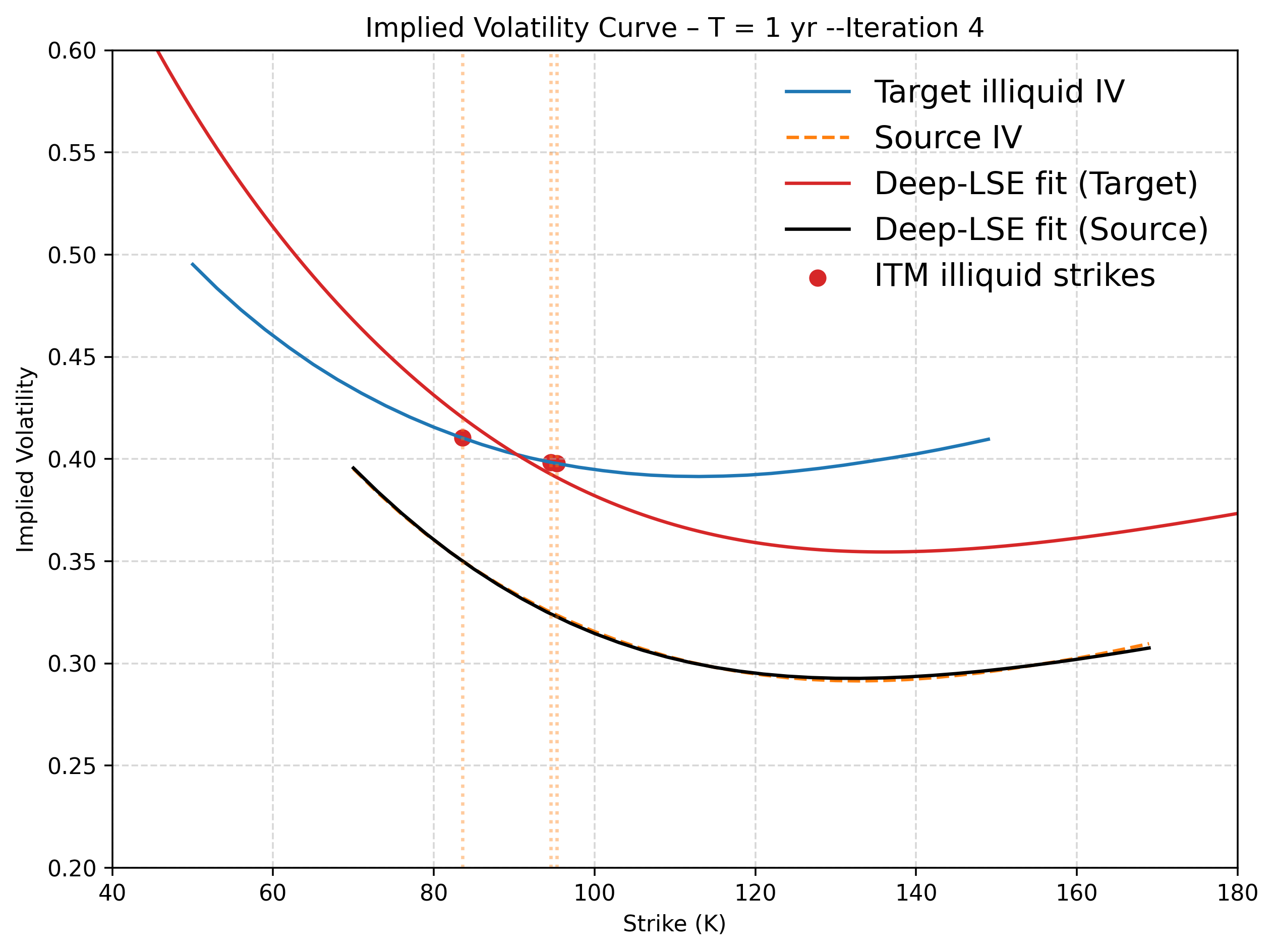}
\end{subfigure}\hspace{1cm}
\begin{subfigure}{0.42\textwidth}
  \includegraphics[width=\linewidth]{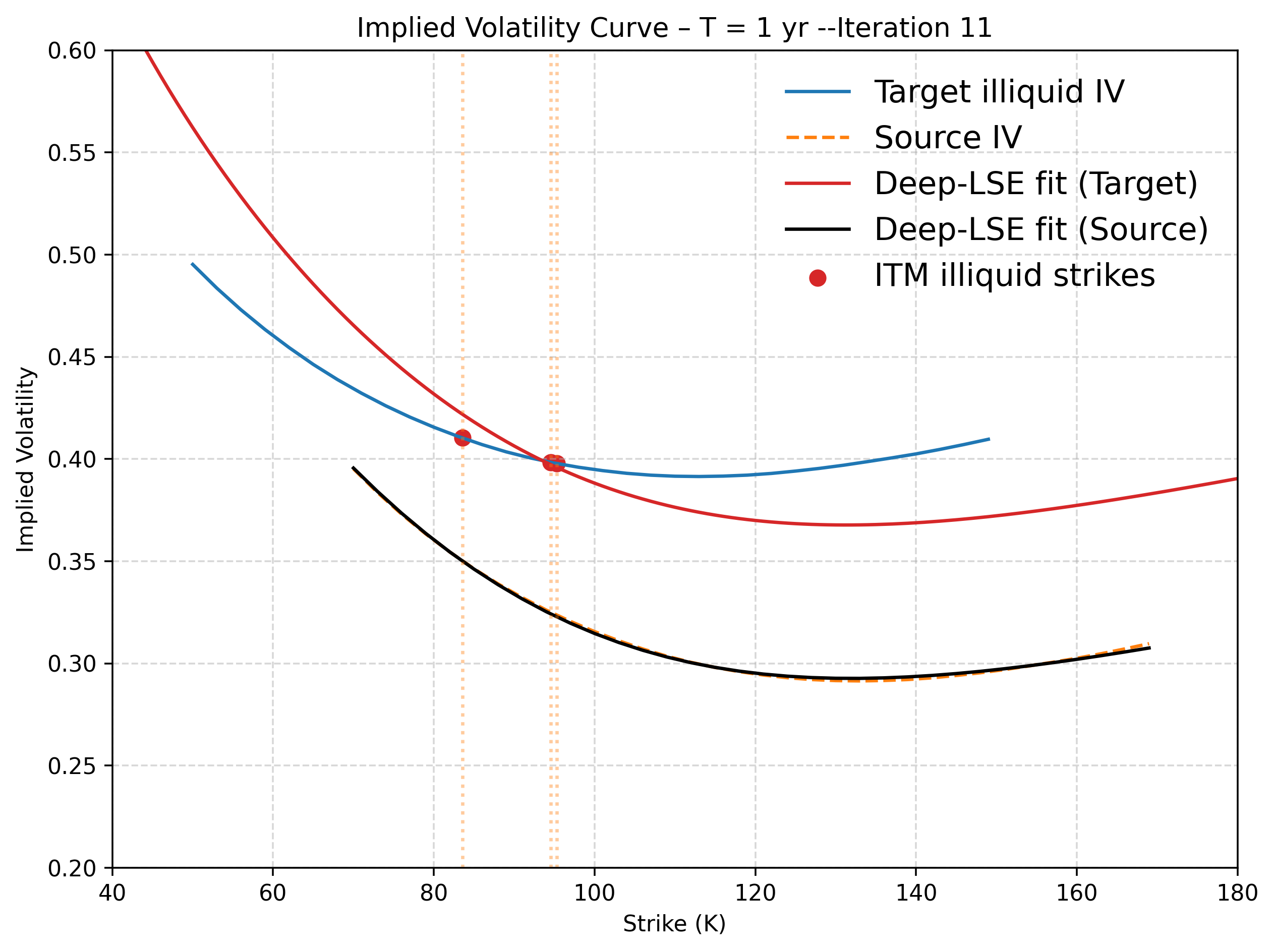}
\end{subfigure}

\medskip

\begin{subfigure}{0.42\textwidth}
  \includegraphics[width=\linewidth]{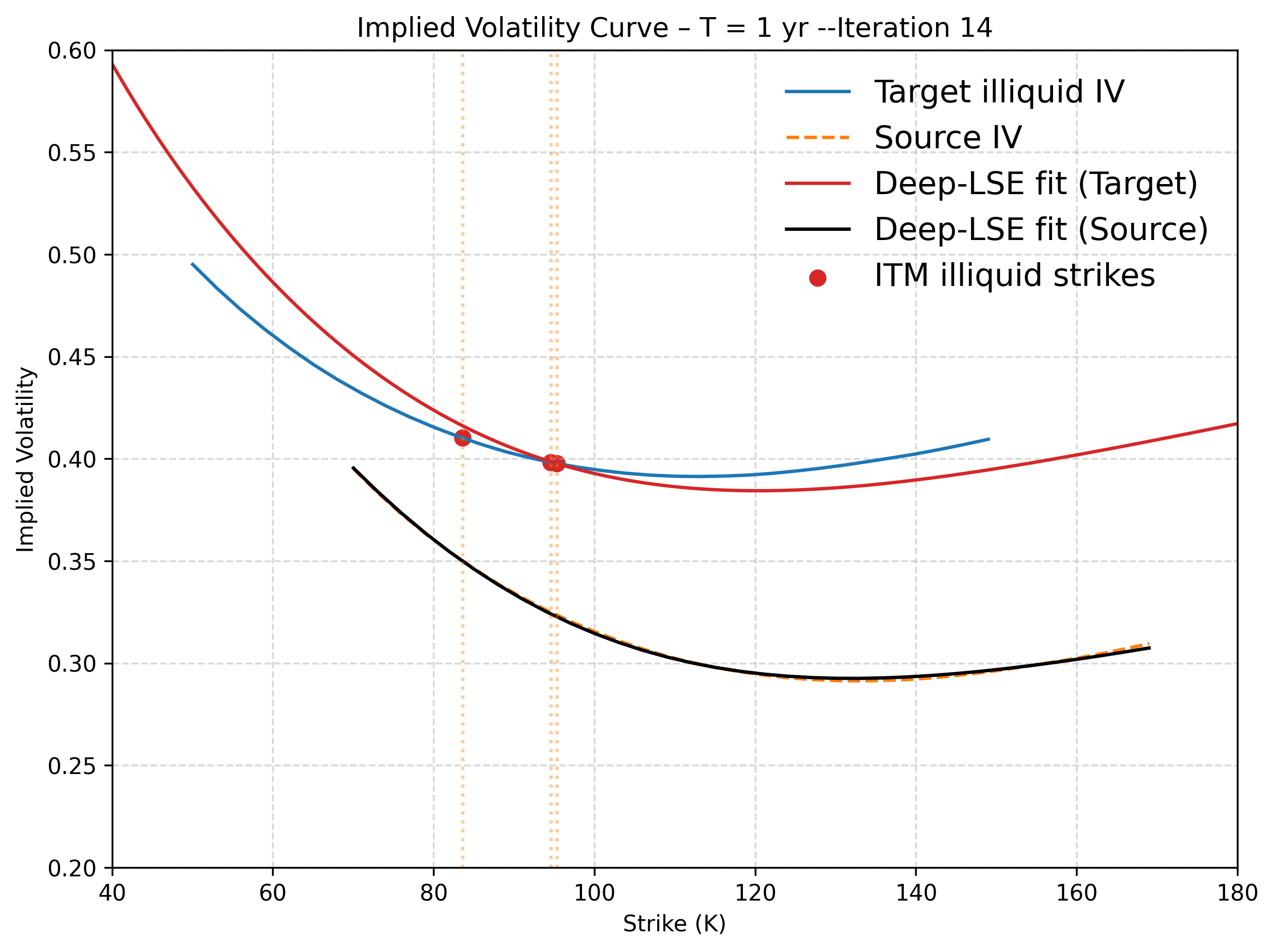}
\end{subfigure}\hspace{1cm}
\begin{subfigure}{0.42\textwidth}
  \includegraphics[width=\linewidth]{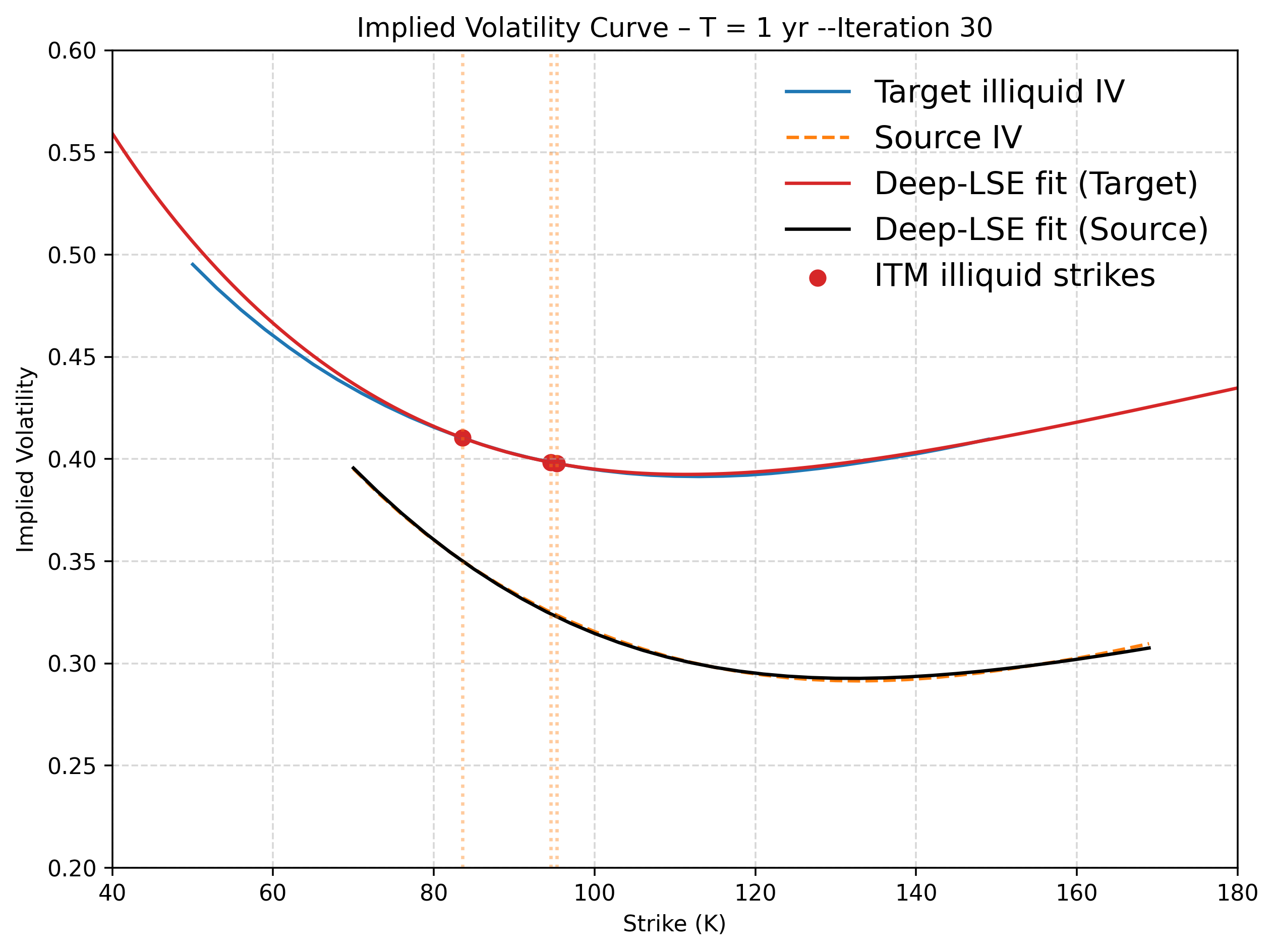}
\end{subfigure}

\caption{Second step recovery - Target Deep-LSE Fit. The model only sees the illiquid (orange) quotes. The blue solid line is the true implied volatility function that the Deep-LSE recovers (green solid line). The solid orange and red curves represent the true and estimated IV curve of the first step. }
\label{fig:train_target_bates}
\end{minipage}
\end{figure}

\section{Empirical Analysis}\label{empa_section}


We test our framework empirically on the SPX option data which consists of January 2015 SPX option data for the source (proxy) and January 2016 SPX option data for the target. 

We emulate the conditions of an illiquid market by censoring the data and using just three random quotes of call options. In particular, this allows us to recover the ground truth from the entire panel of option data, and then compare it with the estimate of the RND from our model, which we train using only three call option quotes with 1 month maturity.

We investigate two forms of severe market illiquidity by randomly selecting three in-the-money call option quotes in the first scenario and three out-of-the-money call option quotes in the second scenario described in Table \ref{tab:strikes_prices_cases}. We emphasize that these three option quotes constitute the only information on the terminal RND available to the models.

\begin{table}[htbp]
    \centering
    \caption{Strikes and prices randomly sampled from in-the-money (ITM) and out-of-the-money (OTM) SPX option quotes to emulate the conditions of an illiquid market.}
    \label{tab:strikes_prices_cases}
    \begin{tabular}{cccc}
        \hline
         & Strike & Price & Implied Volatility (IV)\\
        \hline
         & 1950 & 82.95 & 0.204 \\
        Scenario 1 (ITM) & 1995 & 51.15 & 0.186 \\
         & 2180 & 0.45  & 0.129\\
        \hline
         & 2145 & 1.32  & 0.128\\
        Scenario 2 (OTM) & 2200 & 0.30  & 0.134\\
         & 2230 & 0.125 & 0.137\\
        \hline
    \end{tabular}
\end{table}

Market data is noisy, and we illustrate in Fig. \ref{fig:smooth_empa} the smoothing approach required to recover the ground truth RND. After smoothing option prices, we differentiate twice the Black-Scholes function with respect to the strike to obtain the ground truth RND, which we recover from a dense set of option quotes.

\begin{figure}[H]
    \centering
    \begin{subfigure}{0.8\textwidth}
        \includegraphics[width=\linewidth]{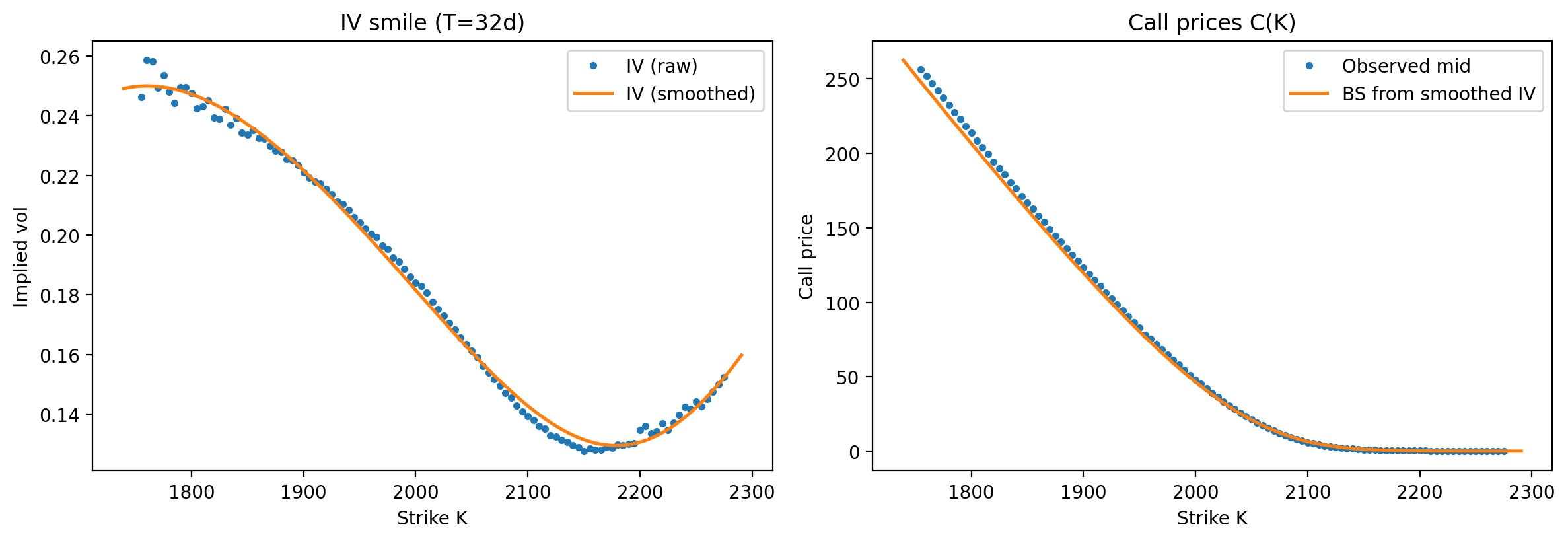} 
    \end{subfigure}
    \caption{Smoothing approach to recover the ground truth RND. On the left panel, the liquid IV curve. On the right, the liquid pricing function for each strike.}
    \label{fig:smooth_empa}
\end{figure}

In Fig. \ref{fig:setup_iv_empa} we illustrate the implied volatility curve of the proxy (blue curve) against the target implied volatility curve of the illiquid market (orange curve), where the green crosses represent the three illiquid strikes identified in Scenario 1. This setup with real option market data is particularly challenging because the proxy source data is heavily convex, while the target data comes from an implied volatility curve that is not convex.

\begin{figure}[H]
    \centering
    \begin{subfigure}{0.6\textwidth}
        \includegraphics[width=\linewidth]{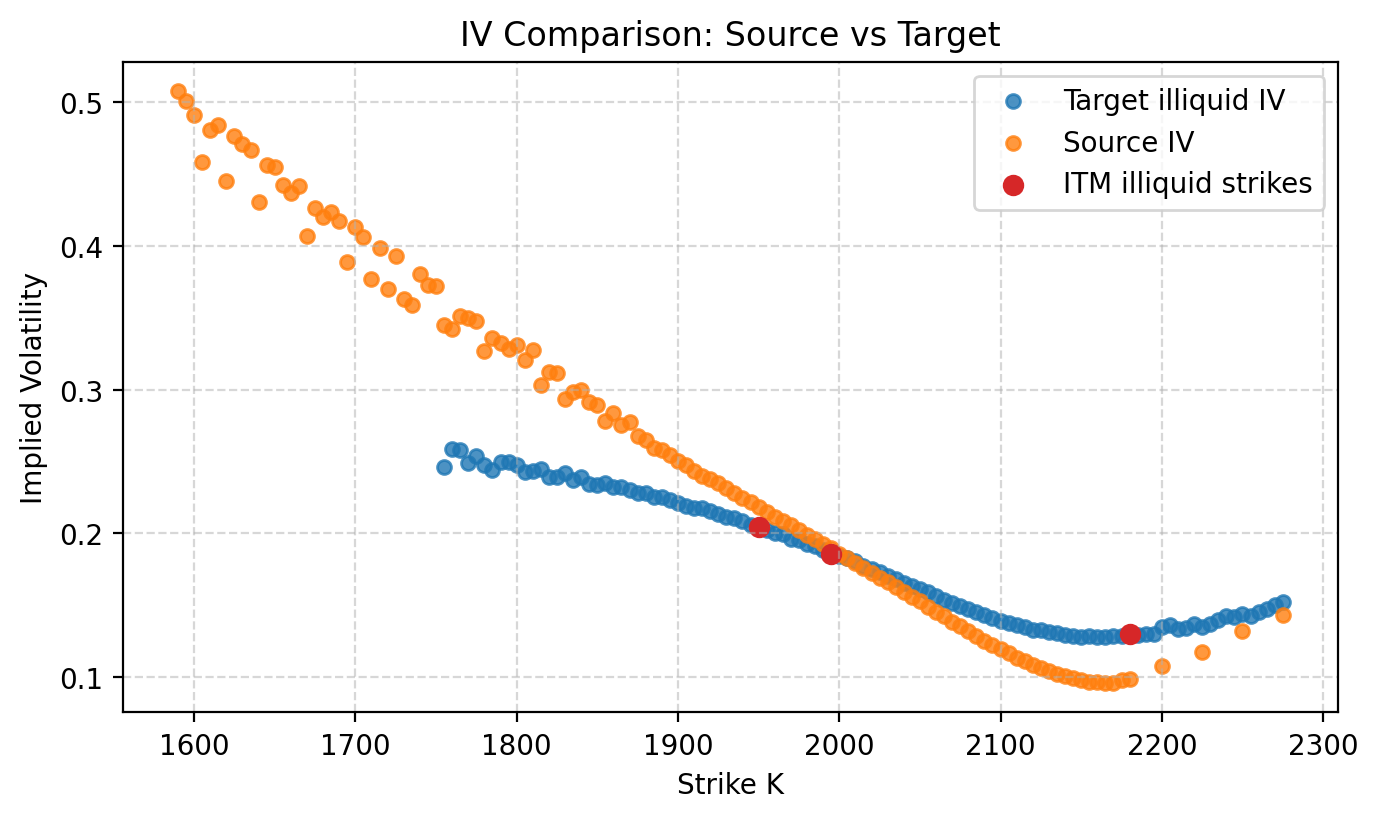} 
    \end{subfigure}
    \caption{Scenario 1 - Setup of implied volatility curves of the empirical analysis on SPX data. In orange, the illiquid target implied volatility. In blue, the liquid proxy we use to train the model. The green crosses are the illiquid strikes.}
    \label{fig:setup_iv_empa}
\end{figure}

Similarly to the simulation studies, we use the Deep-LSE model with 2 layers and 3 affine terms each, as we find it strikes the right balance between flexibility and performance. In Fig. \ref{fig:final_rnd_empa} we illustrate the estimates of the illiquid RND, comparing the Deep-LSE model versus quadratic splines. It is possible to observe that the Deep-LSE model produces a tight recovery of the illiquid RND, particularly in the strike range of $2000-2100$. By contrast, the recovery obtained from quadratic splines substantially deviates from the ground truth illiquid RND, particularly in right tail of the distribution. In Appendix \ref{empa_appendix}, we report the details of the training process of Scenario 1, illustrating the estimation of the liquid proxy and the transfer to the illiquid target.
\begin{figure}[H]
    \centering
    \begin{subfigure}{0.75\textwidth}
        \includegraphics[width=\linewidth]{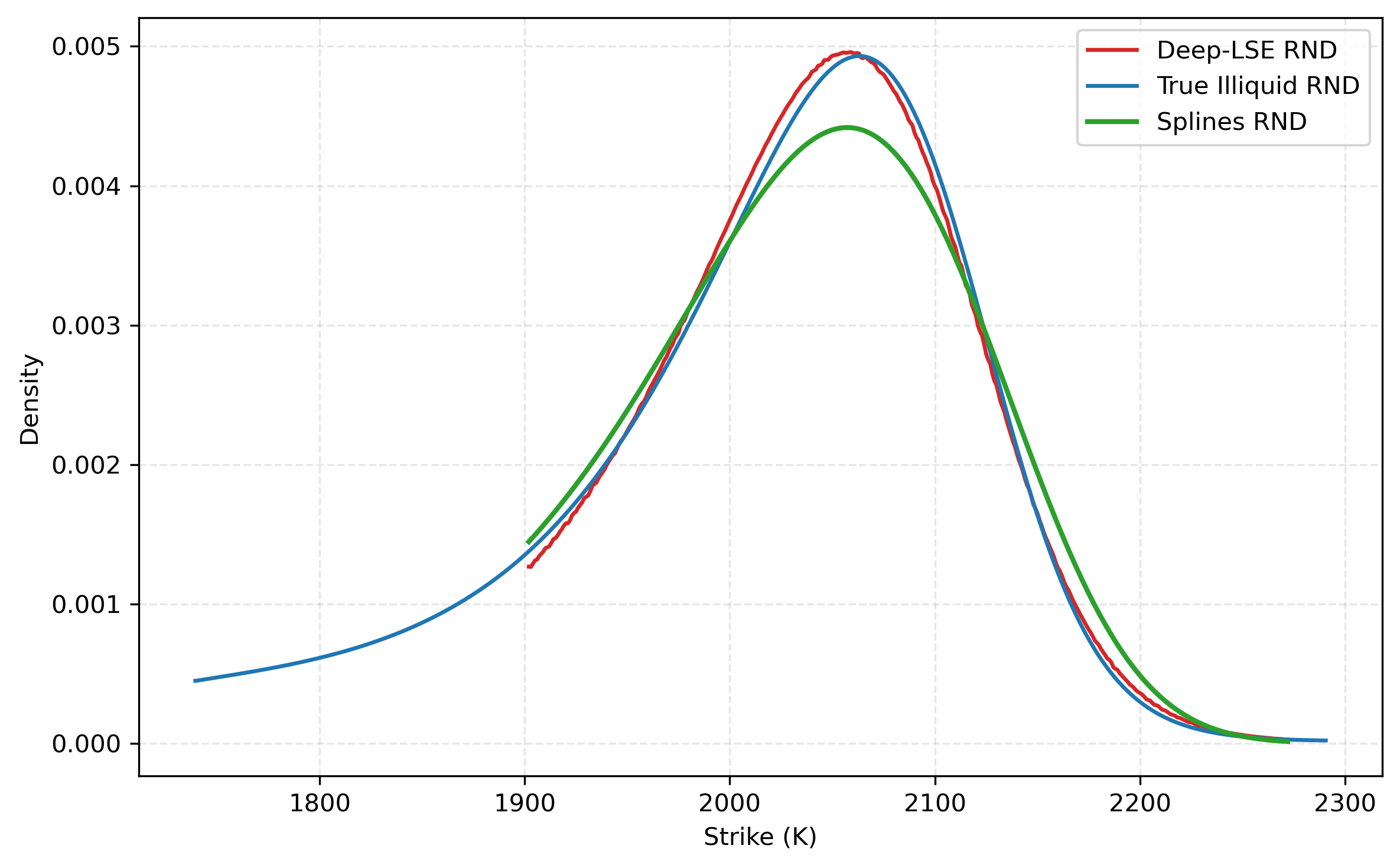} 
    \end{subfigure}
    \caption{Scenario 1 - Illiquid RND recovery of Deep-LSE (blue solid curve) and quadratic splines (green solid line) in comparison with the illiquid ground truth RND (orange dotted line).}
    \label{fig:final_rnd_empa}
\end{figure}
After the Deep-LSE model recovers the target implied volatility function of the target illiquid market, we compute the theoretical price of the option quotes on different strikes. Empirically, the model respects no-arbitrage constraints, as the pricing function of call options is monotone and convex.

We illustrate the validity of our framework under different illiquid market conditions by performing the same empirical analysis on a different set of illiquid strikes (Scenario 2) (Fig. \ref{fig:empa2_strike_+_rnd}). In the first experiment on 2015-2016 SPX data we randomly sample from the left tail consisting of in-the-money option quotes. From the same data, now we randomly sample on the right tail to test out-of-money strikes. The left panel of Fig. \ref{fig:empa2_strike_+_rnd} illustrates the three strikes sampled that represent illiquid market observations.
\begin{figure}[H]
\centering
\begin{minipage}{\textwidth}
\begin{subfigure}{0.55\textwidth}
  \includegraphics[width=\linewidth]{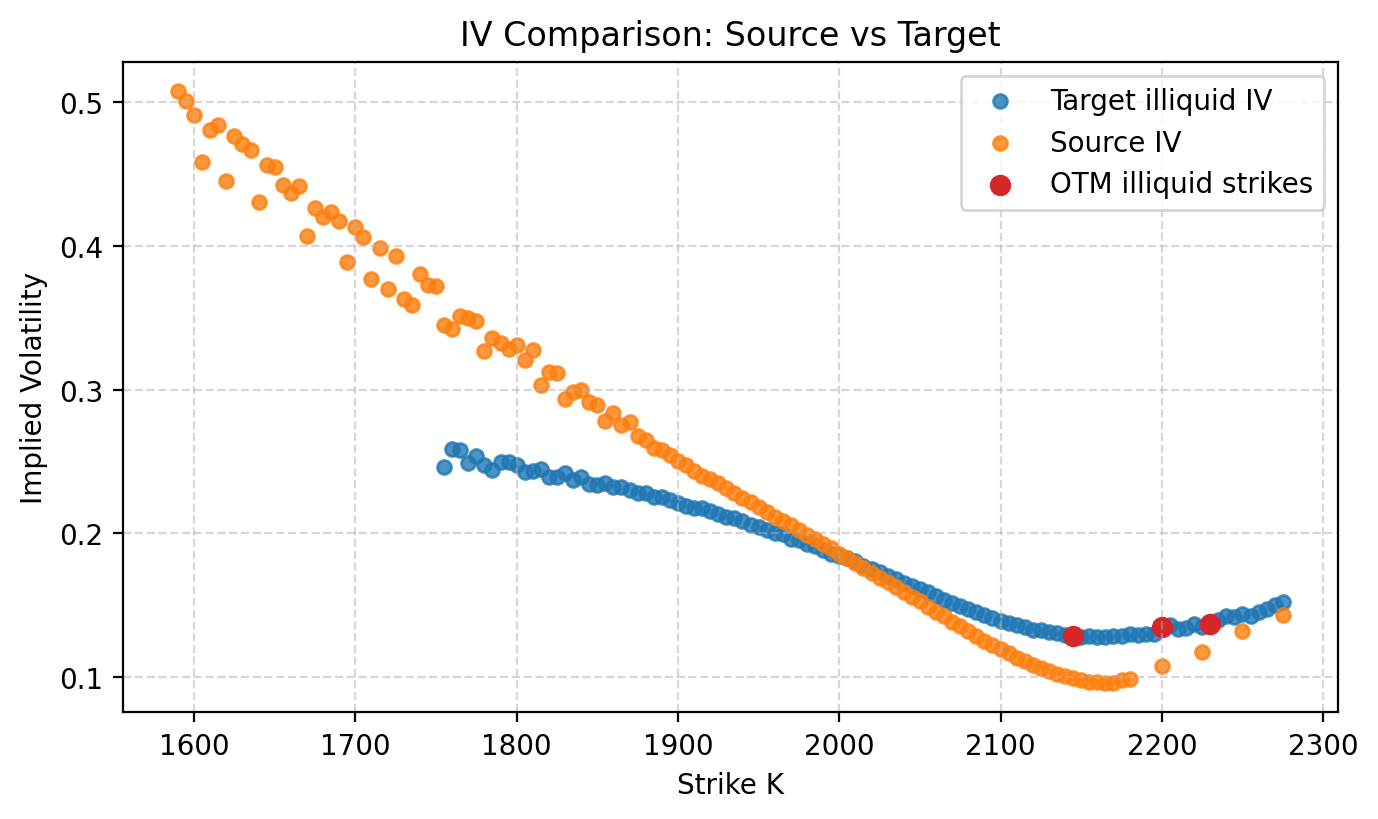}
\end{subfigure}\hfill
\begin{subfigure}{0.55\textwidth}
  \includegraphics[width=\linewidth]{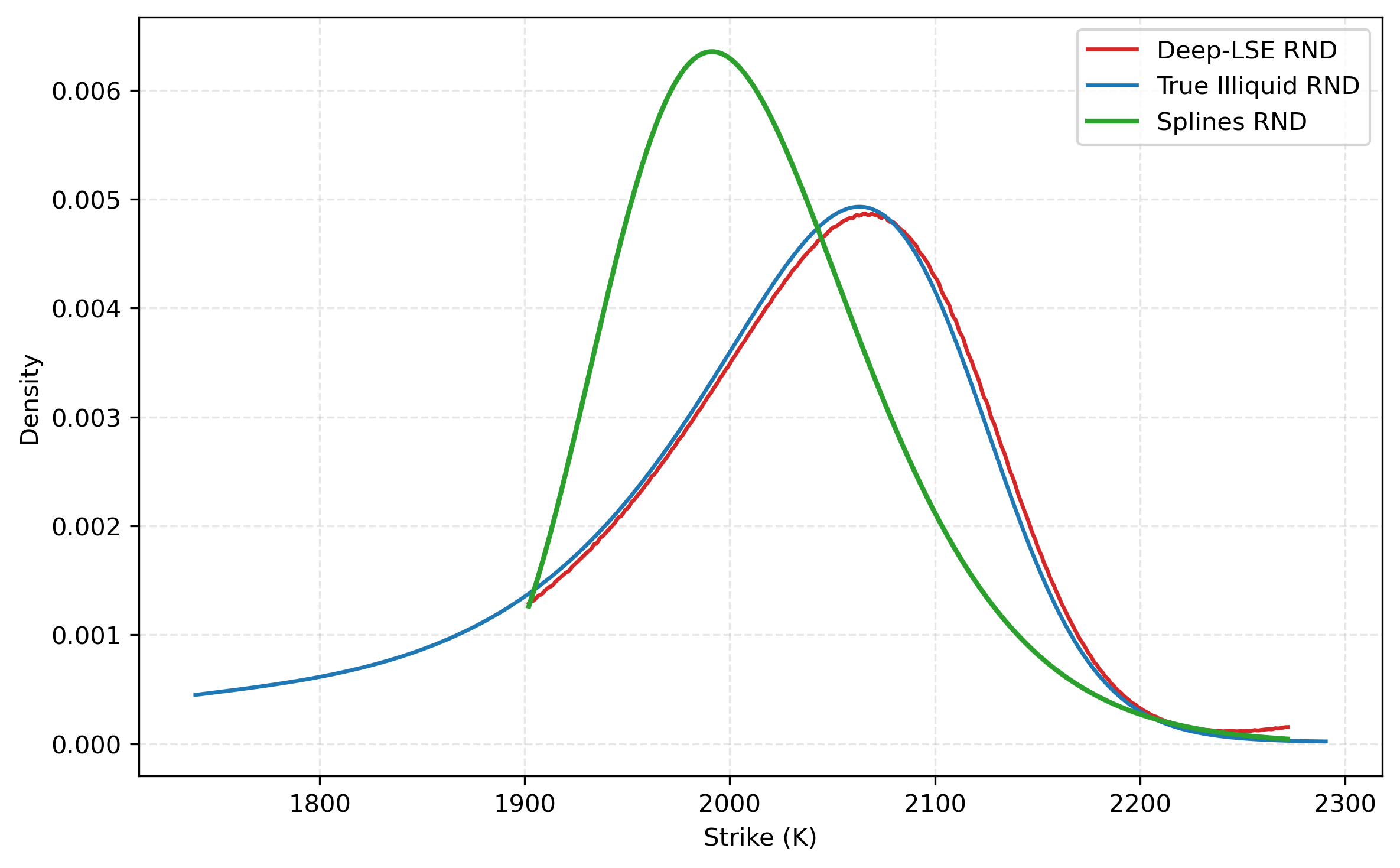}
\end{subfigure}
\caption{Scenario 2 - Setup of implied volatility curves of the empirical analysis on SPX data. On the left panel, strikes selected (green crosses) on the illiquid target (orange) implied volatility curve. On the right panel, illiquid RND recovery of Deep-LSE model (blue solid line) and quadratic splines (green solid line) in comparison with the illiquid ground truth RND (orange dotted line).}
\label{fig:empa2_strike_+_rnd}
\end{minipage}

\end{figure}

In this Scenario, the Deep-LSE model shows an even better recovery with respect to Scenario 1. The right panel of Fig. \ref{fig:empa2_strike_+_rnd} plots the recovery of the illiquid RND for the Deep-LSE model (blue distribution) and quadratic splines (green distribution), comparing them to the ground truth RND (orange dotted distribution).

In Fig. \ref{fig:concave_fit_splines}, we highlight an additional benefit of our model, which guarantees convexity. With noisy and illiquid market quotes, traditional methods might produce concave fits for the implied volatility curve. By contrast, our framework adjusts robustly to illiquid conditions and still recovers a well–behaved RND. This is driven by the pretraining phase, during which they internalize the typical shape of the implied volatility surface and, by extension, of the corresponding risk–neutral density.

\begin{figure}[H]
    \centering
    \begin{subfigure}{0.6\textwidth}
        \includegraphics[width=\linewidth]{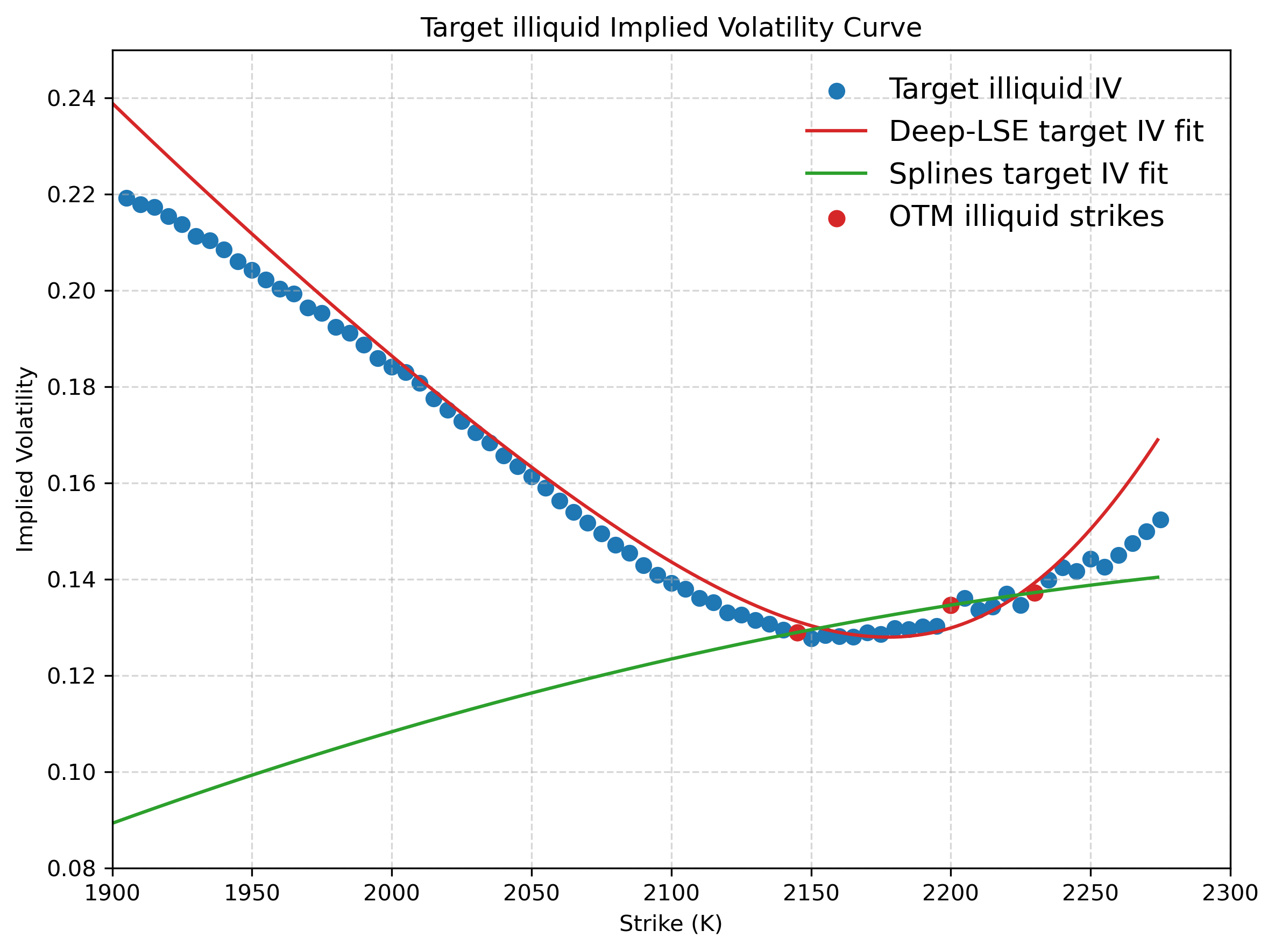} 
    \end{subfigure}
    \caption{Scenario 2 - Market noise and illiquidity might lead to a concave fit of quadratic splines in recovering the implied volatility curve of the illiquid market. The blue solid line is the IV curve of Deep-LSE, the orange solid line is the IV curve of quadratic splines, and blue dots constitute the target IV curve. The red dots are the illiquid strikes.}
    \label{fig:concave_fit_splines}
\end{figure}

Furthermore, we recover option prices from the estimated RND and contrast them with market quotes, from which we infer the associated pricing errors. Let $(S_t)_{t\ge0}$ be the underlying price process under the risk–neutral measure $\mathbb{Q}$, and assume a constant continuously compounded interest rate $r$. Let $\hat{f}_{\mathbb{Q}}(s;T)$ denote the estimated risk–neutral density of $S_T$ on $(0,\infty)$. The estimated price of a European call, in $t=0$, with maturity $T$ and strike $K$ is
\begin{equation}
  \hat{C}(0;K,T)
  = e^{-rT}\,\mathbb{E}^{\mathbb{Q}}\!\left[(S_T-K)^+\right]
  = e^{-rT}\int_K^{\infty} (s-K)\,\hat{f}_{\mathbb{Q}}(s;T)\,\mathrm{d}s 
\end{equation}
and we define the pricing error as $C(0;K,T)-\hat{C}(0;K,T)$, where $C(0;K,T)$ is the price observed on the market. In Table \ref{tab:PE_case1_case2}, we report the pricing error of popular methods to estimate the RND. Specifically, the Kernel-based nonparametric RND estimator replicates the local polynomial approach of \textcite{ait2003nonparametric}, the Lognormal-Weibull Mixture implements the approach \textcite{li2024parametric}, and Maximum-Entropy consists in the Positive Convolution method of \textcite{bondarenko2003estimation}. We infer the RND in all cases by using the same sets of illiquid quotes illustrated in Table \ref{tab:strikes_prices_cases}. The evaluation grid consists of equispaced strikes in the range $1900-2150$, while the Mean Absolute Error (MAE) computes the average absolute error across all strikes.

In Table \ref{tab:PE_case1_case2}, we observe the absolute pricing errors of Scenario 1, which corresponds to the scenario in which only three in-the-money option quotes are available on the market. In this extreme illiquid condition, the benchmarks are not able to fit the RND and extrapolate a meaningful pricing function. Among all benchmarks, our Deep-LSE is the model reports the lowest MAE and the best fit across the strike grid. Parametric models such as the Lognormal-Weibull Mixture, Normal and Lognormal approaches perform better than nonparametric methods such as the kernel-based and Maximum-Entropy approach. Similarly, we observe the absolute pricing errors of Scenario 2, which corresponds to the scenario in which only three out-of-the-money option quotes are available on the market. The Deep-LSE model presents the lowest MAE and absolute pricing error in the sample grid. Also in this case, the parametric model outperforms nonparametric approaches.

\begin{table}[H]
\centering
\caption{Absolute Pricing Error - Scenario 1 and Scenario 2.}
\label{tab:PE_case1_case2}

\begin{tabular}{l*{7}{c}}
\hline
\multicolumn{8}{c}{{Scenario 1: In-the-money illiquid option quotes}} \\
\hline
 & 1900 & 1950 & 2000 & 2050 & 2100 & 2150 & MAE \\
\hline
Kernel-based nonparametric  & 14.56 & 13.56 & 10.82 & 7.74 & 5.19 & 1.35 & 8.87 \\
Deep-LSE                    & 0.20  & 0.69  & 0.26  & 0.28 & 0.45 & 0.27 & 0.53 \\
Lognormal-Weibull Mixture   & 0.50  & 0.79  & 0.61  & 1.10 & 1.09 & 0.45 & 0.76 \\
Maximum-Entropy             & 5.60  & 0.51  & 0.17  & 3.08 & 3.30 & 0.84 & 2.61 \\
Parametric Lognormal        & 1.07  & 2.95  & 2.39  & 0.29 & 2.04 & 1.43 & 1.70 \\
Parametric Normal           & 1.37  & 2.83  & 2.12  & 0.41 & 1.95 & 1.30 & 1.67 \\
Quadratic Splines           & 0.19  & 0.57  & 1.30  & 1.96 & 1.64 & 0.56 & 1.03 \\
\hline
\multicolumn{8}{c}{{Scenario 2: Out-of-the-money illiquid option quotes}} \\
\hline
 & 1900 & 1950 & 2000 & 2050 & 2100 & 2150 & MAE \\
\hline
Kernel-based nonparametric  & 18.21 & 13.43 & 2.16  & 11.73 & 12.91 & 7.77 & 11.04 \\
Deep-LSE                    & 3.57  & 3.30  & 2.83  & 2.10  & 1.14  & 0.45 & 2.23 \\
Lognormal-Weibull Mixture   & 9.64  & 16.07 & 19.18 & 11.97 & 3.36  & 0.32 & 10.94 \\
Maximum-Entropy             & 112.11& 102.08& 87.45 & 63.89 & 29.67 & 0.18 & 65.70 \\
Parametric Lognormal        & 20.60 & 28.47 & 42.27 & 22.40 & 6.90  & 1.35 & 29.80 \\
Parametric Normal           & 52.01 & 41.62 & 26.93 & 12.45 & 3.23  & 0.26 & 22.75 \\
Quadratic Splines           & 14.31 & 18.38 & 16.79 & 9.21  & 2.43  & 0.80 & 10.20 \\
\hline
\end{tabular}

\par\vspace{3mm}
{\small\emph{Notes.} Results indicate, in unit of dollar, the absolute pricing error. We infer the RND by using three illiquid quotes of Scenario 1 and Scenario 2, as illustrated in Table \ref{tab:strikes_prices_cases}. We evaluate on evenly spaced strikes between 1900 and 2150, and report the Mean Absolute Error (MAE) as the average absolute pricing error across all strikes.}
\end{table}

\section{Concluding Remarks}

We address the estimation of the risk-neutral density under illiquid market conditions. Such conditions arise frequently, not only in illiquid option markets, but also for options written on traded equities, and they pose a major challenge for reliable RND estimation. The Deep-LSE model overcomes these difficulties by learning the shape, location, and general form of the implied volatility function of a proxy and liquid market. It then transfers this knowledge to estimate the implied volatility function of the target illiquid market. The simulation study and empirical analysis show that the Deep-LSE model recovers the RND with as few as three option quotes, and yields the lowest pricing error with respect to popular RND estimation methods.

\printbibliography

\newpage

\appendix

\begin{center}
    {\LARGE  Appendix of "Transfer Learning (Il)liquidity" \par}
    \vspace{2cm}

    \begin{minipage}{0.85\textwidth}
        \small

        In Appendix \ref{theory_appendix}, we show additional theoretical details related to the Deep-LSE model. Appendix \ref{proofs_appendix} illustrates the proofs related to the convexity property of the Deep-LSE model (Appendix \ref{convexity_appendix}), the link between the Deep-LSE model and the class of max-affine functions (Appendix \ref{bounds_appendix}), and in Appendix \ref{uat_appendix} we prove the Universal Approximation theorem of the Deep-LSE for convex functions. We conclude the theoretical framework by illustrating in Appendix \ref{sieve_appendix} proofs related to Sieve M-estimation and consistency. Regarding the simulation studies (Appendix \ref{more_sim}), we test the Deep-LSE in illiquid markets using the data-generating process of Kou-Heston (Appendix \ref{kou_heston_appendix}), Andersen-Benzoni-Lund (Appendix \ref{abl_appendix}), and the Three-Factor Double Exponential (Appendix \ref{3fde_appendix}). In the end, we illustrate in Appendix \ref{empa_appendix} additional details regarding the training process of the Deep-LSE model on SPX market data. 
        
    \end{minipage}
\end{center}

\newpage

\section{Architecture of the Deep-LSE}\label{theory_appendix}

We present, as an example, the architecture of the Deep-LSE estimator with two layers. Let \(x\in\mathbb{R}^d\) denote the input and define with \(K_1, K_2\in\mathbb{N}\) the number of affine functions, that play the role of the number of neurons, in the first and second LSE layer, respectively. Define the temperature parameters as the scalar value \(T_1,T_2>0\).

For each \(k=1,\dots,K_1\) in the first layer, define one affine function
\[
\ell^{(1)}_k(x)
\;=\;
a^{(1)}_{k}{}^{\!\top}x + b^{(1)}_k,
\qquad
a^{(1)}_k\in\mathbb{R}^d,\; b^{(1)}_k\in\mathbb{R},
\]
which in scalar form becomes an affine base
\[
L^{(1)}(x)
=\begin{bmatrix}
\ell^{(1)}_1(x)\\ \vdots\\ \ell^{(1)}_{K_1}(x)
\end{bmatrix}
\in\mathbb{R}^{K_1}.
\]
The output of the first layer of the Deep-LSE estimator is
\[
z_1(x)
\;=\;
\operatorname{LSE}_{T_1}\!\big(L^{(1)}(x)\big)
\;=\;
T_1\log\!\Big(\sum_{k=1}^{K_1} e^{\,\ell^{(1)}_k(x)/T_1}\Big)
\in\mathbb{R}.
\]
We can express this also in matrix form as
\[
L^{(1)}(x)=A^{(1)}x+b^{(1)},
\qquad
z_1(x)=\operatorname{LSE}_{T_1}\!\big(A^{(1)}x+b^{(1)}\big).
\]
Regarding the second layer, for each neuron \(k=1,\dots,K_2\), we have a second set affines linear functions
\[
\ell^{(2)}_k(x)
\;=\;
a^{(2)}_{k}{}^{\!\top}x + b^{(2)}_k,
\qquad
a^{(2)}_k\in\mathbb{R}^d,\; b^{(2)}_k\in\mathbb{R}
\]
and define with \(\alpha_k > 0\) the positive recursion weight, which is a positive scalar as a consequence of the softplus function
\[
\alpha_k \;=\; \operatorname{softplus}(\eta_k),
\qquad
\eta_k\in\mathbb{R}.
\]
Then, we have the vector of the second layer that collects the recursive term and the skip connection 
\[
S^{(2)}(x)
\;=\;
\begin{bmatrix}
\alpha_1 z_1(x) + \ell^{(2)}_1(x)\\
\vdots\\
\alpha_{K_2} z_1(x) + \ell^{(2)}_{K_2}(x)
\end{bmatrix}
\;=\;
\alpha\, z_1(x) + L^{(2)}(x)
\;\in\;\mathbb{R}^{K_2},
\]
where \(\alpha=(\alpha_1,\ldots,\alpha_{K_2})^\top\in\mathbb{R}^{K_2}_{> 0}\), and
\[
L^{(2)}(x)
=\begin{bmatrix}
\ell^{(2)}_1(x)\\ \vdots\\ \ell^{(2)}_{K_2}(x)
\end{bmatrix}
= A^{(2)}x + b^{(2)},
\quad
A^{(2)}\in\mathbb{R}^{K_2\times d},\;
b^{(2)}\in\mathbb{R}^{K_2}.
\]
The output of the estimator at the second layer is the result of the log-sum-exp function over the inner vector of the second layer
\[
y(x)
\;=\;
\operatorname{LSE}_{T_2}\!\big(S^{(2)}(x)\big) + c_{\text{out}}
\;=\;
T_2 \log\!\Big(\sum_{k=1}^{K_2} e^{\,(\alpha_k z_1(x)+\ell^{(2)}_k(x))/T_2}\Big) + c_{\text{out}}
\in\mathbb{R}.
\]
All together, we express the output of the Deep-LSE estimator with 2 layers as
\[
\begin{aligned}
z_1(x)
&= T_1 \log\!\Big(\sum_{i=1}^{K_1} e^{\,\big(a^{(1)}_{i}{}^{\!\top}x+b^{(1)}_i\big)/T_1}\Big),\\[2mm]
y(x)
&= T_2 \log\!\Bigg(\sum_{k=1}^{K_2}
\exp\!\Big(
\frac{\alpha_k\, z_1(x) + a^{(2)}_{k}{}^{\!\top}x + b^{(2)}_k}{T_2}
\Big)\Bigg)
\;+\; c_{\text{out}}.
\end{aligned}
\]

\newpage

\section{Proofs}\label{proofs_appendix}

\subsection{Convexity}\label{convexity_appendix}

\bigskip
\begin{proof}[Proof of Lemma \ref{epigraph_lemma}]   
Set $f(x)=(f_1(x),\dots,f_m(x))$ and $\mathrm{epi}\,h=\{(u,t):\,h(u)\le t\}$, which is convex because $h$ is convex. Since $h$ is nondecreasing, for any $u\le v$ and any $t$ with $h(v)\le t$ we also have 
\[
\{(x,t):\,h(f(x))\le t\}\;=\;\{(x,t):\,\exists\,u\ \text{s.t.}\ u\ge f(x)\ \text{and}\ (u,t)\in\mathrm{epi}\,h\}.
\]
The term on the right is the projection of the set $\{(x,u,t):\ u\ge f(x),\ (u,t)\in\mathrm{epi}\,h\}$, and is convex $u\mapsto u-f(x)$ when each $f_i$ is convex. Therefore the epigraph of $x\mapsto h(f(x))$ is convex.
\end{proof}

\bigskip
\begin{proof}[Proof of Theorem \ref{convexity2}]
For $u\in\mathbb{R}^m$ and $T>0$, the log-sum-exp function $\operatorname{LSE}_T(u)=T\log\!\sum_{i=1}^m e^{u_i/T}$ is convex and its gradient is the softmax $p_i(u/T)\ge0$ which has a nonnegative Hessian. Hence, it is a function that preserves convexity (\cite{calafiore2019log}).

By induction, we take the output of the first layer $z^{(1)}(x)=\operatorname{LSE}_{T_1}(A^{(1)}x+b^{(1)})$ and we know that it is convex.

Then, we assume that each succeeding $z^{(\ell-1)}$ is convex. For each $k\le K_\ell$ set $s^{(\ell)}_k(x)=\alpha^{(\ell)}_k z^{(\ell-1)}(x)+a^{(\ell)}_k{}^\top x+b^{(\ell)}_k$. Now, $s^{(\ell)}_k$ is convex because it is a composition of an affine function, convex transformation, the addition of a bias that translates without impact on convexity. In addition, we define the recursive weight to be strictly positive and we know that $z^{(\ell-1)}$ is convex.

Since $\operatorname{LSE}_{T_\ell}$ is convex, it implies $z^{(\ell)}(x)=\operatorname{LSE}_{T_\ell}\big(s^{(\ell)}_1(x),\ldots,s^{(\ell)}_{K_\ell}(x)\big)$ is convex, and by induction, all $z^{(\ell)}$ are convex. Adding the constant $c_{\mathrm{out}}$ does not affect convexity, so $y$ is convex.
\end{proof}

\bigskip
\begin{proof}[Proof of Example \ref{convexity}]
By Theorem \ref{convexity2}, $z_1(x)=\operatorname{LSE}_{T_1}(L^{(1)}(x))$ is convex. When it comes to the second layer, for $k=1,\dots,K_2$ we define the vector
\[
s_k(x)\;=\;\alpha_k\,z_1(x)+\ell^{(2)}_k(x)
\quad\text{with}\quad
\ell^{(2)}_k(x)=a^{(2)}_k{}^\top x+b^{(2)}_k.
\]
Since $\alpha_k> 0$ and $z_1$ is convex, $\alpha_k z_1$ is convex, $\ell^{(2)}_k$ is an affine linear function and hence $s_k$ is convex because this vector is a composition of the recursive term with positive weight and an affine linear function, so convexity still holds. So the consequence is that the vector $S^{(2)}(x)=(s_1(x),\dots,s_{K_2}(x))^\top$ that collects the terms of the second layer is convex because each component is convex.

Similarly, for $h(u)=\operatorname{LSE}_{T_2}(u)$ is convex. Using Lemma \ref{epigraph_lemma} with $h=\operatorname{LSE}_{T_2}$ and $f_k=s_k$ yields that $x\mapsto \operatorname{LSE}_{T_2}(S^{(2)}(x))$ is convex. Adding the constant bias $c_{\mathrm{out}}$ preserves convexity.
\end{proof}

\subsection{Bounds}\label{bounds_appendix}

\bigskip
\begin{proof}[Proof of Theorem \ref{thm:deep_LSE_vs_max}]
For any $T>0$, any $K\in\mathbb{N}$, and any $v=(v_1,\dots,v_K)\in\mathbb{R}^K$,
\begin{equation}\label{bound_calafiore}
\max_{i} v_i\ \le\ \operatorname{LSE}_T(v)\ \le\ \max_{i} v_i\ +\ T\log K,
\end{equation}
as in (\cite{calafiore2019log}).
Now, we define the $K_\ell$-vector
\[
u^{(\ell)}(x)\ =\ \big(u^{(\ell)}_1(x),\dots,u^{(\ell)}_{K_\ell}(x)\big),\qquad
u^{(\ell)}_k(x)\ =\ \alpha^{(\ell)}_k\,z^{(\ell-1)}(x)+\ell^{(\ell)}_k(x)\quad( \text{for } \ \ell\ge2),
\]
and for $\ell=1$ we simply write $u^{(1)}_k(x)=\ell^{(1)}_k(x)$. Then, we construct our known estimator $z^{(\ell)}(x)\ =\ \operatorname{LSE}_{T_\ell}\!\big(u^{(\ell)}(x)\big)$ and the max-affine surrogate as $\bar z^{(\ell)}(x)\ =\ \max_{k\in K_\ell} u^{(\ell)}_k(x)$ and in the recursion we use ${\,z^{(\ell-1)}\ =\ \bar z^{(\ell-1)}}$.
The theorem states that for each $\ell=1,\dots,L$, the bound Eq. \ref{eq:layerwise_sandwich} holds with $\Delta_\ell$ as in Eq. \ref{eq:Delta_recursion}. 

We prove this by induction on $\ell$, since we are generalizing to $L$ layers. 
For the base case $\ell=1$ we have $u^{(1)}(x)=(\ell^{(1)}_k(x))_{k=1}^{K_1}$, so the estimator at layer one is $z^{(1)}(x)=\operatorname{LSE}_{T_1}\!\big(u^{(1)}(x)\big)$ and the deep max-affine is $\bar z^{(1)}(x)=\max_{i\in K_1}\ell^{(1)}_i(x)$.

\noindent Applying Eq. \ref{eq:layerwise_sandwich} with $T=T_1$ and $K=K_1$ gives
\[
\bar z^{(1)}(x)\ \le\ z^{(1)}(x)\ \le\ \bar z^{(1)}(x)+T_1\log K_1,
\]
which is Eq. \ref{eq:layerwise_sandwich} with $\Delta_1=T_1\log K_1$.
By induction, assume Eq. \ref{eq:layerwise_sandwich} holds for some $\ell-1\ge1$, so
\begin{equation}\label{eq:IH2}
\bar z^{(\ell-1)}(x)\ \le\ z^{(\ell-1)}(x)\ \le\ \bar z^{(\ell-1)}(x)+\Delta_{\ell-1}\quad\text{for all }x.
\end{equation}
Regarding the lower bound at layer $\ell$, by definition and left inequality in Eq. \ref{bound_calafiore},
\[
z^{(\ell)}(x)\ =\ \operatorname{LSE}_{T_\ell}\!\big(u^{(\ell)}(x)\big)\ \ge\ \max_{k}u^{(\ell)}_k(x)
\ =\ \max_k\Big(\alpha^{(\ell)}_k\,z^{(\ell-1)}(x)+\ell^{(\ell)}_k(x)\Big).
\]
Using the left inequality in Eq. \ref{eq:IH2} ($z^{(\ell-1)}\ge\bar z^{(\ell-1)}$) and the fact that each $\alpha^{(\ell)}_k > 0$ we obtain
\[
z^{(\ell)}(x)\ \ge\ \max_k\Big(\alpha^{(\ell)}_k\,\bar z^{(\ell-1)}(x)+\ell^{(\ell)}_k(x)\Big)\ =\ \bar z^{(\ell)}(x),
\]
proving the left side of Eq. \ref{eq:layerwise_sandwich} at layer $\ell$.

Regarding the upper bound at layer $\ell$ we apply the right inequality in Eq. \ref{bound_calafiore} to $u^{(\ell)}(x)$
\begin{align}
z^{(\ell)}(x)
  &= \operatorname{LSE}_{T_\ell}\!\big(u^{(\ell)}(x)\big) \notag\\
  &\le T_\ell \log K_\ell \;+\; \max_k u^{(\ell)}_k(x) \notag\\
  &= T_\ell \log K_\ell \;+\; \max_k\!\Big(\alpha^{(\ell)}_k\, z^{(\ell-1)}(x)+\ell^{(\ell)}_k(x)\Big).
  \label{eq:ub_step1}
\end{align}
Insert the upper bound from the induction hypothesis Eq. \ref{eq:layerwise_sandwich} into the rightmost term of Eq. \ref{eq:ub_step1}
\[
\max_k\Big(\alpha^{(\ell)}_k\,z^{(\ell-1)}(x)+\ell^{(\ell)}_k(x)\Big)
\ \le\ \max_k\Big(\alpha^{(\ell)}_k\,\big(\bar z^{(\ell-1)}(x)+\Delta_{\ell-1}\big)+\ell^{(\ell)}_k(x)\Big).
\]
Split the contribution of $\Delta_{\ell-1}$ out of the maximum using $\alpha^{(\ell)}_k\le \alpha^{(\ell)}_{\max}$
\begin{align*}
\max_k\Big(\alpha^{(\ell)}_k\,\big(\bar z^{(\ell-1)}+\Delta_{\ell-1}\big)+\ell^{(\ell)}_k\Big)
&= \max_k\Big(\alpha^{(\ell)}_k\,\bar z^{(\ell-1)}+\ell^{(\ell)}_k\ +\ \alpha^{(\ell)}_k \Delta_{\ell-1}\Big)\\
&\le \max_k\Big(\alpha^{(\ell)}_k\,\bar z^{(\ell-1)}+\ell^{(\ell)}_k\Big)\ +\ \alpha^{(\ell)}_{\max}\,\Delta_{\ell-1}\\
&=\ \bar z^{(\ell)}(x)\ +\ \alpha^{(\ell)}_{\max}\,\Delta_{\ell-1}.
\end{align*}
Combining this result with the Eq. \ref{eq:ub_step1} we obtain
\[
z^{(\ell)}(x)\ \le\ \bar z^{(\ell)}(x)\ +\ {\Big(T_\ell\log K_\ell+\alpha^{(\ell)}_{\max}\Delta_{\ell-1}\Big)},
\]
where the term in brackets is $\Delta_\ell$.
This is the right side of Eq. \ref{eq:layerwise_sandwich} with the recursive definition Eq. \ref{eq:Delta_recursion}. We have proved Eq. \ref{eq:layerwise_sandwich} for layer $\ell$ assuming it for $\ell-1$. Together with the base case this holds for all $\ell=1,\dots,L$.
Developing the linear recursion Eq. \ref{eq:Delta_recursion} gives
\[
\Delta_\ell
= T_\ell\log K_\ell + \alpha^{(\ell)}_{\max}\Delta_{\ell-1}
= T_\ell\log K_\ell + \alpha^{(\ell)}_{\max}\Big(T_{\ell-1}\log K_{\ell-1} + \alpha^{(\ell-1)}_{\max}\Delta_{\ell-2}\Big),
\]
and continue until $\Delta_1=T_1\log K_1$. Collecting terms we obtain
\[
\Delta_\ell
=\sum_{j=1}^{\ell}\Bigg(T_j\log K_j\ \prod_{r=j+1}^{\ell}\alpha^{(r)}_{\max}\Bigg),
\]
which is Eq. \ref{eq:Delta_closed_form}. From this expansion, we can see that the term for the temperature of each layer ($T_j\log K_j$) is amplified by the positive recursion weights. Then the we consider the Deep-LSE estimator $y(x)=z^{(L)}(x)+c_{\mathrm{out}}$ and the deep max-affine surrogate $\bar y(x)=\bar z^{(L)}(x)+c_{\mathrm{out}}$ and we add the same bias to both function to obtain
\[
0\ \le\ y(x)-\bar y(x)\ =\ z^{(L)}(x)-\bar z^{(L)}(x)\ \le\ \Delta_L.
\]
\end{proof}

\begin{proof}[Proof of Example \ref{bounds1}]
In the first layer, by the standard LSE bounds,
\[
\bar z_1(x)\ \le\ z_1(x)\ \le\ \bar z_1(x)+T_1\log K_1. \tag{1}
\]
In the second layer, lower bound, since $\operatorname{LSE}_{T_2}$ is coordinatewise nondecreasing and it is greater than the pointwise maximum,
\begin{align*}
y(x)-c_{\mathrm{out}}
&=\operatorname{LSE}_{T_2}\!\big(\{\alpha_k z_1(x)+\langle A^{(2)}_{k,\cdot},x\rangle+b^{(2)}_k\}_{k=1}^{K_2}\big) \\
&\ge \max_k \big(\alpha_k z_1(x)+\langle A^{(2)}_{k,\cdot},x\rangle+b^{(2)}_k\big).
\end{align*}
Using $z_1(x)\ge\bar z_1(x)$ from (1) and adding $c_{\mathrm{out}}$ gives $y(x)\ge \bar y(x)$.
In the second layer, upper bound, again by the LSE bound,
set, for each $k$, 
\[
u_k(x)=\alpha_k z_1(x)+\langle A^{(2)}_{k,\cdot},x\rangle+b^{(2)}_k,
\]
with $u(x)=(u_k(x))_{k=1}^{K_2}\in\mathbb{R}^{K_2}$.
Then $y(x)-c_{\mathrm{out}}=\operatorname{LSE}_{T_2}(u(x))$, so applying the inequality with $K=K_2$, $T=T_2$ to $u(x)$ yields
\[
\operatorname{LSE}_{T_2}(u(x))
\ \le\ T_2\log K_2+\max_k u_k(x).
\]
Therefore,
\begin{align*}
y(x)-c_{\mathrm{out}}
&\le T_2\log K_2+\max_k \big(\alpha_k z_1(x)+\langle A^{(2)}_{k,\cdot},x\rangle+b^{(2)}_k\big)\\
&\le T_2\log K_2+\max_k \big(\alpha_k(\bar z_1(x)+T_1\log K_1)+\langle A^{(2)}_{k,\cdot},x\rangle+b^{(2)}_k\big) \quad\text{by (1) and } \alpha_k>0\\
&= T_2\log K_2+\alpha_{\max}T_1\log K_1+\max_k \big(\alpha_k\bar z_1(x)+\langle A^{(2)}_{k,\cdot},x\rangle+b^{(2)}_k\big).
\end{align*}
Adding $c_{\mathrm{out}}$ yields the upper bound for $y(x)$.
\end{proof}

\subsection{Universal Approximation Theorem}\label{uat_appendix}

\begin{proof}[Proof of Theorem \ref{ybar_pwma}]
We proceed by induction on the depth $\ell$ to establish the max–affine representation of the theorem and the recursive formulas for $(A_p^{[\ell]},b_p^{[\ell]})$.

\noindent For the base case $\ell=1$, by definition,
\[
\bar z^{(1)}(x)=\max_{k_1\in\{1,\dots,K_1\}}\bigl(\langle A^{(1)}_{k_1},x\rangle+b^{(1)}_{k_1}\bigr),
\]
so the statement holds with $P_1=\{1,\dots,K_1\}$ and $A^{[1]}_{k_1}=A^{(1)}_{k_1}$, $b^{[1]}_{k_1}=b^{(1)}_{k_1}$.

\noindent Regarding the induction step, we assume that for some $\ell\ge2$ that
\[
\bar z^{(\ell-1)}(x)=\max_{p\in P_{\ell-1}}\bigl(\langle A^{[\ell-1]}_p,x\rangle+b^{[\ell-1]}_p\bigr),
\]
and so we have that
\[
\alpha^{(\ell)}_k\,\bar z^{(\ell-1)}(x)
=\max_{p\in P_{\ell-1}}\bigl(\alpha^{(\ell)}_k\langle A^{[\ell-1]}_p,x\rangle+\alpha^{(\ell)}_kb^{[\ell-1]}_p\bigr).
\]
Therefore,
\begin{align*}
\bar z^{(\ell)}(x)
&=\max_{k}\Bigl(\alpha^{(\ell)}_k\,\bar z^{(\ell-1)}(x)+\langle A^{(\ell)}_k,x\rangle+b^{(\ell)}_k\Bigr)\\
&=\max_{k}\ \max_{p\in P_{\ell-1}}
\Bigl(\big\langle \alpha^{(\ell)}_k A^{[\ell-1]}_p + A^{(\ell)}_k,\ x\big\rangle
+ \alpha^{(\ell)}_k b^{[\ell-1]}_p + b^{(\ell)}_k\Bigr)\\
&=\max_{(p,k)\in P_{\ell-1}\times\{1,\dots,K_\ell\}}
\bigl(\langle A^{[\ell]}_{(p,k)},x\rangle+b^{[\ell]}_{(p,k)}\bigr),
\end{align*}
where $A^{[\ell]}_{(p,k)}=\alpha^{(\ell)}_k A^{[\ell-1]}_p + A^{(\ell)}_k$ and $b^{[\ell]}_{(p,k)}=\alpha^{(\ell)}_k b^{[\ell-1]}_p + b^{(\ell)}_k$. So this representation holds at depth $\ell$ with index set $P_\ell=P_{\ell-1}\times\{1,\dots,K_\ell\}$, and we obtain that the deep max-affine surrogate is still a max of finite affine functions.

\end{proof}

\begin{proof}[Proof of Theorem \ref{thm:uniform_approx_lse_streamlined}]
To establish the proof, we use two main results: the deep-LSE bound that links our model to max-affine functions, and the density result that links max-affine functions with any convex function.

The bound (proved in Theorem \ref{thm:deep_LSE_vs_max}) ensures that the Deep-LSE stays within the max-affine surrogate and $\Delta_L$ for all $x\in\mathbb{R}^d$:
\begin{equation}\label{eq:sandwich_stream}
\bar y(x)\ \le\ y(x)\ \le\ \bar y(x)+\Delta_L,
\end{equation}
with
\begin{equation}\label{eq:Delta_stream}
\Delta_L\ =\ \sum_{j=1}^{L}\Bigl(T_j\log K_j\prod_{r=j+1}^{L}\alpha^{(r)}_{\max}\Bigr),
\end{equation}
with $\alpha^{(\ell)}_{\max}$ being the max of $\alpha$ for each $k$ and $\alpha^{(1)}_{\max}=1$. Most importantly, by Theorem \ref{ybar_pwma} $\bar{y}(x)$ is still a max-affine function.

Let $\mathcal{M}$ denote all the possible deep max-affine surrogate functions that one can create with a Deep-LSE estimator with $L$ layers, and $\alpha^{(\ell)}>0$.
As a consequence, max-affine functions for continuous convex functions are dense on compact convex sets \cite{huang2020convex}. Hence, for any $\eta>0$, there exist parameters such that
\begin{equation}\label{eq:under_from_below_stream}
0\ \le\ f(x)-\bar y(x)\ \le\ \eta\qquad\text{for all }x\in K.
\end{equation}
We have the bound in Eq. \ref{eq:sandwich_stream} that yields, for all $x\in K$,
\begin{equation}\label{eq:y_minus_ybar_stream}
0\ \le\ y(x)-\bar y(x)\ \le\ \Delta_L .
\end{equation}
Combining Eq. \ref{eq:under_from_below_stream} and Eq. \ref{eq:y_minus_ybar_stream}, for every $x\in K$,
\[
|y(x)-f(x)|\ \le\ |y(x)-\bar y(x)| + |\bar y(x)-f(x)|
\ \le\ \Delta_L + \eta .
\]
Since $\Delta_L$ in Eq. \ref{eq:Delta_stream} is linear in $T_1,\dots,T_L$, we can choose the temperatures small enough to ensure $\Delta_L\le \varepsilon/2$. So that we obtain
\[
T_j\ =\ \frac{\varepsilon}{2^{j+1}}\ \Big/\ \Bigl(\log K_j\ \prod_{r=j+1}^{L}\alpha^{(r)}_{\max}\Bigr)\quad j=1,\dots,L
\]
implies $\Delta_L\le\sum_{j=1}^L \varepsilon/2^{j+1}\le \varepsilon/2$.
Finally we set $\eta=\varepsilon/2$ in Eq. \ref{eq:under_from_below_stream}.
Then $|y(x)-f(x)|\le \varepsilon$ for all $x\in K$, hence $\sup_{x\in K}|y(x)-f(x)|\le\varepsilon$.

\end{proof}

\subsection{Sieve M-estimation and Consistency}\label{sieve_appendix}

\begin{proof}[Proof of Theorem \ref{envelope}]

For a closed domain $\mathcal{X}=\{x\in\mathbb{R}^d:\ \|x\|\le R\}$ with $0<R<\infty$, we define a vector norm $\|\cdot\|$ on $\mathbb{R}^d$ and its {dual} $\|\cdot\|_*$ (so $\lvert\langle u,x\rangle\rvert \le \|u\|_*\,\|x\|$). Recall that we also define the affine pieces as
\[
\ell^{(\ell)}_k(x)=\langle a^{(\ell)}_k,\,x\rangle + b^{(\ell)}_k .
\]
and the Deep-LSE estimator recursion starting from $z^{(1)}(x)=\mathrm{LSE}_{T_1}\!\bigl(\,(\ell^{(1)}_k(x))_{k=1}^{K_1}\bigr)$ and for $\ell\ge2$
\[
z^{(\ell)}(x)=\mathrm{LSE}_{T_\ell}\!\bigl(\,(\alpha^{(\ell)}_k\,z^{(\ell-1)}(x)+\ell^{(\ell)}_k(x))_{k=1}^{K_\ell}\bigr).
\]
whose output is
\[
y(x)=z^{(L)}(x)+c_{\mathrm{out}} .
\]
We define the sieve constraints for each $\ell$ as $M_\ell = \max_k \|a^{(\ell)}_k\|_* \le S_\ell$, $B_\ell = \max_k |b^{(\ell)}_k| \le \bar B_\ell$, and $q_\ell = \max_k \alpha^{(\ell)}_k \le q < 1$. We also can constrain the temperature parameter and the number of neurons for each layer $K_\ell \le \bar K_\ell$, $T_\ell \le \Theta_\ell$ and the bias $|c_{\mathrm{out}}| \le C$.
Recall that for every $\ell \ge 1$ and $x \in \mathbb{R}^d$ we have
\[
\bar z^{(\ell)}(x) \le z^{(\ell)}(x) \le \bar z^{(\ell)}(x) + \Delta_\ell .
\]
\[
\Delta_\ell
= \sum_{j=1}^{\ell} T_j \log K_j \prod_{r=j+1}^{\ell} \max_k \alpha^{(\ell)}_k \ =\sum_{j=1}^{\ell} T_j \log K_j \prod_{r=j+1}^{\ell} q_r .
\]
and with $\Delta_1 = T_1 \log K_1$.

We have that $z^{(\ell)}=\bar z^{(\ell)}+\varepsilon_\ell$ with $0\le \varepsilon_\ell\le \Delta_\ell$, hence
$|z^{(\ell)}|\le |\bar z^{(\ell)}|+|\varepsilon_\ell|\le |\bar z^{(\ell)}|+\Delta_\ell$.
As a consequence, for all $\ell$ and $x$,
\[
\bigl|z^{(\ell)}(x)\bigr| \le \bigl|\bar z^{(\ell)}(x)\bigr| + \Delta_\ell \quad \text{(A)}.
\]
Now, by triangle inequality and dual norm inequality $(|\langle u, v\rangle| \le \|u\|_*\, \|v\|)$,
\[
\bigl|\ell^{(\ell)}_k(x)\bigr|
=\bigl|\langle a^{(\ell)}_k, x\rangle + b^{(\ell)}_k\bigr|
\le \bigl|\langle a^{(\ell)}_k, x\rangle\bigr| + \bigl|b^{(\ell)}_k\bigr|
\le \|a^{(\ell)}_k\|_*\,\|x\| + \bigl|b^{(\ell)}_k\bigr|.
\]
Now we take $\displaystyle \sup_{x\in\mathcal X}$ with $k$ fixed. Since $\lvert b_{k}^{(\ell)}\rvert$ does not depend on $x$
$$ \sup_{\|x\|\le R}\langle a,x\rangle = R\|a\|_* \quad,\quad \sup_{\|x\|\le R}\bigl|\langle a,x\rangle\bigr| = R\|a\|_* $$
\[
\sup_{x\in\mathcal X}\!\bigl(\, \lvert \langle a_{k}^{(\ell)},x\rangle \rvert + \lvert b_{k}^{(\ell)}\rvert \,\bigr)
= \|a_{k}^{(\ell)}\|_{*}\, \sup_{x\in\mathcal X}\|x\| + \lvert b_{k}^{(\ell)}\rvert
= \|a_{k}^{(\ell)}\|_{*}\, R + \lvert b_{k}^{(\ell)}\rvert .
\]
Finally, we take $\max_k$,
\[
\sup_{x\in\mathcal X}\ \max_k \bigl|\ell_{k}^{(\ell)}(x)\bigr|
\;\le\; \max_k\Bigl(\|a_{k}^{(\ell)}\|_{*}\,R + \lvert b_{k}^{(\ell)}\rvert\Bigr)
\;\le\; R \max_k \|a_{k}^{(\ell)}\|_{*} + \max_k \lvert b_{k}^{(\ell)}\rvert
\;=\; R M_\ell + B_\ell 
\;\le\; R\,S_\ell + \bar B_\ell. \quad \text{(B)}
\]
which follows from $M_\ell \le S_\ell$ and $B_\ell \le \bar B_\ell$.
Now let
\[
A_\ell \;=\; \sup_{x\in\mathcal X}\bigl|z^{(\ell)}(x)\bigr|.
\]
For $\ell=1$, combine (A) with (B) and $\Delta_1\le \Theta_1\log K_1$
\begin{align*}
A_1
&= \sup_{x\in\mathcal X}\bigl|z^{(1)}(x)\bigr|
  \le \sup_{x\in\mathcal X}\bigl|\bar z^{(1)}(x)\bigr| + \Delta_1 \\
&\le \sup_{x\in\mathcal X}\max_k \bigl|\ell^{(1)}_k(x)\bigr| + \Theta_1\log K_1
 \le R S_1 + \bar B_1 + \Theta_1\log K_1 .
\end{align*}
For $\ell\ge2$, for each $x$, by (A),
\begin{align*}
\bigl|z^{(\ell)}(x)\bigr|
&\le \bigl|\bar z^{(\ell)}(x)\bigr| + \Delta_\ell \\
&\le \max_k \bigl|\alpha^{(\ell)}_k z^{(\ell-1)}(x) + \ell^{(\ell)}_k(x)\bigr| + \Delta_\ell \\
&\le \max_k \bigl(\alpha^{(\ell)}_k \bigl|z^{(\ell-1)}(x)\bigr|
                 + \bigl|\ell^{(\ell)}_k(x)\bigr|\bigr) + \Delta_\ell \\
&\le q_\ell \bigl|z^{(\ell-1)}(x)\bigr|
   + \max_k \bigl|\ell^{(\ell)}_k(x)\bigr| + \Delta_\ell .
\end{align*}
Taking $\sup_{x\in\mathcal X}$ and using (B) and $\Delta_\ell\le \Theta_\ell\log K_\ell$ gives
\[
A_\ell \le q_\ell A_{\ell-1} + R S_\ell + \bar B_\ell + \Theta_\ell \log K_\ell .
\]
with $A_1 \;\le\; R S_1 + \bar B_1 + \Theta_1 \log K_1$.

To find the analytical formula for the recursion, we define $c_\ell = R S_\ell + \bar B_\ell + \Theta_\ell \log K_\ell$ so that
\[
A_1 \le c_1, \quad
\text{and } \quad A_\ell \le q_\ell A_{\ell-1} + c_\ell \quad (\ell \ge 2).
\]
Then, we show by induction
\[
A_L \le q_L A_{L-1} + c_L
   \le q_L \sum_{j=1}^{L-1} c_j \prod_{r=j+1}^{L-1} q_r + c_L
   = \sum_{j=1}^{L-1} c_j \prod_{r=j+1}^{L} q_r + c_L,
\]
\[
\
A_L \le \sum_{j=1}^{L} c_j \prod_{r=j+1}^{L} q_r
= \sum_{j=1}^{L} \bigl(R S_j + \bar B_j + \Theta_j \log K_j\bigr)
  \prod_{r=j+1}^{L} q_r .\ 
\]
Finally,
\[
\sup_{x\in\mathcal X} |y(x)|
\;\le\; |c_{\mathrm{out}}| + \sup_{x\in\mathcal X} \bigl|z^{(L)}(x)\bigr|
\;\le\; C + A_L,
\]
\[
\;
\sup_{x\in\mathcal X} |y(x)|
\;\le\; C + \sum_{j=1}^{L} \Bigl( R S_j + \bar B_j + \Theta_j \log K_j \Bigr)
\prod_{r=j+1}^{L} q_r \; 
\]
\[
\sup_{x\in\mathcal X} |y(x)|
\;\le\; V_n \; 
\]
which is the envelope depending only on box constants
$(S_j,\,\bar B_j,\,q,\,\Theta_j,\,K_j,\,C,\,R)$.
Since the right-hand side does not depend on $x$ and $f$, we get
\[
\sup_{\theta\in\Theta_n}\ \sup_{x\in\mathcal X} |y_\theta(x)|
\;=\;
\sup_{f\in\mathcal F_n}\|f\|_\infty
\;\le\; V_n\;.
\]

\end{proof}

\bigskip

\begin{proof}[Proof of Theorem \ref{lse_map}]

Let $S(u)=\sum_{i=1}^{m} e^{u_i/T}$. For $j=1,\dots,m$,
\[
\frac{\partial}{\partial u_j}\mathrm{LSE}_T(u)
= T\cdot \frac{1}{S(u)}\cdot \frac{\partial S(u)}{\partial u_j}
= T\cdot \frac{1}{S(u)}\cdot \frac{1}{T} e^{u_j/T}
= \frac{e^{u_j/T}}{S(u)}=p_j(u).
\]
Thus $\nabla \mathrm{LSE}_T(u)=p(u)=(p_1(u),\ldots,p_m(u))$, where each $p_j(u)\in(0,1)$ and
$\sum_{j=1}^{m} p_j(u)=1$.

\noindent For any norm $\|\cdot\|$ and dual $\|\cdot\|_*$, the mean value inequality gives
\[
\bigl|\mathrm{LSE}_T(u)-\mathrm{LSE}_T(v)\bigr|
\le \sup_{\xi\in [u,v]} \bigl\|\nabla \mathrm{LSE}_T(\xi)\bigr\|_*\,\|u-v\|,
\]
so the Lipschitz constant is $\sup_{\xi}\|p(\xi)\|_*$.

\noindent Now we consider the cases 

if $\|\cdot\|=\|\cdot\|_\infty$:\;
$\|\cdot\|_*=\|\cdot\|_1$ and $\|p(\xi)\|_1=\sum_j p_j(\xi)=1$.

if $\|\cdot\|=\|\cdot\|_2$:\;
$\|\cdot\|_*=\|\cdot\|_2$ and $\|p(\xi)\|_2\le \|p(\xi)\|_1=1$.

\noindent Therefore, the Lipschitz constant is $1$ in both cases. 
\end{proof}

\bigskip

\begin{proof}[Proof of Theorem \ref{consistency}]

To prove that 
\[
\sup_{f\in \mathcal{F}_{n}}
\bigl|\,Q_n(f)-\overline{Q}_n(f)\,\bigr|
\;\xrightarrow{\,p^{*}\,}\;0 \qquad as \ n\to\infty.
\]
under a given growth condition of $\mathcal F_n$, \textcite{shen2023asymptotic} shows that it suffices to have 
\begin{align*}
\mathbb{E}_{\xi}&\!\left[
  \sup_{f\in\mathcal F_{n}}
  \left|\frac{1}{n}\sum_{i=1}^{n}\xi_i\epsilon_i\bigl(f(x_i)-f_0(x_i)\bigr)\right|
\right]\\
&\le
\mathbb{E}_{\xi}\!\left[
  \left|\frac{1}{n}\sum_{i=1}^{n}\xi_i\epsilon_i\bigl(f_n^{\ast}(x_i)-f_0(x_i)\bigr)\right|
\right]
+ K\!\int_{0}^{2V_n}\!\sqrt{\frac{\log \mathcal N\!\left(\frac{\eta}{2\sqrt{\sigma^{2}}+1},\ \mathcal F_{n},\ \|\cdot\|_{\infty}\right)}{n}}\,d\eta .
\end{align*}
Regarding the first term, we apply the Universal Approximation theorem (\ref{thm:uniform_approx_lse_streamlined}) and choose $f_n^*=\pi_{n} f_0$. This means that $\sup_{x\in\mathcal X}|f_n^*(x)-f_0(x)|\to 0$ as $n\to\infty$. Therefore, the first term is smaller than any $\zeta>0$, for a sufficiently large $n$.

Regarding the second term, we apply Theorem 14.5 of \textcite{anthony2009neural}
\[
\mathcal{N}_\infty(\epsilon,\mathcal{F},m)
\;\le\;
\left(\frac{4\,e\,m\,b\,W\,(L V)^{\ell}}{\epsilon\,(L V-1)}\right)^{W},
\]
\begin{align*}
\mathcal N\!\left(\frac{1}{2\sqrt{\sigma^{2}+1}}\,\eta,\ \mathcal F_{n},\ \|\cdot\|_\infty\right)
&\le \left(
\frac{(8\sqrt{\sigma^{2}+1})\,e\,[\,b\,][\,W\,]\left(V_n\right)^{L}}
{\eta\left(V_n-1\right)}
\right)^{\,W} ,\notag\\
&\le \ \left(
\frac{(8\sqrt{\sigma^{2}+1})\,e\,[\,V_n\,][\,W\,]\left(V_n\right)^{L}}
{\eta\left(V_n-1\right)}
\right)^{\,W} ,\\
&= \ \left(
\frac{(8\sqrt{\sigma^{2}+1})\,e\,[\,W\,]\left(V_n\right)^{L+1}}
{\eta\left(V_n-1\right)}
\right)^{\,W} ,\\
&= \ \tilde B_{W,d,V_n}\,\eta^{-W},
\end{align*}
where
\[
\tilde B_{W,d,V_n}
=
\left(\frac{\bigl(8\sqrt{\sigma^{2}+1}\bigr)\,e[\,W\,]\,(V_n)^{L+1}}{V_n-1} \right)^{\,W}.
\]
Now, let
\[
B_{W,d,V_n}=\log\tilde B_{W,d,V_n}- W .
\]
Then
\begin{align*}
B_{W,d,V_n}
& = W \left(
 \log\frac{\bigl(8\sqrt{\sigma^{2}+1}\bigr)\,e[\,W\,]\,(V_n)^{L+1}}{V_n-1}-1\right),  \\
&= W \left(
 \log\frac{W\,(V_n)^{L+1}}{V_n-1}+\log(8\sqrt{\sigma^{2}+1})\right),  \\
&\le 2W\, \log\frac{W\,(V_n)^{L+1}}{V_n-1}, 
\end{align*}
that holds because $V_n^{L+1} - V_n + 1 \ge 0$ for all $V_n$ and $L+1$ being even, so that
\[
\log\!\left(\frac{ W V_n^{L+1}}{V_n-1}\right)
\;\ge\;
\log\!\left(\frac{(2\sqrt{\sigma^{2}+1})(V_n-1)}{V_n-1}\right)
=\log\!\bigl(2\sqrt{\sigma^{2}+1}\bigr).
\]
Since $V^{L+1}-eV+e\ge 0$ for all $V$, we have $\dfrac{V^{L+1}}{V-1}\ge e$, hence
\begin{align*}
\log\!\left(W\frac{V^{L+1}}{V-1}\right)
&\ge \log\!\left(\frac{V^{L+1}}{V-1}\right)
\ge \log\!\left(\frac{e(V-1)}{V-1}\right)
= 1 .
\end{align*}
But we also have:
\begin{align*}
B_{W,d,V}
&= \log \tilde B_{W,d,V} - W, \\
&= \log \left(\frac{\bigl(8\sqrt{\sigma^{2}+1}\bigr)\,eW(V_n)^{L+1}}{V_n-1} \right)^{\,W} - W ,\\
&= W\!\left[\log\!\left(\bigl(8\sqrt{\sigma^{2}+1}\bigr)e\,W\frac{V^{L+1}}{V-1}\right)-1\right], \\
&= W\,\log\!\left(\bigl(8\sqrt{\sigma^{2}+1}\bigr)W\frac{V^{L+1}}{V-1}\right),\\
&= W * (\mathbb{R} > 1),\\
&\ge W.
\end{align*}
Hence,
\begin{align*}
H\!\left(\frac{1}{2\sqrt{\sigma^{2}+1}}\,\eta,\ \mathcal F_{n},\ \|\cdot\|_\infty\right)
&= \log \mathcal N\!\left(\frac{1}{2\sqrt{\sigma^{2}+1}}\,\eta,\ \mathcal F_{n},\ \|\cdot\|_\infty\right) ,\notag\\
&\le B_{W,d,V_n}\!\left(1+\frac{1}{\eta}\right) \ \ \text{since } (B \ge W). 
\end{align*}\label{Hbound}
As a result, for sufficiently large $n$ we have
\begin{align*}
\int_{0}^{2V_n} H^{1/2}\!\left(\frac{1}{2\sqrt{\sigma^{2}}+1}\,\eta,\ \mathcal F_{n},\ \|\cdot\|_\infty\right)\, d\eta
&\le 4\sqrt{2}\, B_{W,d,V_n}^{1/2}\, V_n . \label{eq:intHhalf}
\end{align*}
Recall that
\begin{equation*}\label{eq:Bdef}
B_{W,d,V_n}
= \log \tilde B_{W,d,V_n} - W
= W \log\!\left( (8\sqrt{\sigma^2+1})\,W\,\frac{V_n^{L+1}}{V_n-1} \right).
\end{equation*}
Then the bound entropy integral becomes
\begin{align*}
\int_{0}^{2V_n}\!
\sqrt{\frac{ H\!\left(\frac{\eta}{2\sqrt{\sigma^2}+1},\,\mathcal F_{n},\,\|\cdot\|_\infty\right)}{n}}\; d\eta
&\le 4\sqrt{2}\,\frac{V_n}{\sqrt{n}}\;\sqrt{B_{W,d,V_n}} ,\notag\\
&\le 4\sqrt{2}\,\frac{V_n}{\sqrt{n}}\;
   \sqrt{\,W\Bigg[\log\!\left(\frac{V_n^{L+1}}{V_n-1}\right)
   + \log\!\big((8\sqrt{\sigma^2+1})W\big)\Bigg] }.
\end{align*}
As \(V_n\to\infty\),
\[
\log\!\left(\frac{V_n^{L+1}}{V_n-1}\right)=\log V_n^L+o(1),
\]
so
\begin{equation*}
\sqrt{B_{W,d,V_n}}
\;\lesssim\; \sqrt{\,W\,[\log V_n^L + \log W +\log (8\sqrt{\sigma^2+1})+o(1)]\,}
\;\le\; \sqrt{\,2W\,\log(V_n^L W)\,},
\end{equation*}
because as \(\log V_n^LW\to\infty,\  8\sqrt{\sigma^2+1}\le\log V_n^LW\). 
Therefore,
\begin{align*}
\int_{0}^{2V_n}\!
\sqrt{\frac{ H\!\left(\frac{\eta}{2\sqrt{\sigma^2+1}},\,\mathcal F_{n},\,\|\cdot\|_\infty\right)}{n}}\; d\eta
&\le 4\sqrt{2}\,\frac{V_n}{\sqrt{n}}\;\sqrt{\,2W\,\log(V_n^L W)\,}, \notag\\
&= 8\,\sqrt{\frac{W\,V_n^2\,\log(V_n^L W)}{n}}\;.
\end{align*}
Under the assumptions $W V_n^{2}\,\log\!\big(V_n^L W\big)=o(n)\quad as \ n\to\infty,$ we have that
\[
\sqrt{\frac{W\, V_n^{2}\,
\log\!\big(V_n^L W\big)}{n}}
< \frac{\zeta}{8}\,.
\]
Finally, we have
\begin{align*}
\mathbb{E}_{\xi}\!\Bigg[
  \sup_{f\in\mathcal F_{n}}
  \Bigl|\frac1n\sum_{i=1}^{n}\xi_i\epsilon_i\bigl(f(x_i)-f_0(x_i)\bigr)\Bigr|
\Bigg]
&\le \sqrt{\sigma^2+1}\,\bigl\|\pi_{n}f_0-f_0\bigr\|_\infty + 4\sqrt{2} K\,\frac{V_n}{\sqrt{n}}\;\sqrt{\,2W\,\log(V_n^L W)\,},
\label{eq:aux-bound}
\end{align*}
which goes to  $0$ as $n\to\infty$. This means that 
\begin{align*}
\mathbb{E}^{*}&\!\Bigg[
  \sup_{f\in\mathcal F_{n}}
  \Bigl|\frac1n\sum_{i=1}^{n}\epsilon_i\bigl(f(x_i)-f_0(x_i)\bigr)\Bigr|
\Bigg]\xrightarrow[n\to\infty]{} 0 ,
\end{align*}
completing the proof.

\end{proof}

\newpage

\section{Additional Simulations}\label{more_sim}

\subsection{Kou-Heston Stochastic Volatility model}\label{kou_heston_appendix}

The \textcite{kou2002jump} Stochastic Volatility model has the volatility evolution
$$
d V_t=\kappa\left(\theta-V_t\right) d t+\sigma_v \sqrt{V_t} d W_t^{(2)}
$$
that is simulated using Euler-Maruyama. The asset price evolution is
$$
d S_t=S_t\left[\left(r-q-\lambda\kappa_J\right) d t+\sqrt{V_t} d W_t^{(1)}+d J_t\right]
$$
$d J_t$ is the jump component from a compound Poisson process using intensity $\lambda_j$ and double-exponential jump sizes
$$
Y \sim \begin{cases}\operatorname{Exp}\left(1 / \eta_1\right), & \text { with prob } p_{\text {up }} \\ -\operatorname{Exp}\left(1 / \eta_2\right), & \text { with prob } 1-p_{\text {up }}\end{cases}
$$

Two correlated Wiener processes using Cholesky
$$
\operatorname{Corr}\left(d W_t^{(1)}, d W_t^{(2)}\right)=\rho
$$

\[
\mathrm dJ_t
     =\sum_{i=1}^{\mathrm dN_t}\!\Bigl(e^{Y_i}-1\Bigr),
\qquad
N_t\sim\operatorname{Poisson}(\lambda t),
\]
\[
Y_i\sim
\begin{cases}
\mathrm{Exp}(1/\eta_1), & \text{with probability } p_{\text{up}},\\[4pt]
-\mathrm{Exp}(1/\eta_2),& \text{with probability } 1-p_{\text{up}}.
\end{cases}
\]

\[
\kappa_J
=\mathbb E\!\bigl[e^{Y}-1\bigr]
=p_{\text{up}}\dfrac{\eta_1}{\eta_1-1}
 +(1-p_{\text{up}})\dfrac{\eta_2}{\eta_2+1}-1,
\quad
\eta_1>1,\;\eta_2>-1.
\]
For risk-neutral simulation, one can use
$$
S_{t+\Delta t}=S_t \cdot \exp \left[\left(r-q-\lambda\left(\frac{p_{\mathrm{up}} \eta_1}{\eta_1-1}+\frac{\left(1-p_{\mathrm{up}}\right) \eta_2}{\eta_2+1}-1\right)-\frac{1}{2} V_t\right) \Delta t+\sqrt{V_t} \Delta W_t^{(1)}+J_t\right]
$$

In Table \ref{tab:KH_param} we report the parameters we use for the simulation. To construct the target implied volatility curve, we assume it is obtained by translating the source curve. In particular, we shift implied volatilities downward by 10\% (y-axis) and increase strikes by $20\$$ (x-axis).

\begin{table}[h]
    \centering
    \begin{tabular}{c c c c c c c c c c c c}
    \hline
    $S_0$ & $r$ & $q$ & $v_0$ & $\kappa$ & $\theta$ & $\sigma$ & $\rho$ & $\lambda$ & $p_{\text{up}}$ & $\eta_1$ & $\eta_2$ \\
    \hline
    100 & 0.05 & 0.00 & 0.04 & 2.0 & 0.04 & 0.8 & -0.5 & 0.12 & 0.35 & 8.0 & 10.0 \\
    \end{tabular}
    \caption{Simulated parameters for Kou-Heston model.}
    \label{tab:KH_param}
\end{table}

We study two different situations of severe market illiquidity by randomly selecting three in-the-money call option quotes in Scenario 1 and three out-of-the-money call option quotes in Scenario 2. We emphasize that these three call option quotes constitute the only information on the terminal RND available to the models. In Fig. \ref{fig:kh_setup}, we report Scenario 1, which consists of three illiquid observations of in-the-money (ITM) options.

\begin{figure}[H]
\centering
\begin{minipage}{\textwidth}
\begin{subfigure}{0.55\textwidth}
  \includegraphics[width=\linewidth]{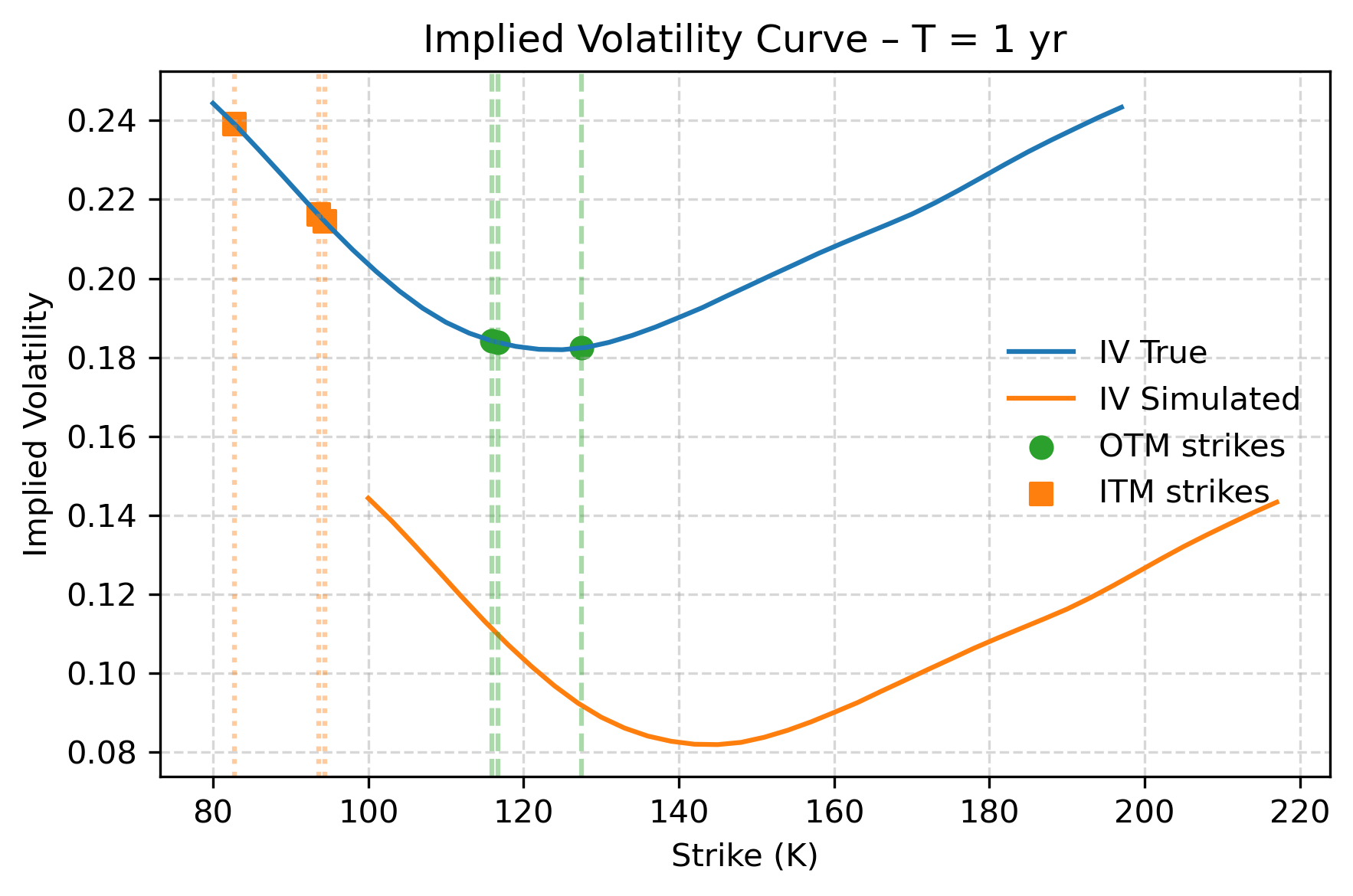}
\end{subfigure}\hfill
\begin{subfigure}{0.55\textwidth}
  \includegraphics[width=\linewidth]{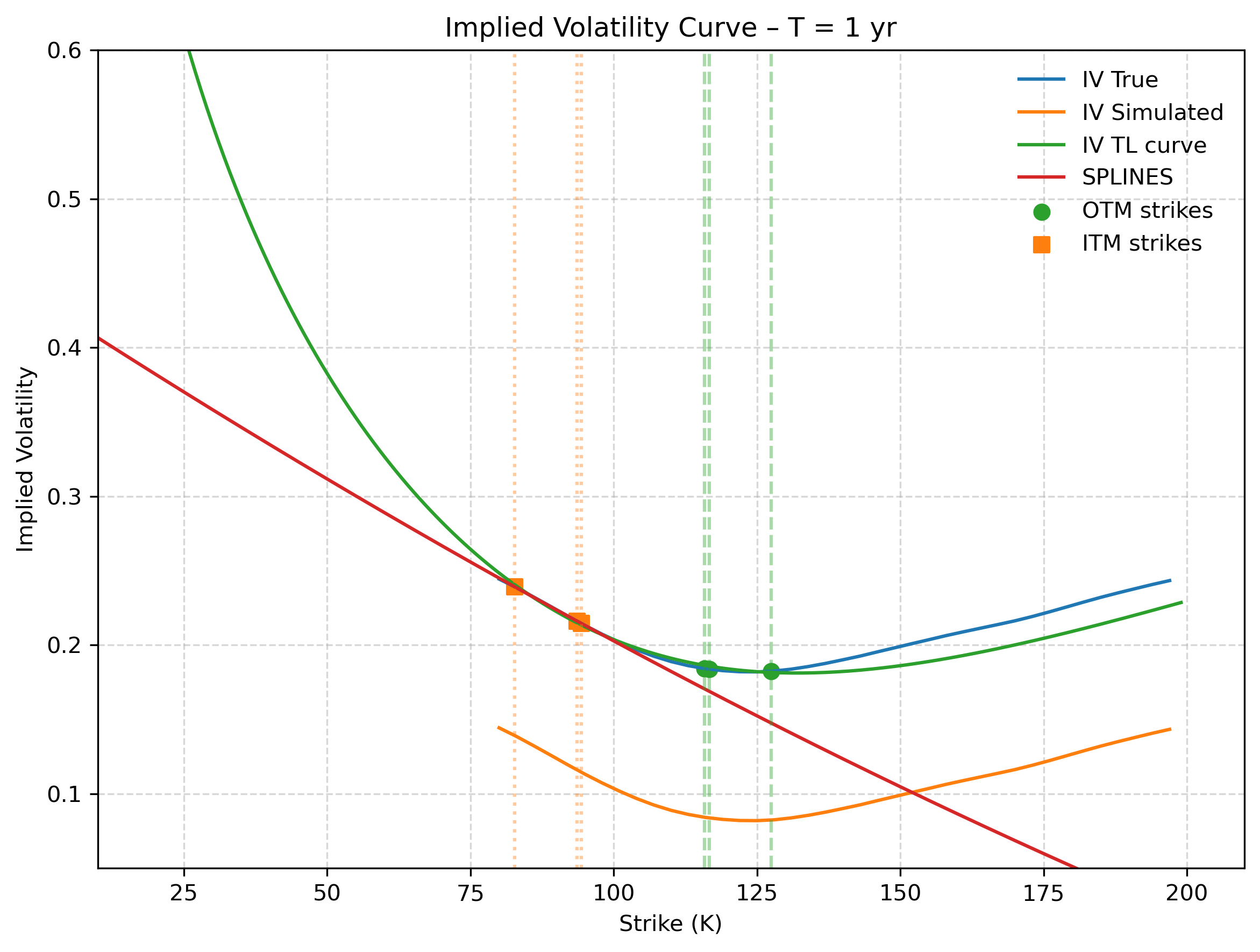}
\end{subfigure}
\caption{Scenario 1 - On the left panel, the orange dots are the ITM strikes selected on the target illiquid implied volatility curve (blue solid line), and the orange solid line is the liquid source (proxy) implied volatility curve. On the right panel, the interpolation of the illiquid implied volatility curve of the Deep-LSE model (green solid line) and quadratic splines (red solid line).}
\label{fig:kh_setup}
\end{minipage}

\end{figure}

We observe in Fig. \ref{fig:kh_final} the estimates of Deep-LSE and quadratic splines compared to the ground truth illiquid RND. It emerges that the Deep-LSE accurately recovers the illiquid RND.

\begin{figure}[H]
    \centering
    \begin{subfigure}{0.75\textwidth}
        \includegraphics[width=\linewidth]{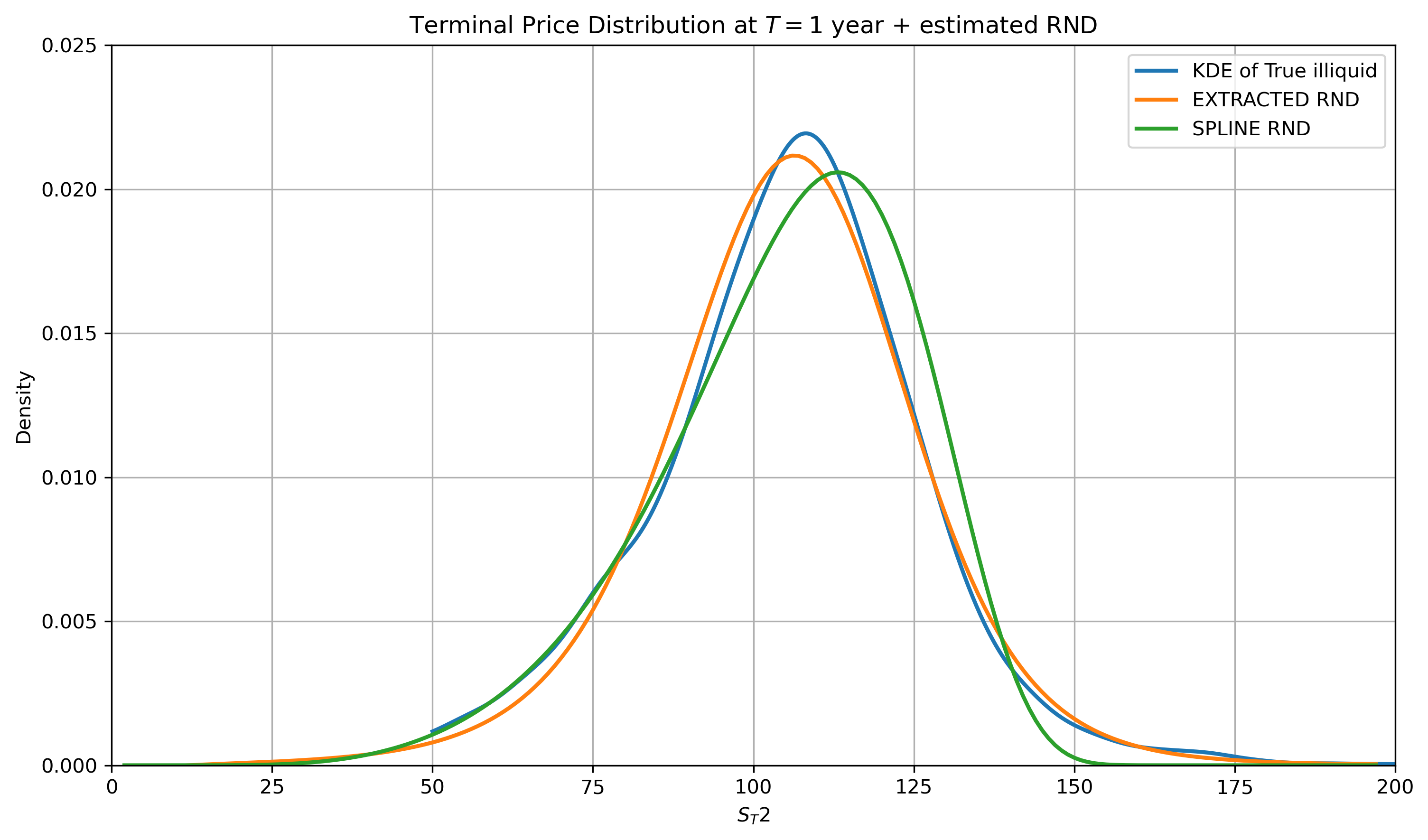} 
    \end{subfigure}
    \caption{Scenario 1 - Illiquid RND recovery of Deep-LSE (orange curve) and quadratic splines (green curve) in comparison with the illiquid target ground truth simulated RND (blue curve).}
    \label{fig:kh_final}
\end{figure}

We also test the Deep-LSE model on out-of-the-money call options and illiquid strikes (Scenario 2), and Fig. \ref{fig:kh_setup2} illustrates this case.

\begin{figure}[H]
\centering
\begin{minipage}{\textwidth}
\begin{subfigure}{0.55\textwidth}
  \includegraphics[width=\linewidth]{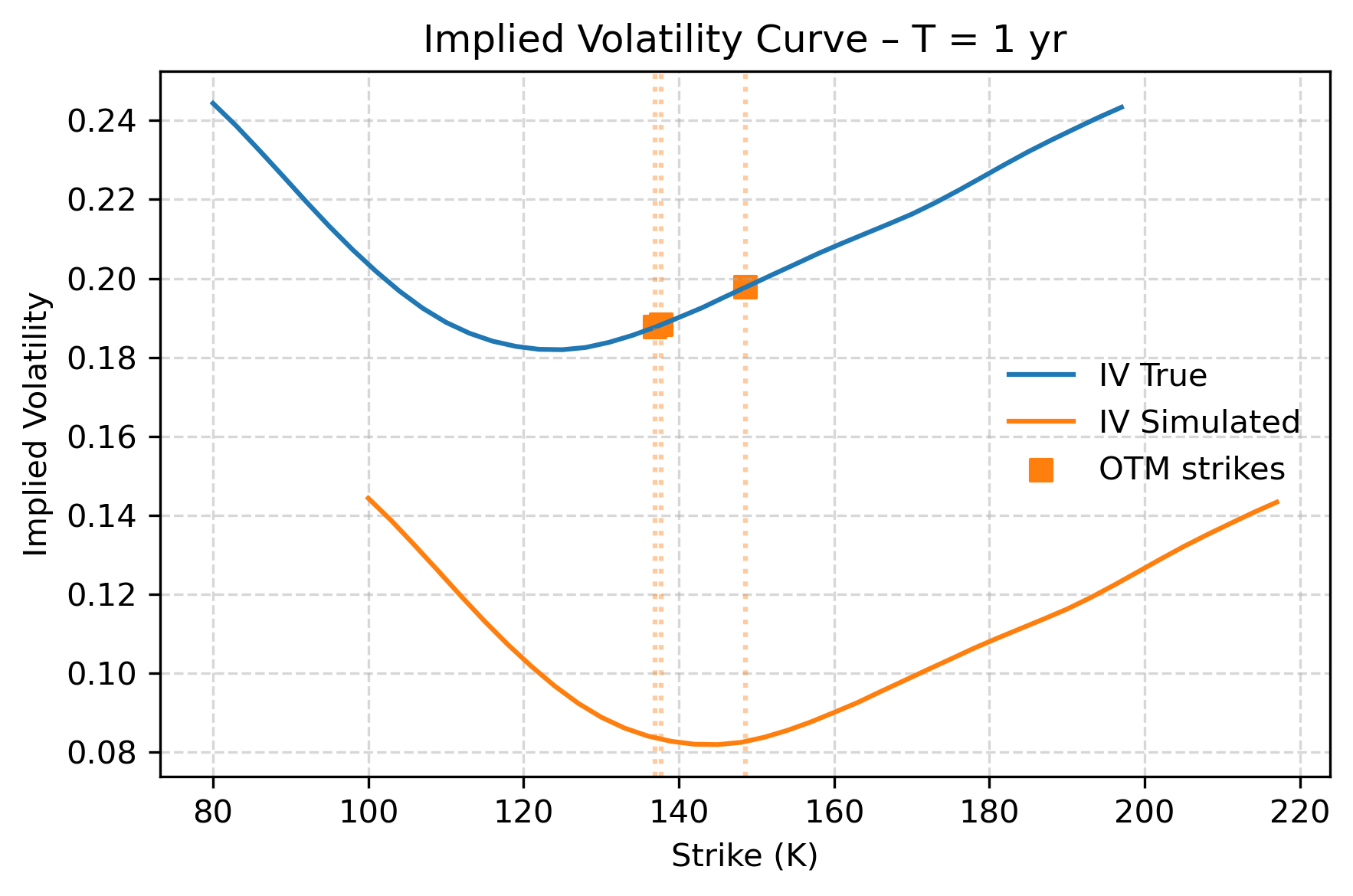}
\end{subfigure}\hfill
\begin{subfigure}{0.55\textwidth}
  \includegraphics[width=\linewidth]{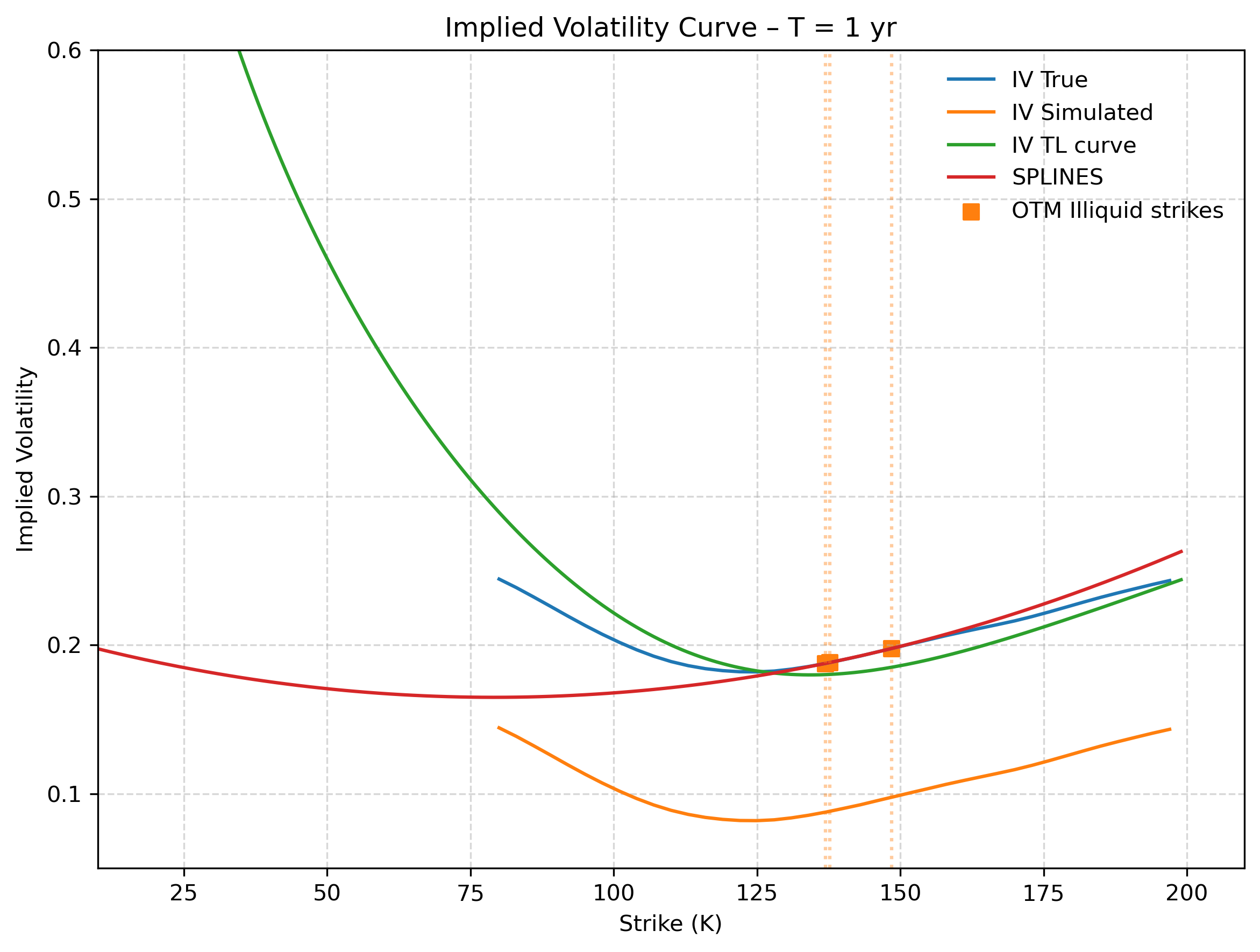}
\end{subfigure}
\caption{Scenario 2 - On the left panel, the orange dots are the OTM strikes selected on the target illiquid implied volatility curve (blue solid line), and the orange solid line is the liquid source (proxy) implied volatility curve. On the right panel, the interpolation of the illiquid implied volatility curve of the Deep-LSE model (green solid line) and quadratic splines (red solid line).}
\label{fig:kh_setup2}
\end{minipage}

\end{figure}

\begin{figure}[H]
    \centering
    \begin{subfigure}{0.75\textwidth}
        \includegraphics[width=\linewidth]{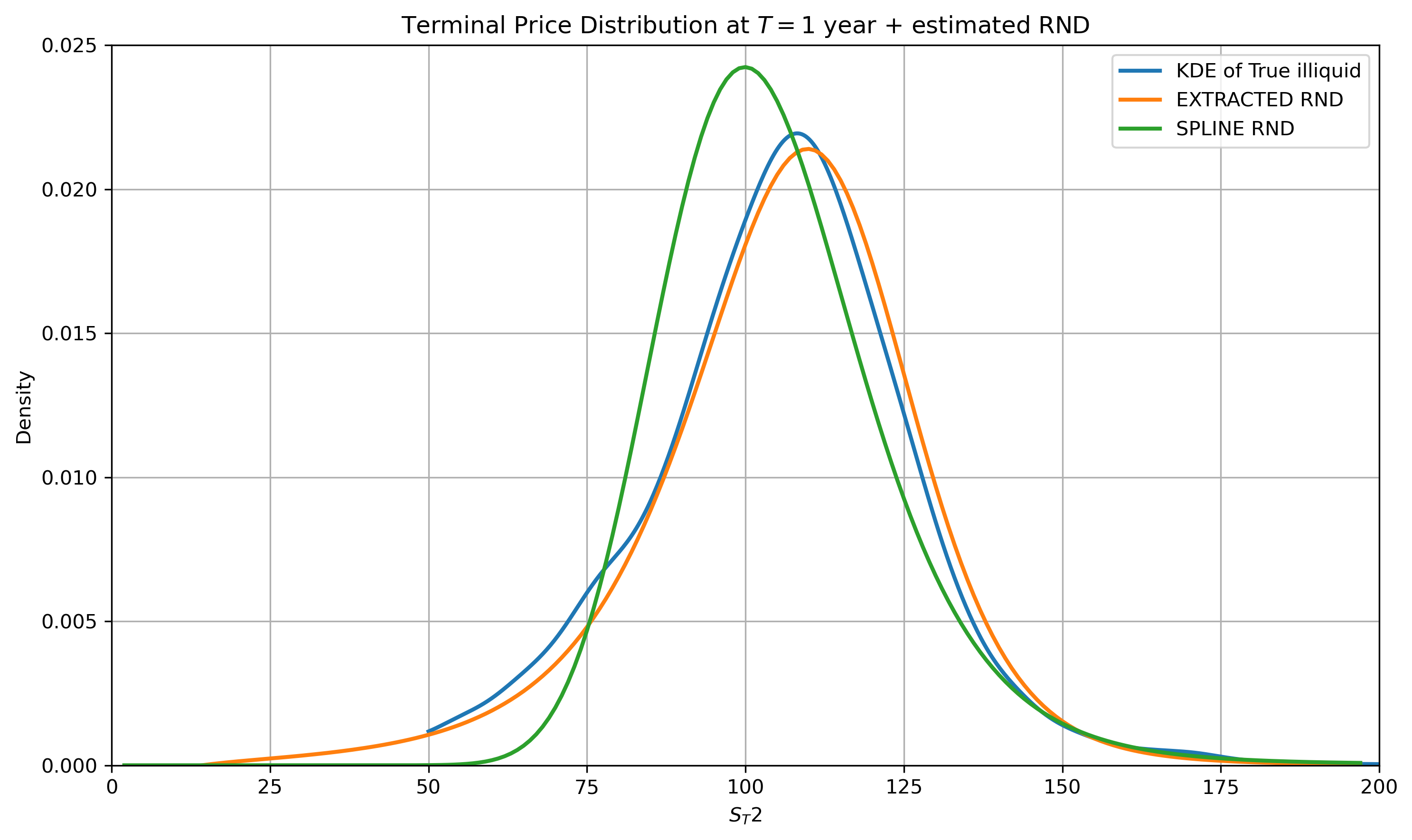} 
    \end{subfigure}
    \caption{Scenario 2 - Illiquid RND recovery of Deep-LSE (orange curve) and quadratic splines (green curve) in comparison with the illiquid target ground truth simulated RND (blue curve).}
    \label{fig:kh_final2}
\end{figure}

In Fig. \ref{fig:kh_final2}, we observe that, also in Scenario 2, the Deep-LSE model recovers the illiquid RND better with respect to the quadratic splines.

\subsection{Andersen-Benzoni-Lund Multifactor Model}\label{abl_appendix}

The \textcite{andersen2002empirical} model follows 
$$
d S_t=S_t\left[\left(\mu-\frac{1}{2} \sum_{i=1}^n V_t^{(i)}\right) d t+\sqrt{V_t} d Z_t+d J_t\right]
$$

Where $V_t=\sum_{i=1}^n V_t^{(i)}$, $d Z_t$ is a Brownian motion correlated with the volatility factors, and $d J_t \sim \sum_{k=1}^{dN_t} Y_k$, where $dN_t \sim \operatorname{Poisson}\left(\lambda_j d t\right)$, and $Y_k \sim \mathcal{N}\left(\mu_j, \sigma_j^2\right)$ is the jump size. Each volatility factor $V_t^{(i)}$ follows a square-root process)
$$
d V_t^{(i)}=\kappa_i\left(\theta_i-V_t^{(i)}\right) d t+\sigma_i \sqrt{V_t^{(i)}} d W_t^{(i)}
$$

Where $d W_t^{(i)}$ are independent Brownian motions that are correlated with $d Z_t$.
The correlation between $d Z_t$ and $d W_t^{(i)}$ is handled via
$$
d Z_t=\sum_{i=1}^n \rho_i  d W_t^{(i)}+\sqrt{1-\sum_{i=1}^n \rho_i^2} \cdot d W_t^{(0)}
$$

To simulate under the risk-neutral probability measure, one can use
\[
\mu_{\mathbb{Q}} \;=\; r \;-\; \lambda_j\!\left( e^{\,\mu_j \;+\; \frac12\,\sigma_j^{2}} - 1 \right).
\]

In Table \ref{tab:ABL_param} we report the parameters we use for the simulation. To construct the target implied volatility curve, we define two different set of parameters and obtain two distinct implied volatility curves. Regarding the first, we assume it is the liquid proxy, whereas the second is the illiquid target.

\begin{table}[h]
\centering
\resizebox{\textwidth}{!}{%
\begin{tabular}{l c c c c c c c c c c c c c c c c c c}
\hline
Set 
& $S_0$ & $r$ & $\rho$ 
& $\kappa_1$ & $\theta_1$ & $\sigma_1$ & $v_{0,1}$
& $\kappa_2$ & $\theta_2$ & $\sigma_2$ & $v_{0,2}$
& $\kappa_3$ & $\theta_3$ & $\sigma_3$ & $v_{0,3}$
& $\lambda_j$ & $\mu_j$ & $\sigma_j$ \\
\hline
1 
& 100 & 0.05 & $[-0.3,\ 0.0,\ 0.3]$
& 3.0 & 0.02 & 0.2 & 0.02
& 1.5 & 0.04 & 0.3 & 0.04
& 0.5 & 0.06 & 0.4 & 0.06
& 0.20 & 0.00 & 0.55 \\
2 
& 100 & 0.05 & $[-0.3,\ 0.0,\ 0.3]$
& 3.0 & 0.02 & 0.2 & 0.02
& 1.5 & 0.04 & 0.3 & 0.04
& 0.5 & 0.06 & 0.4 & 0.06
& 0.25 & 0.18 & 0.60 \\
\hline
\end{tabular}%
}
\caption{Simulated parameters for Andersen-Benzoni-Lund model. Set 1 for the liquid proxy and set 2 for the illiquid target.}
\label{tab:ABL_param}
\end{table}

We examine two instances of severe market illiquidity by randomly selecting three in-the-money call option quotes for Scenario 1 and three out-of-the-money call option quotes for Scenario 2. We emphasize that these three option quotes constitute the only information on the terminal RND available to the models. In Fig. \ref{fig:abl_setup}, we represent the first case under consideration (Scenario 1), which consists of three illiquid observations of in-the-money (ITM) call options.

\begin{figure}[H]
\centering
\begin{minipage}{\textwidth}
\begin{subfigure}{0.55\textwidth}
  \includegraphics[width=\linewidth]{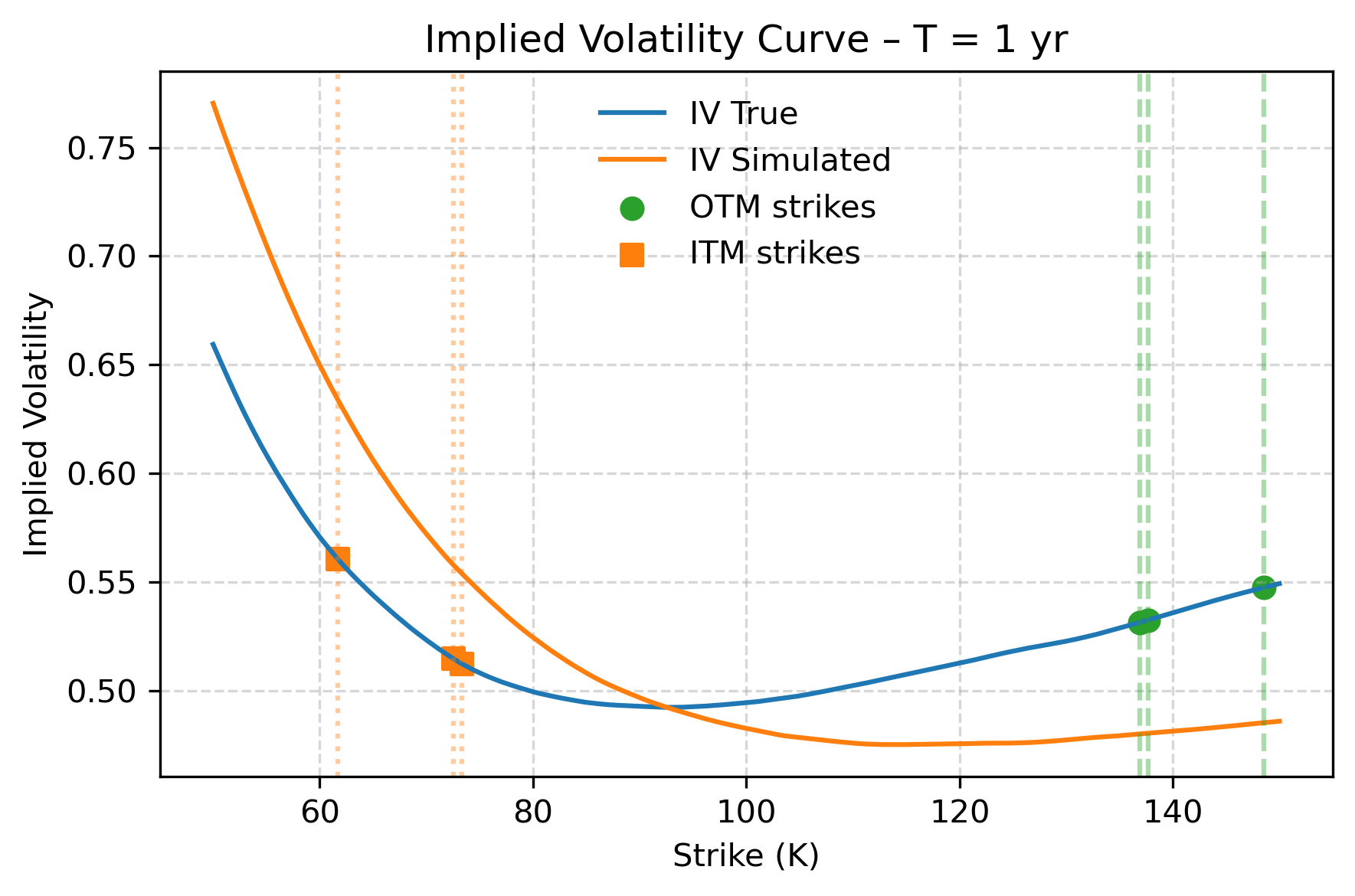}
\end{subfigure}\hfill
\begin{subfigure}{0.55\textwidth}
  \includegraphics[width=\linewidth]{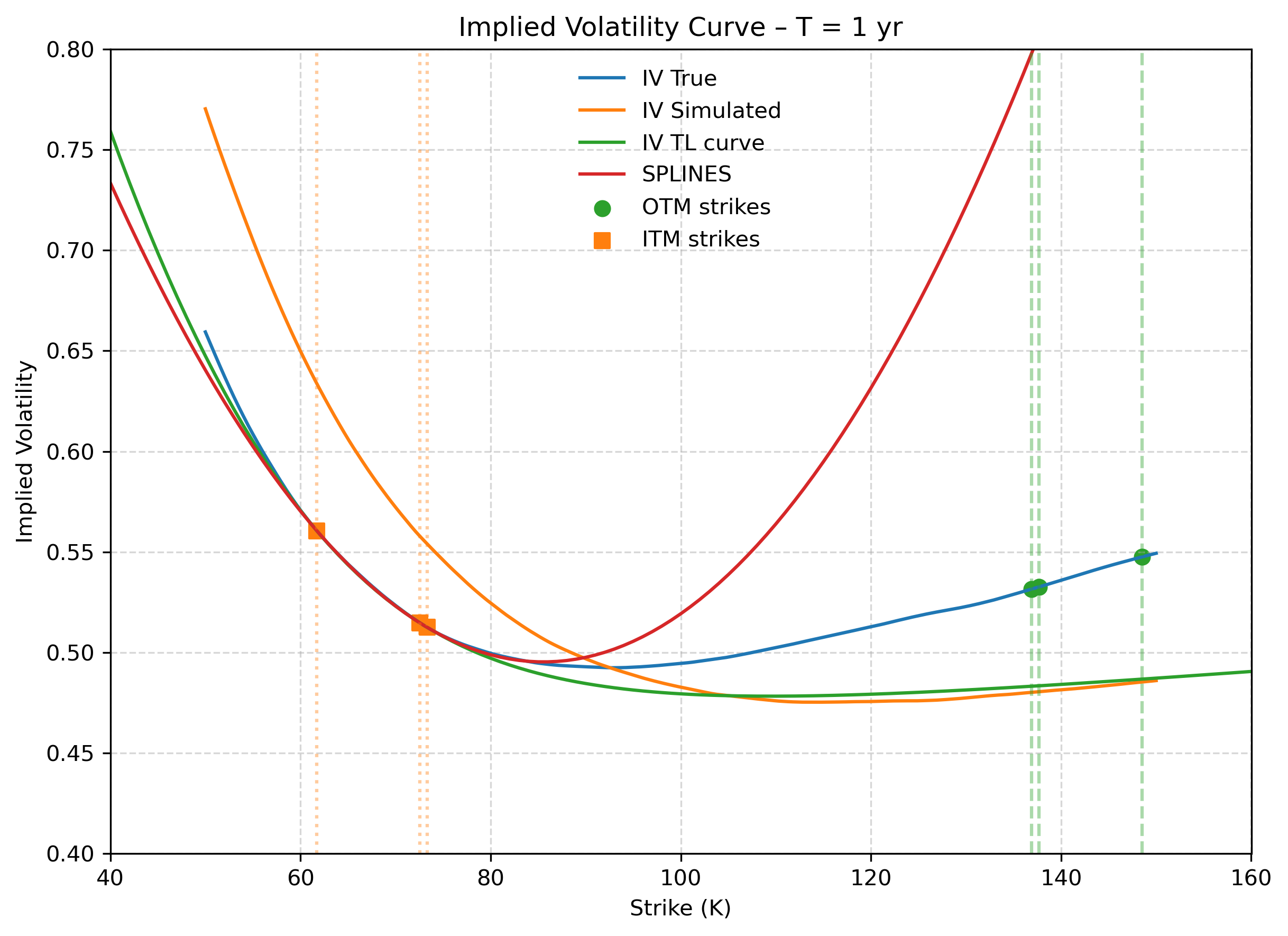}
\end{subfigure}
\caption{Scenario 1 - On the left panel, the orange dots are the ITM strikes selected on the target illiquid implied volatility curve (blue solid line), and the orange solid line is the liquid source (proxy) implied volatility curve. On the right panel, the interpolation of the illiquid implied volatility curve of the Deep-LSE model (green solid line) and quadratic splines (red solid line).}
\label{fig:abl_setup}
\end{minipage}

\end{figure}

We observe in Fig. \ref{fig:abl_final} the estimates of Deep-LSE and quadratic splines compared to the ground truth illiquid RND. It emerges that the Deep-LSE accurately recovers the illiquid RND, while the estimate of the quadratic spline is erratic and does not approximate the target RND.

\begin{figure}[H]
    \centering
    \begin{subfigure}{0.75\textwidth}
        \includegraphics[width=\linewidth]{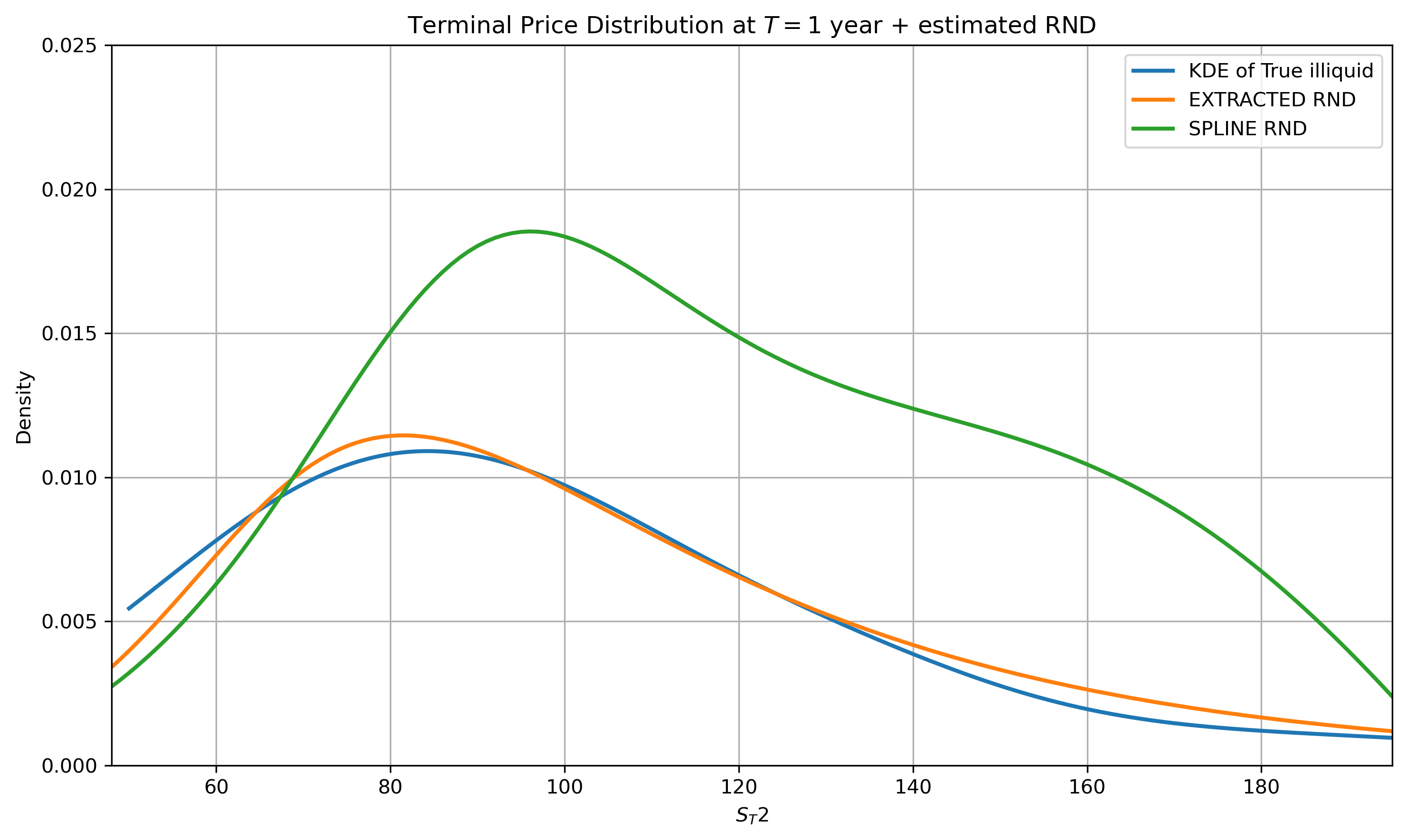} 
    \end{subfigure}
    \caption{Scenario 1 - Illiquid RND recovery of Deep-LSE (orange curve) and quadratic splines (green curve) in comparison with the illiquid target ground truth simulated RND (blue curve).}
    \label{fig:abl_final}
\end{figure}


We also test the Deep-LSE model on out-of-the-money call options and illiquid strikes (Scenario 2), and Fig. \ref{fig:abl_setup2} illustrates this case.

\begin{figure}[H]
\centering
\begin{minipage}{\textwidth}
\begin{subfigure}{0.55\textwidth}
  \includegraphics[width=\linewidth]{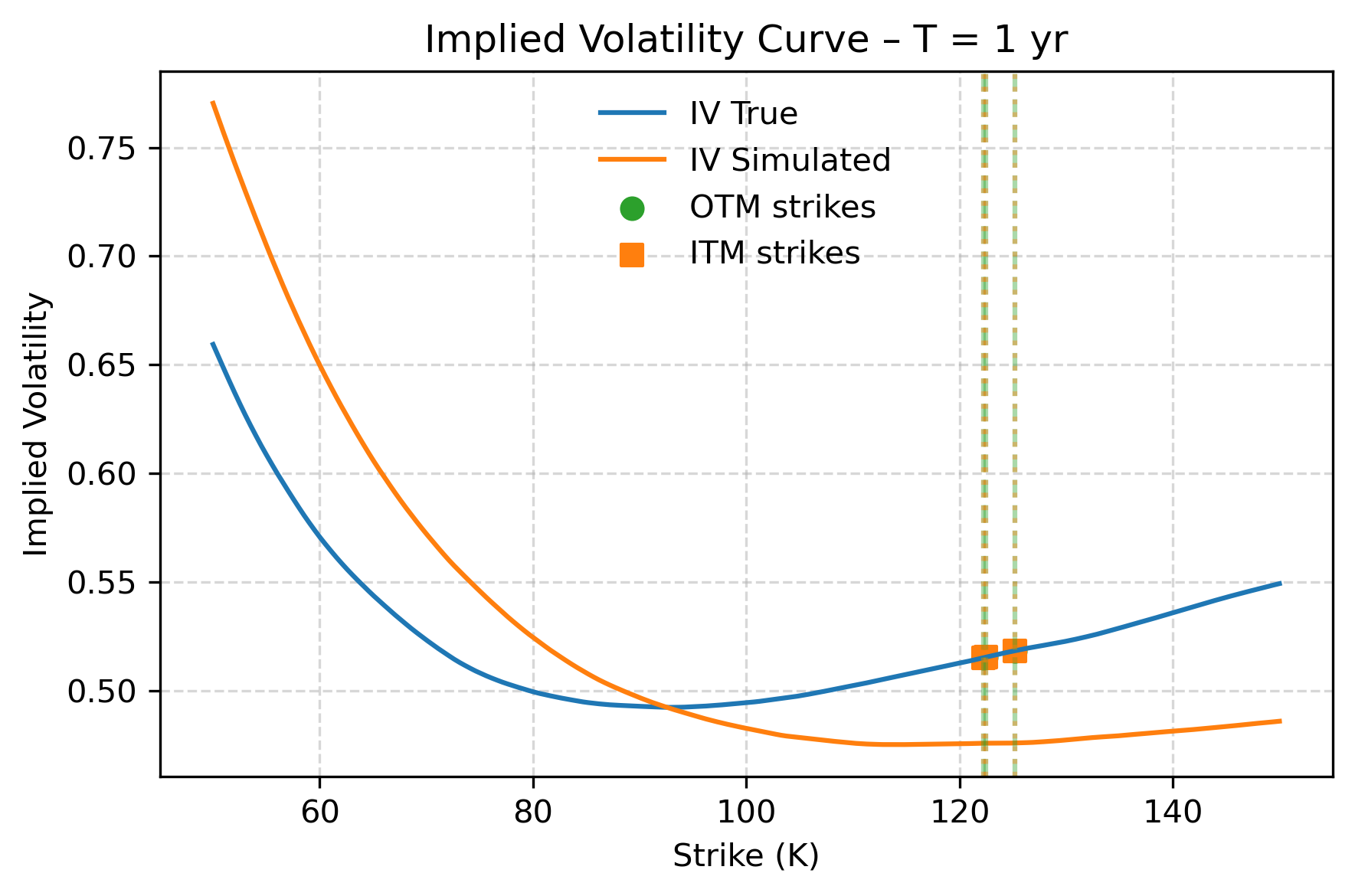}
\end{subfigure}\hfill
\begin{subfigure}{0.55\textwidth}
  \includegraphics[width=\linewidth]{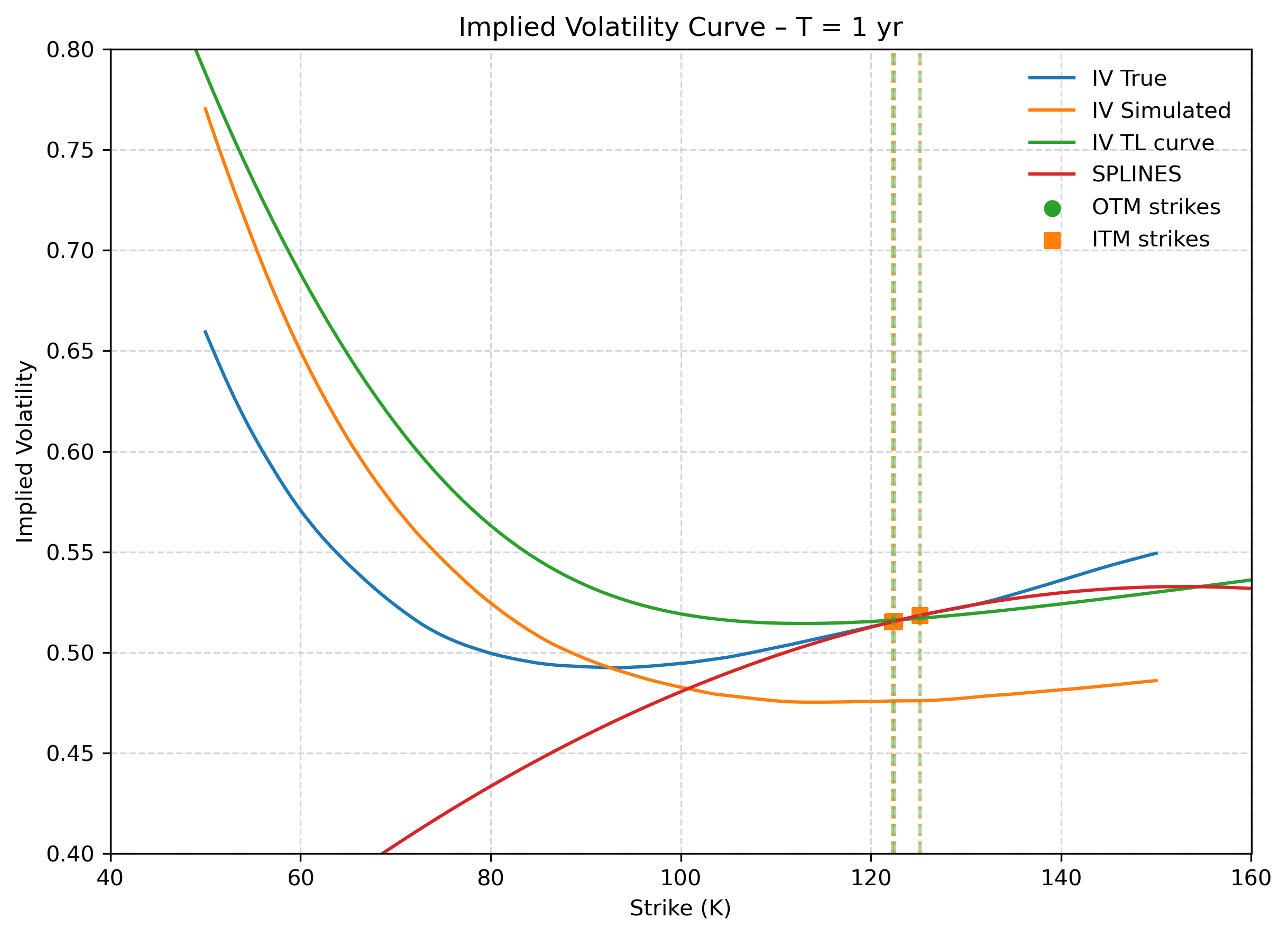}
\end{subfigure}
\caption{Scenario 2 - On the left panel, the orange dots are the OTM strikes selected on the target illiquid implied volatility curve (blue solid line), and the orange solid line is the liquid source (proxy) implied volatility curve. On the right panel, the interpolation of the illiquid implied volatility curve of the Deep-LSE model (green solid line) and quadratic splines (red solid line).}
\label{fig:abl_setup2}
\end{minipage}

\end{figure}

In Fig. \ref{fig:abl_final2}, we observe that also in Scenario 2, the Deep-LSE model is able to recover the illiquid RND, while quadratic splines yield an inaccurate estimate.

\begin{figure}[H]
    \centering
    \begin{subfigure}{0.75\textwidth}
        \includegraphics[width=\linewidth]{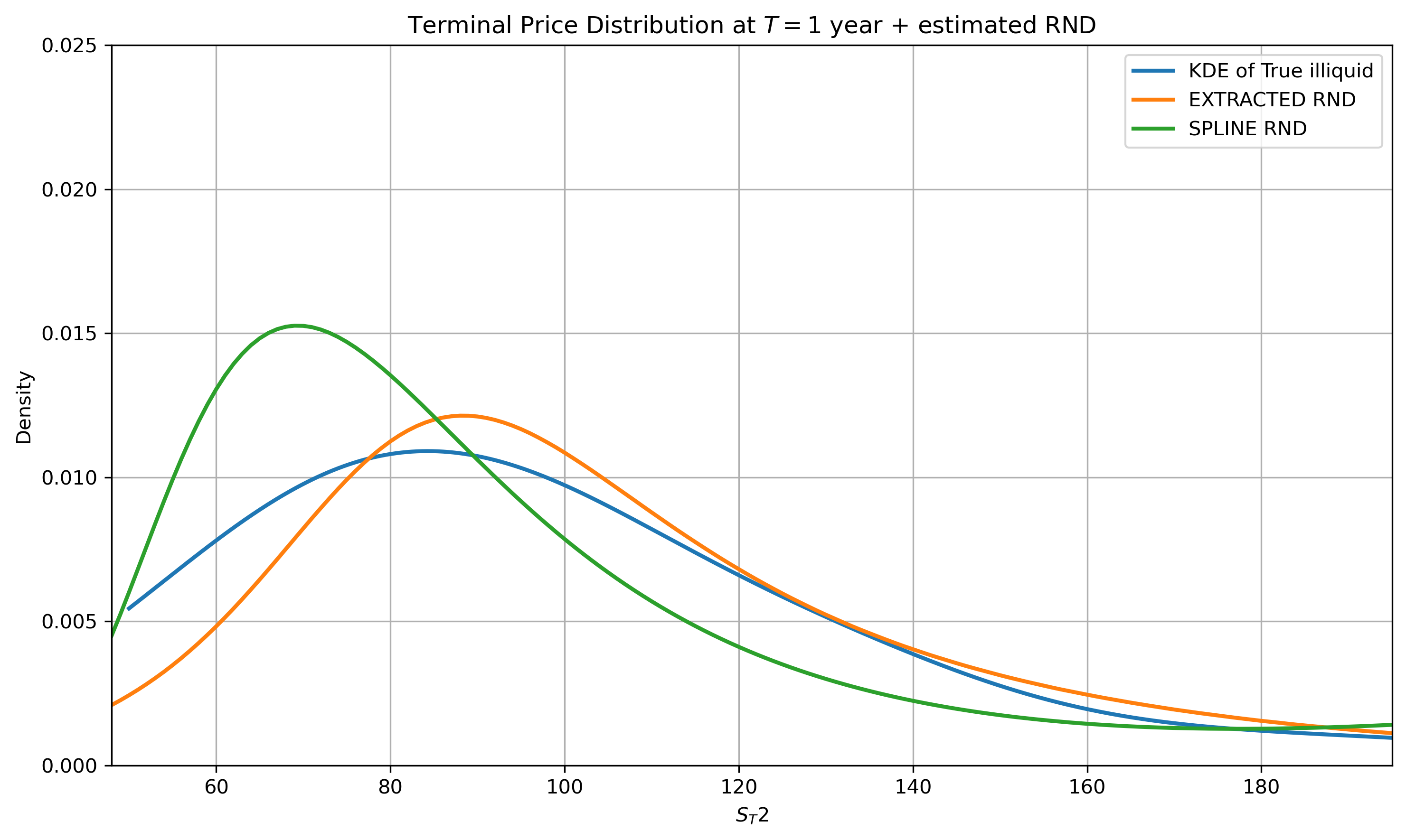} 
    \end{subfigure}
    \caption{Scenario 2 - Illiquid RND recovery of Deep-LSE (orange curve) and quadratic splines (green curve) in comparison with the illiquid target ground truth simulated RND (blue curve).}
    \label{fig:abl_final2}
\end{figure}

\subsection{Three-Factor Double Exponential Stochastic Volatility Model}\label{3fde_appendix}

The three-factor double exponential stochastic volatility model of \textcite{andersen2015risk} models the forward price $F_t$ in the risk-neutral measure as follows:
$$
\begin{aligned}
\frac{d F_t}{F_{t-}} & =\sqrt{V_{1, t}} d W_{1, t}^{\mathbb{Q}}+\sqrt{V_{2, t}} d W_{2, t}^{\mathbb{Q}}+\eta \sqrt{U_t} d W_{3, t}^{\mathbb{Q}}+\int_{\mathbb{R}^2}\left(e^x-1\right) \widetilde{\mu}^{\mathbb{Q}}(d t, d x, d y) \\
d V_{1, t} & =\kappa_1\left(\bar{v}_1-V_{1, t}\right) d t+\sigma_1 \sqrt{V_{1, t}} d B_{1, t}^{\mathbb{Q}}+\mu_v \int_{\mathbb{R}^2} x^2 1_{\{x<0\}} \mu(d t, d x, d y) \\
d V_{2, t} & =\kappa_2\left(\bar{v}_2-V_{2, t}\right) d t+\sigma_2 \sqrt{V_{2, t}} d B_{2, t}^{\mathbb{Q}} \\
d U_t & =-\kappa_u U_t d t+\mu_u \int_{\mathbb{R}^2}\left[\left(1-\rho_u\right) x^2 1_{\{x<0\}}+\rho_u y^2\right] \mu(d t, d x, d y)
\end{aligned}
$$
where $\left(W_{1, t}^{\mathbb{Q}}, W_{2, t}^{\mathbb{Q}}, W_{3, t}^{\mathbb{Q}}, B_{1, t}^{\mathbb{Q}}, B_{2, t}^{\mathbb{Q}}\right)$ is a Brownian motion in five dimension with $\operatorname{corr}\left(W_{1, t}^{\mathbb{Q}}, B_{1, t}^{\mathbb{Q}}\right)=\rho_1$, $\operatorname{corr}\left(W_{2, t}^{\mathbb{Q}}, B_{2, t}^{\mathbb{Q}}\right)=\rho_2$, while the other Brownian motions are independent.

Jumps in the forward price $F$ and in the state vector $(V_1,V_2,U)$ are modeled by an
integer-valued counting measure $\mu$. Under the risk-neutral measure $\mathbb{Q}$,
the jump intensity is
\[
dt \otimes v_t^{\mathbb{Q}}(dx,dy),
\]
and the martingale jump measure is
\[
\widetilde{\mu}^{\mathbb{Q}}(dt,dx,dy)
= \mu(dt,dx,dy) - dt\, v_t^{\mathbb{Q}}(dx,dy).
\]
The jump structure uses two components. The variable $x$ captures co-jumps in
$F_t$, $V_{1,t}$ and $U_t$ (if $\rho_u<1$), while $y$ represents shocks
specific to $U_t$, and may also affect return volatility when $\eta>0$. The density is
\[
\frac{v_t^{\mathbb{Q}}(dx,dy)}{dx\,dy}
=
\begin{cases}
c^{-}(t)\mathbf{1}_{\{x<0\}} \lambda_- e^{-\lambda_-|x|}
+
c^{+}(t)\mathbf{1}_{\{x>0\}} \lambda_+ e^{-\lambda_+ x},
& \text{if } y=0,\\[6pt]
c^{-}(t)\lambda_- e^{-\lambda_-|y|},
& \text{if } x=0 \text{ and } y<0.
\end{cases}
\]
Thus, $x\neq 0$ (with $y=0$) implies joint price--volatility jumps, whereas $x=0$ and $y<0$ yields independent jumps in $U$. Positive jumps in $U$ are either independent of $V_1$ when $\rho_u=1$, or proportional to the jumps in $V_1$ when $\rho_u=0$. Price jumps follow a double-exponential law with tail parameters $\lambda_-$ and $\lambda_+$ for negative and positive jumps. For parsimony, the independent $U$ shocks share the same distribution as negative price jumps. Time-varying jump intensities are affine in the state
\[
\begin{aligned}
c^{-}(t) &= c_0^{-}+c_1^{-}V_{1,t,-}+c_2^{-}V_{2,t,-}+c_u^{-}U_{t,-},\\
c^{+}(t) &= c_0^{+}+c_1^{+}V_{1,t,-}+c_2^{+}V_{2,t,-}+c_u^{+}U_{t,-}.
\end{aligned}
\]
Under the three-factor double-exponential stochastic volatility model, the
spot diffusive variance of the forward return is
\[
V_t = V_{1,t}+V_{2,t}+\eta^2 U_t.
\]

In Table \ref{tab:3fde} we report the parameters we use for the simulation. To construct the target implied volatility curve, we define two different sets of parameters and obtain two distinct implied volatility curves. Regarding the first one, we assume it is the liquid proxy, while the second is the illiquid target.

\begin{table}[H]
\centering
\footnotesize
\begin{tabular}{l c c}
\hline
Parameter & Liquid Proxy & Illiquid Target \\
\hline
$T$ & 1.0 & 1.0 \\
$S_0$ & 100 & 100 \\
$r_f$ & 0.05 & 0.05 \\
\hline
$V1_0$ & 0.01 & 0.01 \\
$V2_0$ & 0.04 & 0.04 \\
$U_0$  & 0.0  & 0.0  \\
\hline
$\kappa_1$ & 10.0 & 10.0 \\
$\bar v_1$ & 0.01 & 0.01 \\
$\sigma_1$ & 0.4  & 0.4  \\
$\rho_1$   & -0.9 & -0.9 \\
\hline
$\kappa_2$ & 0.2  & 0.2  \\
$\bar v_2$ & 0.04 & 0.03 \\
$\sigma_2$ & 0.12 & 0.06 \\
$\rho_2$   & -0.8 & -0.6 \\
\hline
$\kappa_u$ & 0.6 & 0.6 \\
$\eta$     & 0.0 & 0.0 \\
\hline
$\mu_v$  & 0.7   & 0.7   \\
$\mu_u$  & 10.0  & 10.0  \\
$\rho_u$ & 0.001 & 0.001 \\
\hline
$c_-$ & $(0.0,\,6.0,\,0.22,\,10.0)$ & $(0.0,\,1.0,\,0.1,\,7.0)$ \\
$c_+$ & $(0.3,\,20.0,\,18.0,\,0.0)$ & $(0.05,\,15.0,\,18.0,\,0.0)$ \\
\hline
$\lambda_-$ & 8.0  & 10.0 \\
$\lambda_+$ & 6.0  & 5.7  \\
\hline
\end{tabular}
\caption{Simulated parameters for Three Factor Double Exponential model. Set 1 for the liquid proxy and set 2 for the illiquid target.}
\label{tab:3fde}
\end{table}

We study two situations of severe market illiquidity by randomly selecting three in-the-money call option quotes for Scenario 1 and three out-of-the-money call option quotes for Scenario 2. We emphasize that these three option quotes constitute the only information on the terminal RND available to the models. In Fig. \ref{fig:3fde_setup}, we represent the first case under consideration (Scenario 1), which consists of three illiquid observations of in-the-money (ITM) call options.

\begin{figure}[H]
\centering
\begin{minipage}{\textwidth}
\begin{subfigure}{0.55\textwidth}
  \includegraphics[width=\linewidth]{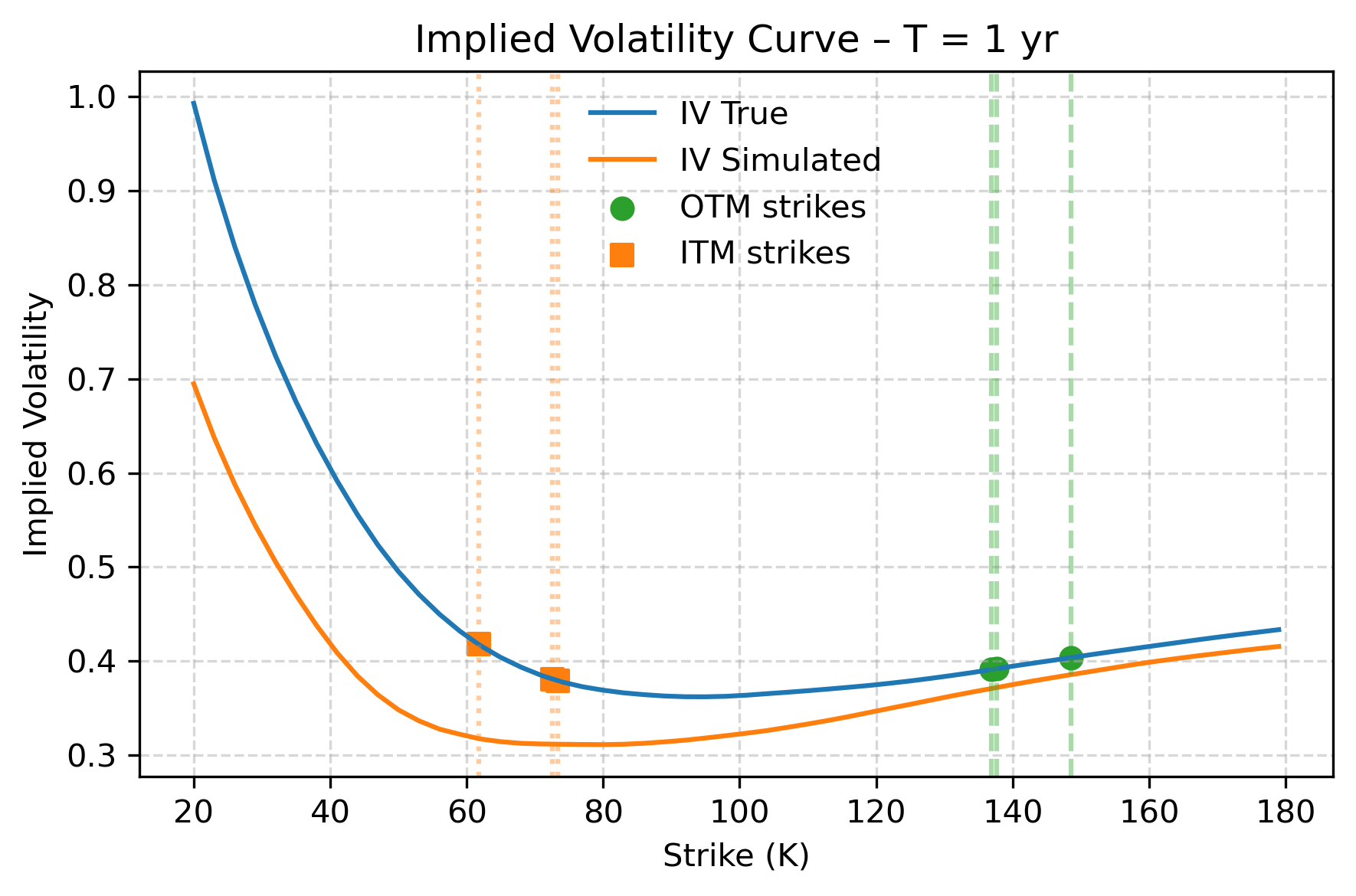}
\end{subfigure}\hfill
\begin{subfigure}{0.55\textwidth}
  \includegraphics[width=\linewidth]{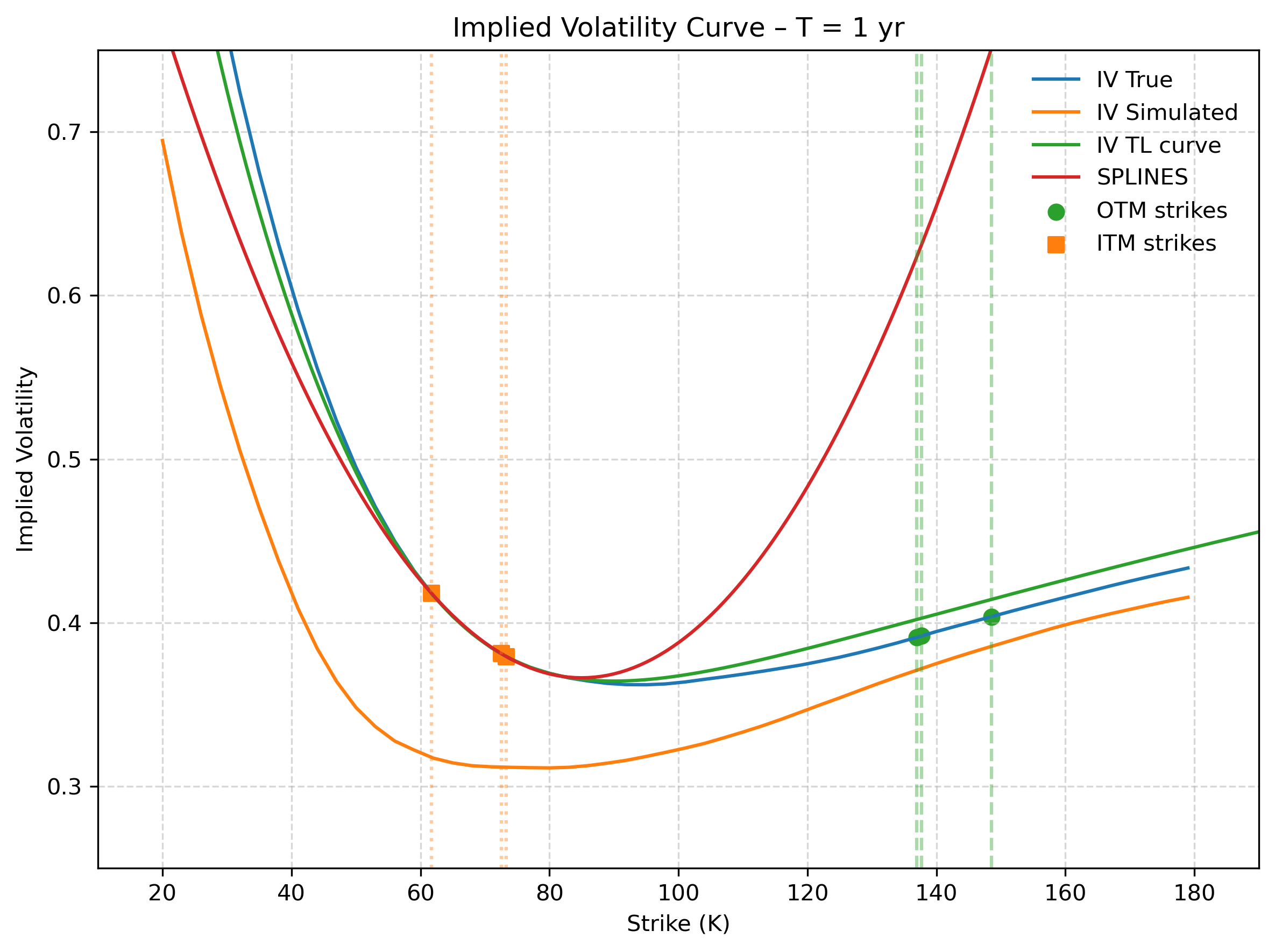}
\end{subfigure}
\caption{Scenario 1 - On the left panel, the orange dots are the ITM strikes selected on the target illiquid implied volatility curve (blue solid line), and the orange solid line is the liquid source (proxy) implied volatility curve. On the right panel, the interpolation of the illiquid implied volatility curve of the Deep-LSE model (green solid line) and quadratic splines (red solid line).}
\label{fig:3fde_setup}
\end{minipage}

\end{figure}

We observe in Fig. \ref{fig:3fde_final} the estimates of Deep-LSE and quadratic splines compared to the ground truth illiquid RND. The Deep-LSE accurately recovers the illiquid RND, while the quadratic spline recovers an unstable estimate of the illiquid RND.

\begin{figure}[H]
    \centering
    \begin{subfigure}{0.75\textwidth}
        \includegraphics[width=\linewidth]{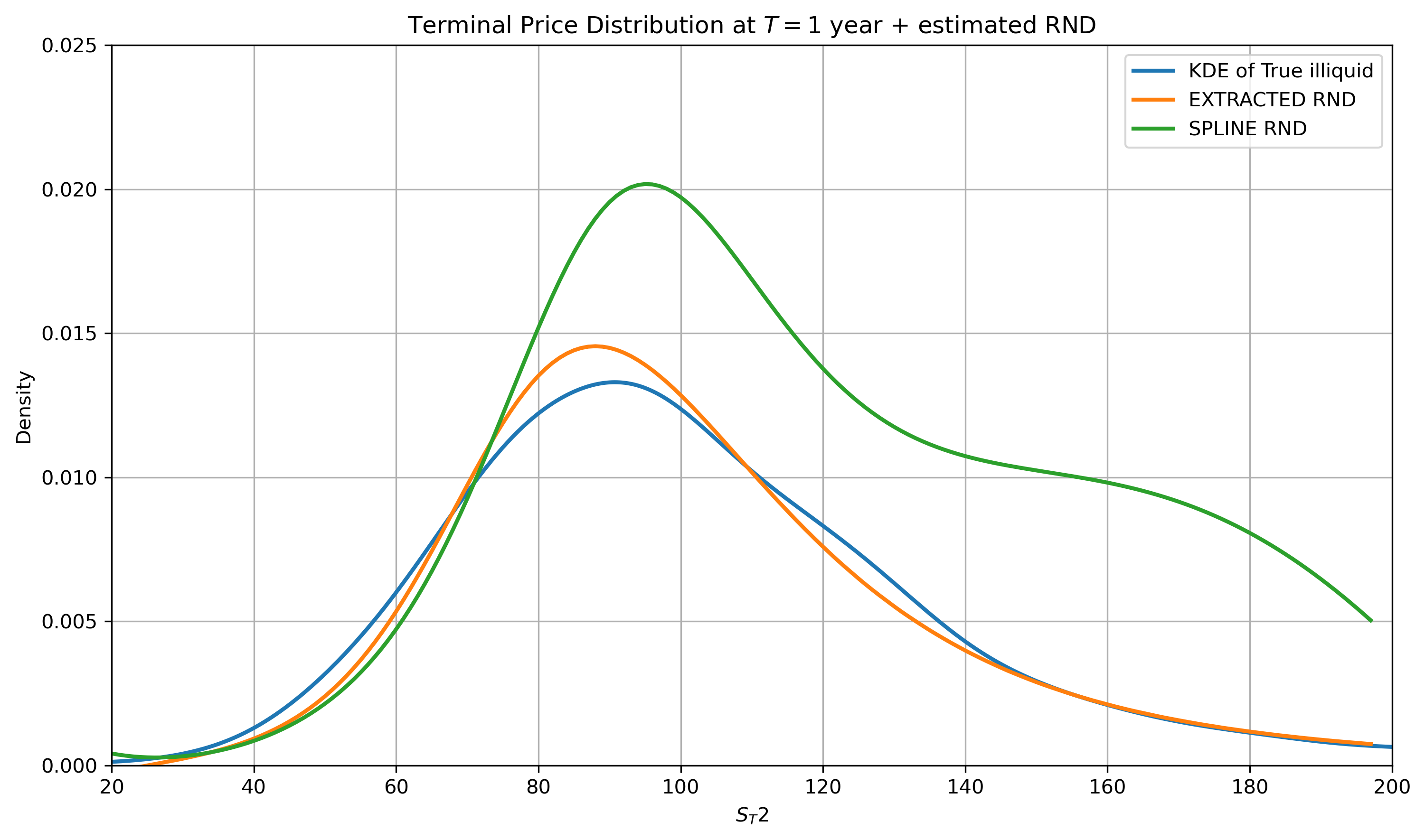} 
    \end{subfigure}
    \caption{Scenario 1 - Illiquid RND recovery of Deep-LSE (orange curve) and quadratic splines (green curve) in comparison with the illiquid target ground truth simulated RND (blue curve).}
    \label{fig:3fde_final}
\end{figure}


We also test the Deep-LSE model on out-of-the-money call options and illiquid strikes (Scenario 2), and Fig. \ref{fig:3fde_setup2} illustrates this case. We illustrate in Fig. \ref{fig:3fde_final2} the estimates of the Deep-LSE and quadratic splines versus the ground truth illiquid RND. We conclude that, also in this case, the Deep-LSE recovers the RND better than quadratic splines.

\begin{figure}[H]
\centering
\begin{minipage}{\textwidth}
\begin{subfigure}{0.55\textwidth}
  \includegraphics[width=\linewidth]{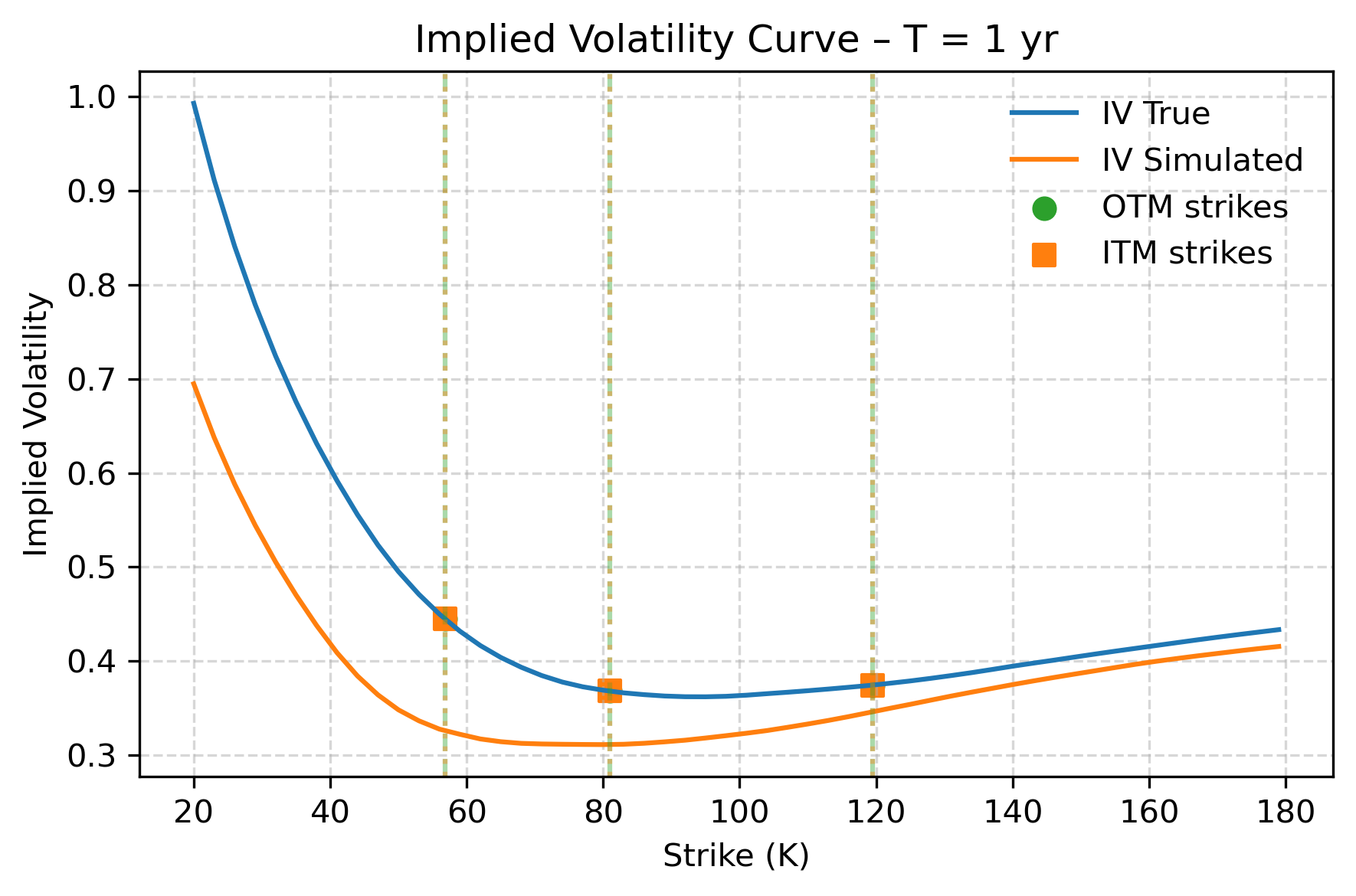}
\end{subfigure}\hfill
\begin{subfigure}{0.55\textwidth}
  \includegraphics[width=\linewidth]{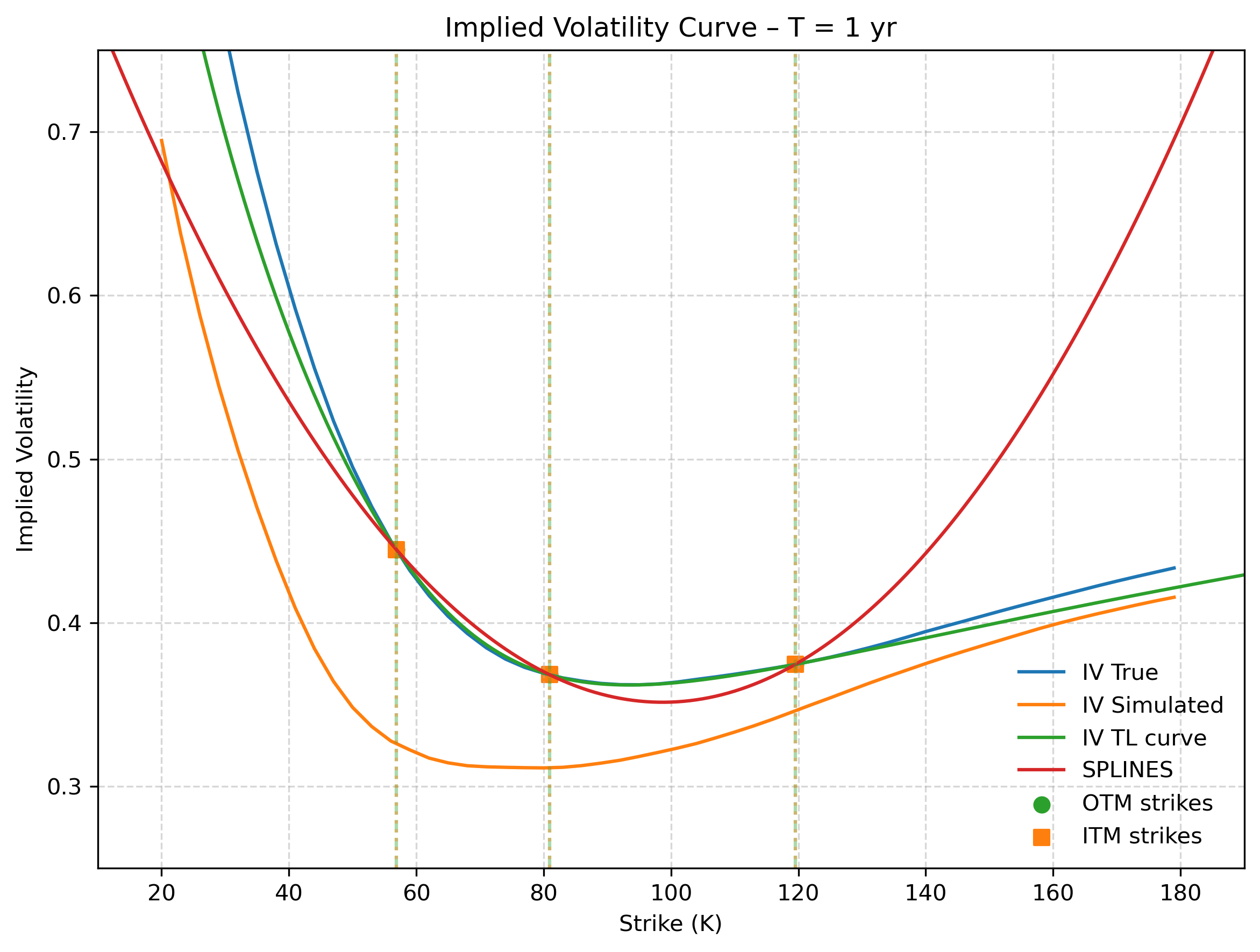}
\end{subfigure}
\caption{Scenario 2 - On the left panel, the orange dots are the OTM strikes selected on the target illiquid implied volatility curve (blue solid line), and the orange solid line is the liquid source (proxy) implied volatility curve. On the right panel, the interpolation of the illiquid implied volatility curve of the Deep-LSE model (green solid line) and quadratic splines (red solid line).}
\label{fig:3fde_setup2}
\end{minipage}

\end{figure}

\begin{figure}[H]
    \centering
    \begin{subfigure}{0.75\textwidth}
        \includegraphics[width=\linewidth]{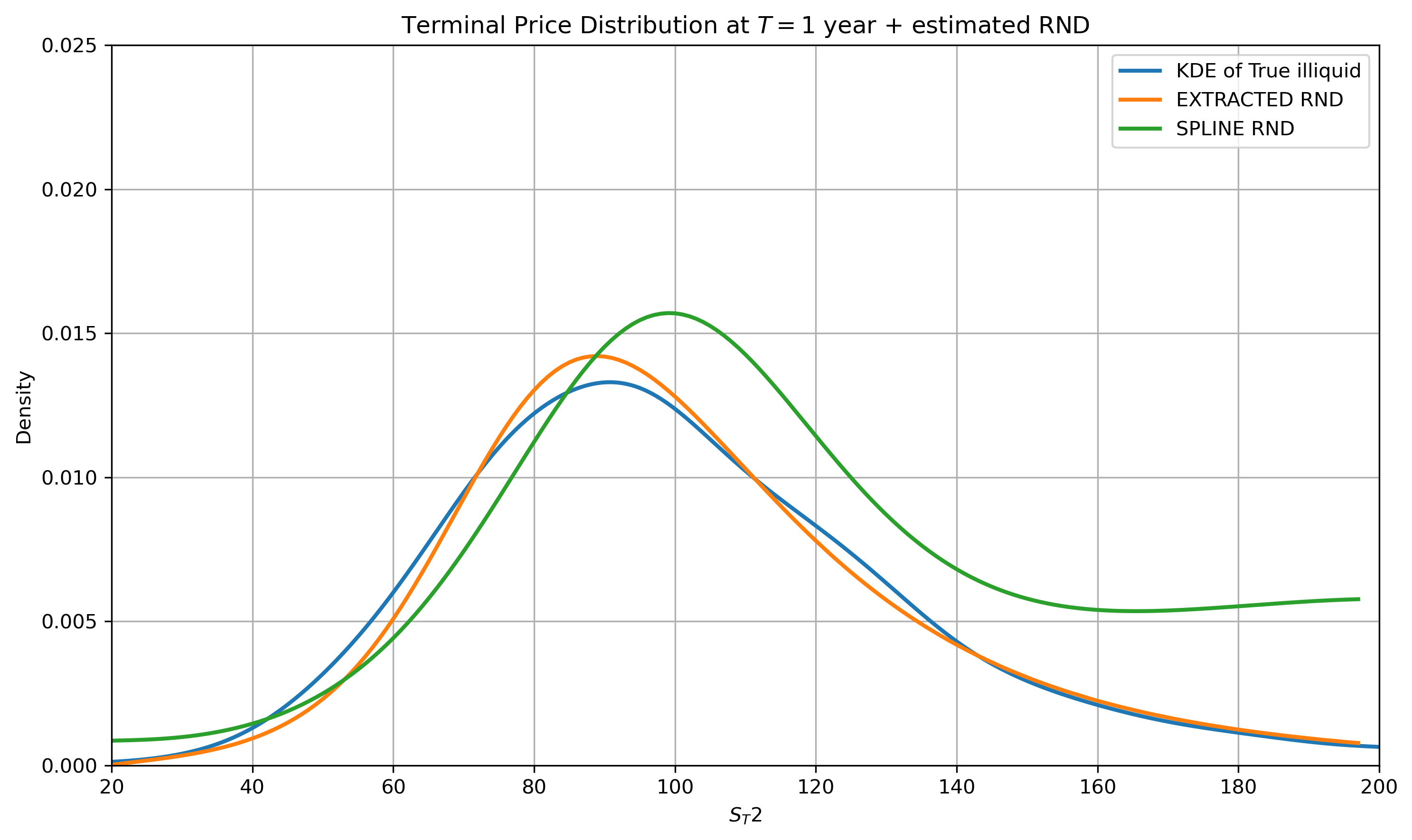} 
    \end{subfigure}
    \caption{Case 2 - Illiquid RND recovery of Deep-LSE (orange curve) and quadratic splines (green curve) in comparison with the illiquid target ground truth simulated RND (blue curve).}
    \label{fig:3fde_final2}
\end{figure}

\newpage

\section{Additional Empirical Analysis}\label{empa_appendix}

Regarding Scenario 1 of the empirical analysis on the SPX, Fig. \ref{fig:train_source_empa} and Fig. \ref{fig:train_target_empa} depict the learning process of the Deep-LSE from the source (proxy) data and from the illiquid target data respectively. Fig. \ref{fig:train_source_empa} is the first step of the estimation procedure. The model receives as input the liquid proxy data (blue points), which consists of the 2015 SPX implied volatility recovered from market data. The data points of the implied volatility curve are heavily convex in the strike range of $2100-2200$, making the learning process more challenging. In addition, the scale of the data (extremely small) and the right tail, which consists only of 4 observations, are additional challenges. Nonetheless, at iteration number 100, the approximation of the implied volatility curve is well calibrated.

\begin{figure}[]
\centering
\begin{minipage}{\textwidth}
\begin{subfigure}{0.45\textwidth}
  \includegraphics[width=\linewidth]{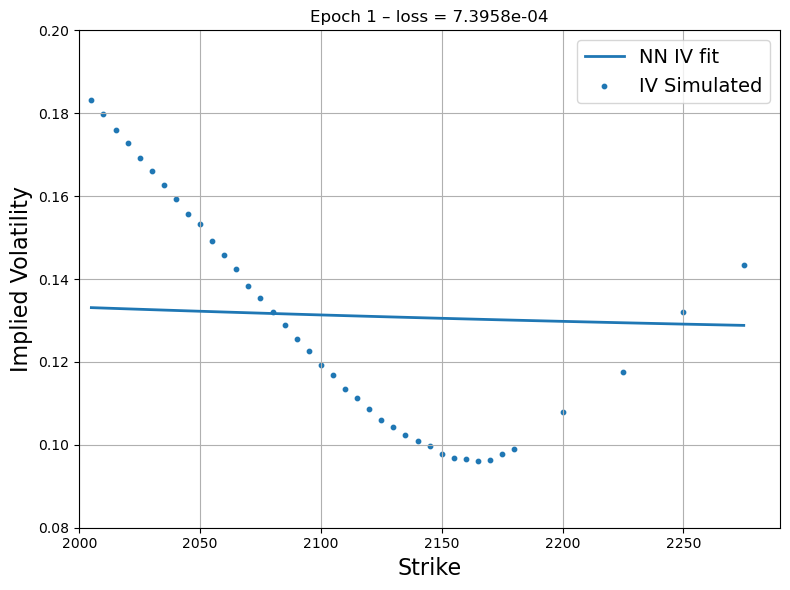}
\end{subfigure}\hspace{1cm}
\begin{subfigure}{0.45\textwidth}
  \includegraphics[width=\linewidth]{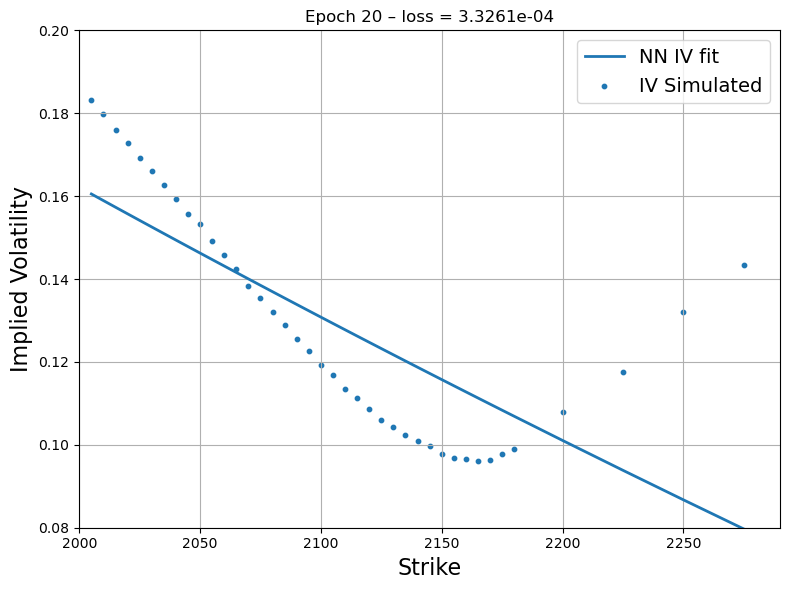}
\end{subfigure}

\medskip

\begin{subfigure}{0.45\textwidth}
  \includegraphics[width=\linewidth]{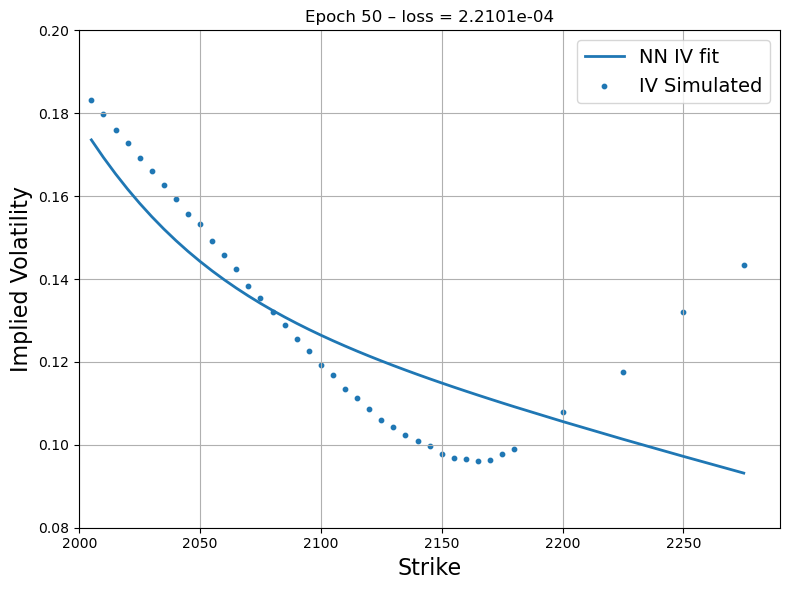}
\end{subfigure}\hspace{1cm}
\begin{subfigure}{0.45\textwidth}
  \includegraphics[width=\linewidth]{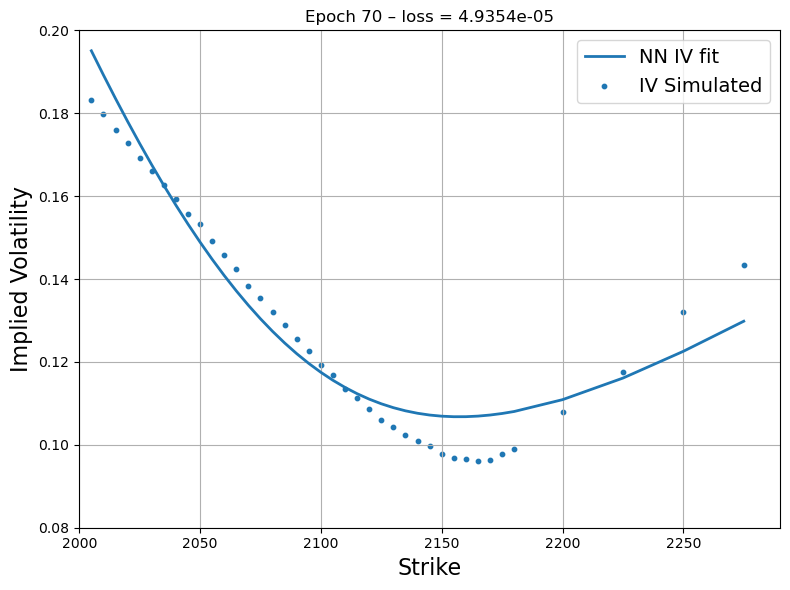}
\end{subfigure}

\medskip

\begin{subfigure}{0.45\textwidth}
  \includegraphics[width=\linewidth]{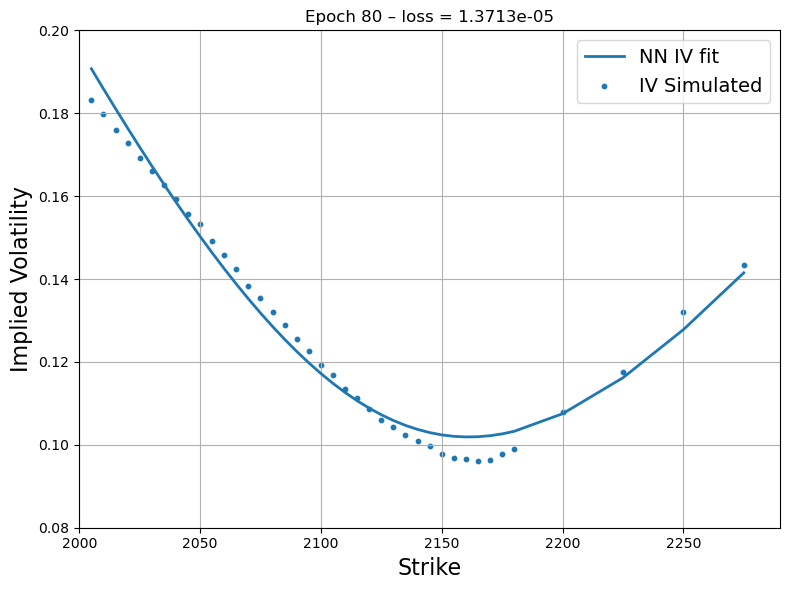}
\end{subfigure}\hspace{1cm}
\begin{subfigure}{0.45\textwidth}
  \includegraphics[width=\linewidth]{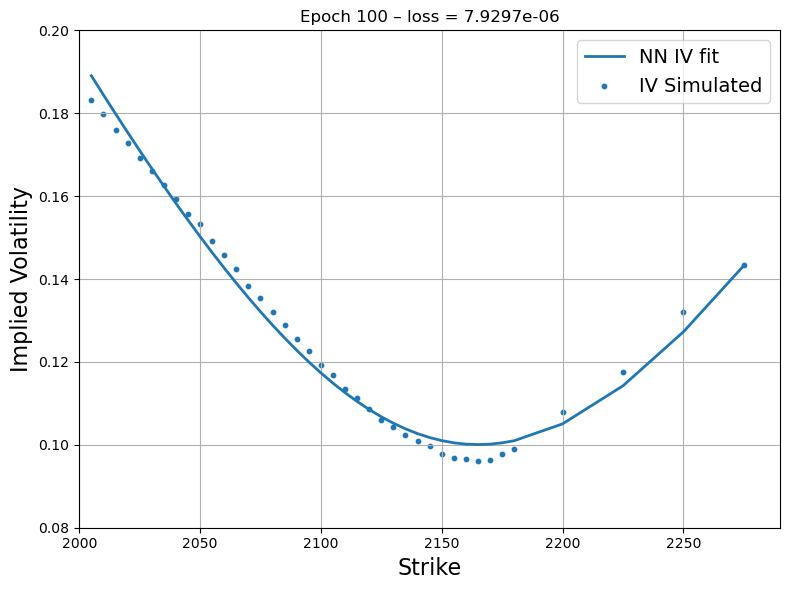}
\end{subfigure}

\caption{First step recovery (Scenario 1) - Source Deep-LSE Fit. The blue dots represent the implied volatility curve of option quotes of the liquid proxy asset while the blue solid line represents the fit of the interpolating function of the Deep-LSE model.}
\label{fig:train_source_empa}
\end{minipage}
\end{figure}

Fig. \ref{fig:train_target_empa} represents the second step of the estimation process to recover the illiquid RND. In this phase, we perform transfer learning, fine-tuning the model pre-trained on the liquid proxy, using the illiquid market data. The orange data points consist are the implied volatility of the liquid proxy (2015 SPX) data, and the orange curve is the implied volatility curve that the Deep-LSE estimates during the first phase.

\begin{figure}[]
\centering
\begin{minipage}{\textwidth}
\begin{subfigure}{0.45\textwidth}
  \includegraphics[width=\linewidth]{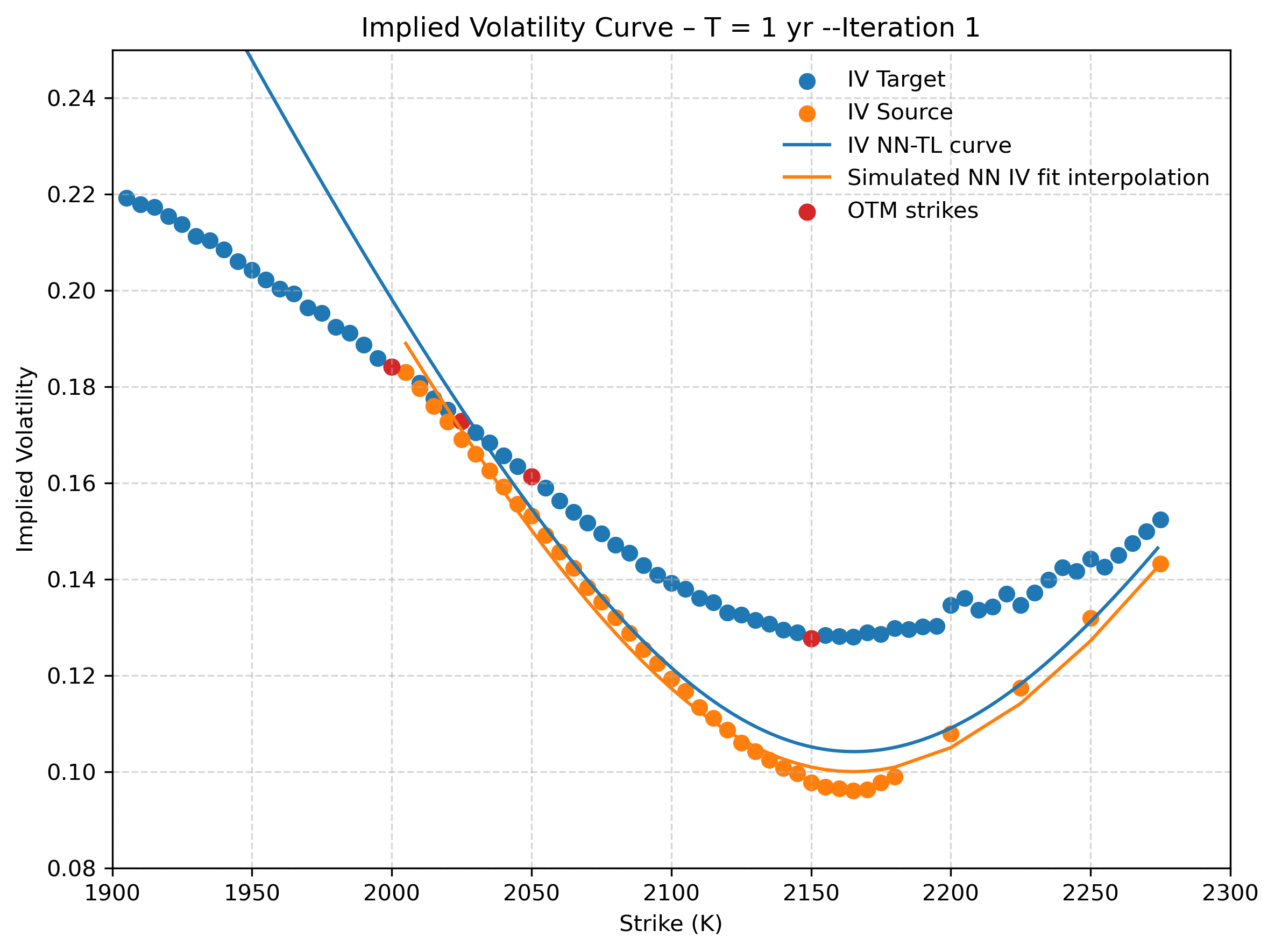}
\end{subfigure}\hspace{1cm}
\begin{subfigure}{0.45\textwidth}
  \includegraphics[width=\linewidth]{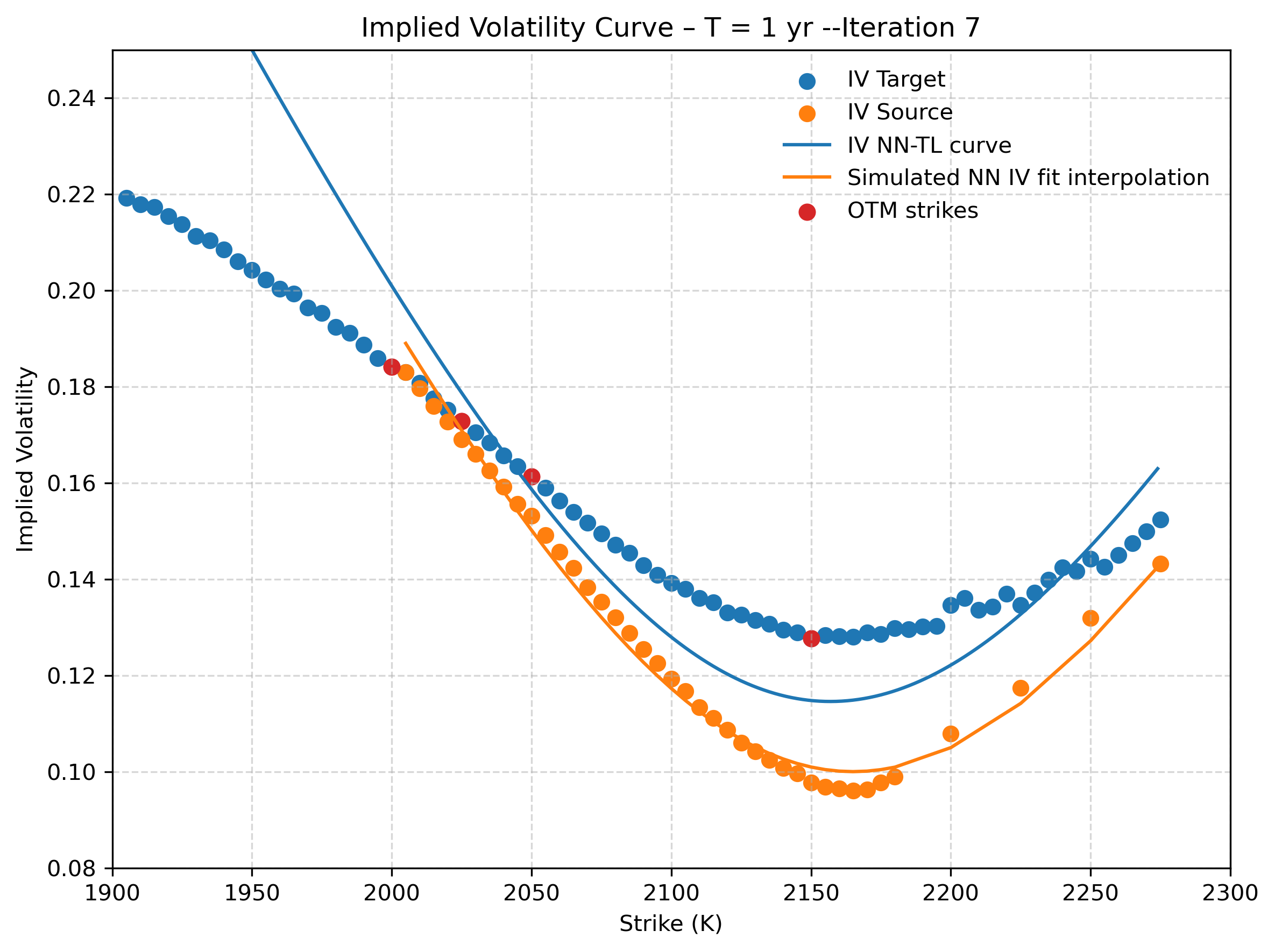}
\end{subfigure}

\medskip

\begin{subfigure}{0.45\textwidth}
  \includegraphics[width=\linewidth]{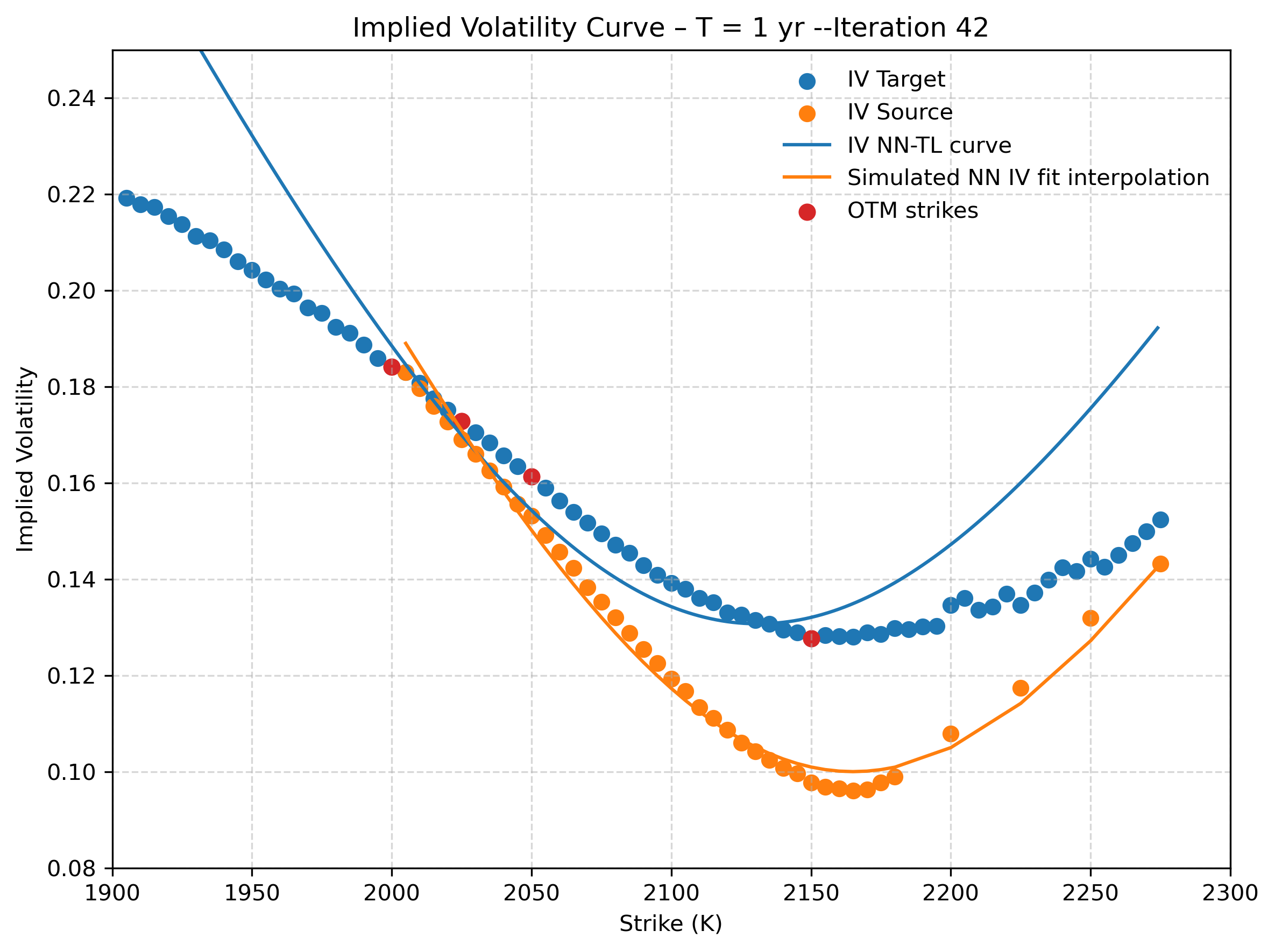}
\end{subfigure}\hspace{1cm}
\begin{subfigure}{0.45\textwidth}
  \includegraphics[width=\linewidth]{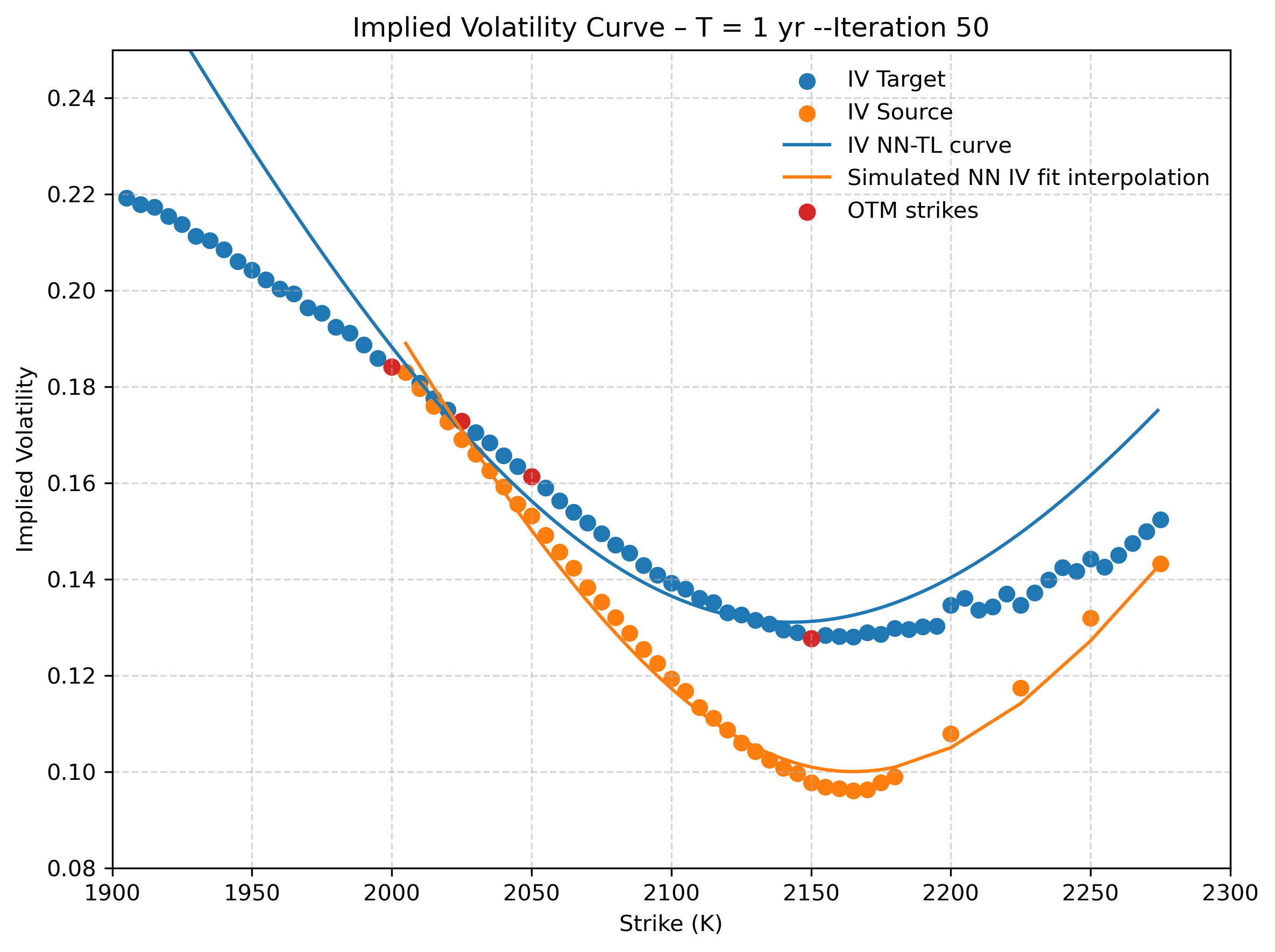}
\end{subfigure}

\medskip

\begin{subfigure}{0.45\textwidth}
  \includegraphics[width=\linewidth]{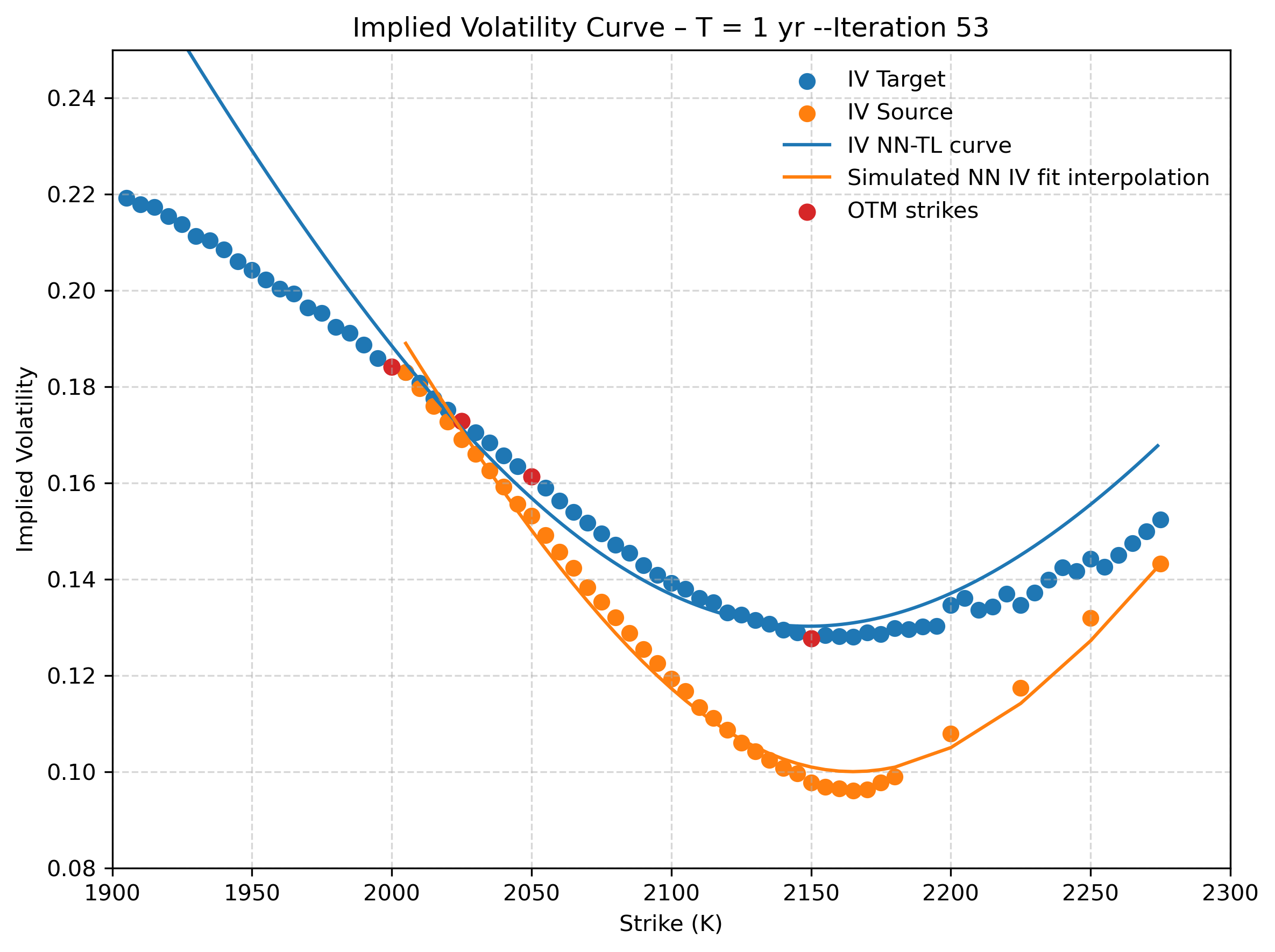}
\end{subfigure}\hspace{1cm}
\begin{subfigure}{0.45\textwidth}
  \includegraphics[width=\linewidth]{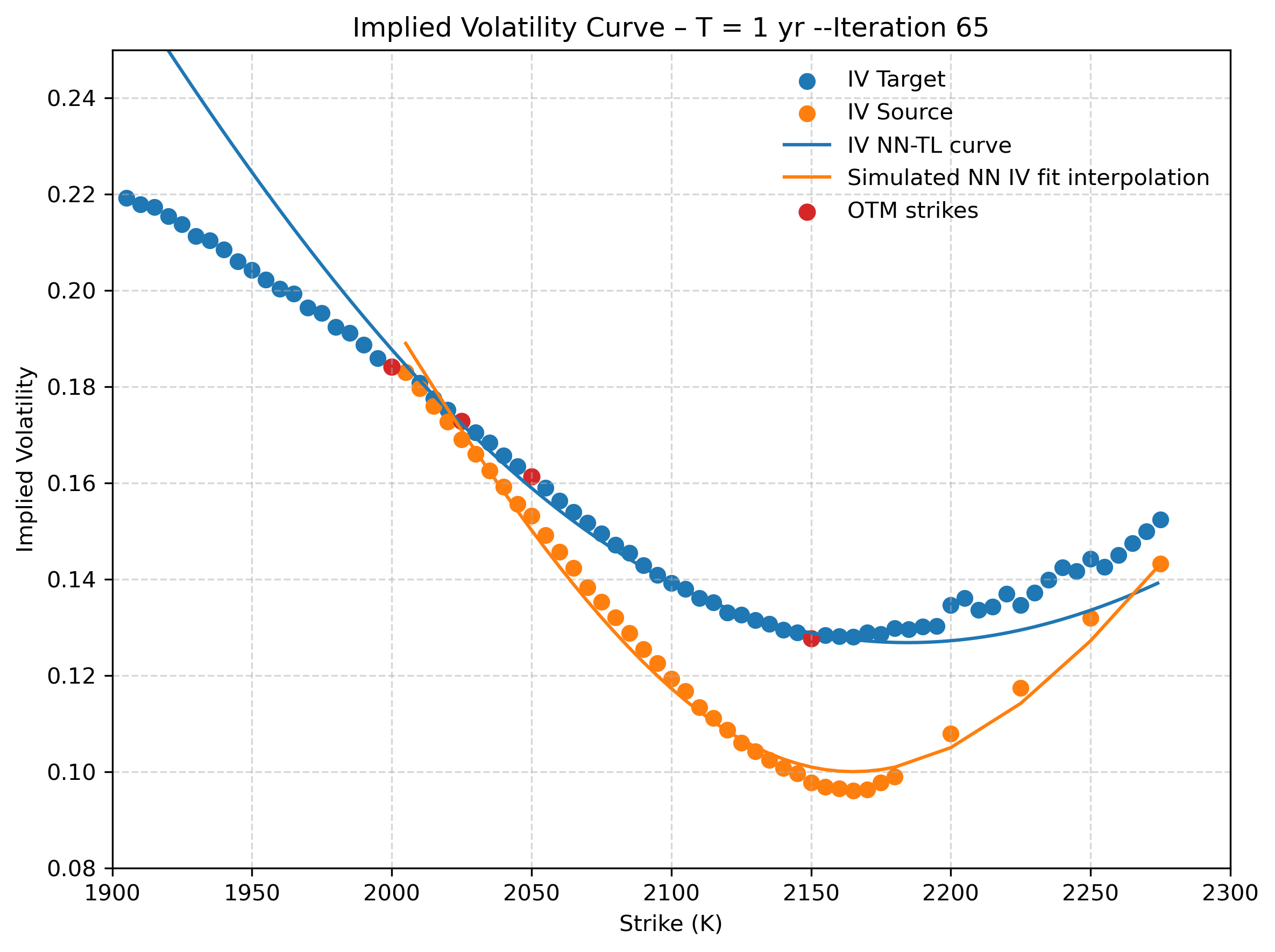}
\end{subfigure}

\caption{Second step recovery (Scenario 1) - Target Deep-LSE Fit. The model only sees the illiquid (red) quotes. The blue dots are the true implied volatility quotes that the Deep-LSE recovers (blue solid line). The solid orange and red curves represent the estimated IV curve of the first step.}
\label{fig:train_target_empa}
\end{minipage}
\end{figure}

Now, the Deep-LSE model receives as input the illiquid option quotes (the three red points) from the observed SPX 2016 data. The three red points mimic the sparse quotes of an illiquid market, whereas in reality, one observes the full SPX implied volatility curve for 2016 (blue points). The blue points, therefore, represent the underlying true implied volatility surface, which we deliberately sparsify to emulate illiquid conditions. We observe how the model adapts the knowledge learned during the first phase to the target data points, effectively recovering the illiquid implied volatility curve and producing the blue solid curve.


\end{document}